\documentclass[reprint,amsfonts, amssymb, amsmath,  showkeys,pra, superscriptaddress, twocolumn,longbibliography,prl]{revtex4-2}
\usepackage{float}
\makeatletter
\let\newfloat\newfloat@ltx
\makeatother
\usepackage[english]{babel}
\usepackage[utf8]{inputenc}
\usepackage{graphics}
\usepackage{selinput}
\usepackage[normalem]{ulem}
\usepackage[shortlabels]{enumitem}

\usepackage{braket}
\usepackage{amsthm}
\usepackage{mathtools}
\usepackage{physics}
\usepackage{xcolor}
\usepackage{graphicx}
\usepackage[left=16mm,right=16mm,top=35mm,columnsep=15pt]{geometry} 
\usepackage{adjustbox}
\usepackage{placeins}
\usepackage[T1]{fontenc}
\usepackage{lipsum}
\usepackage{csquotes}
\usepackage{bm}

\usepackage[linesnumbered,ruled,vlined]{algorithm2e}
\SetKwInput{kwInit}{Init}

\def\HC{\mathcal{H}}

\def\ad{^{\dagger}}

\newcommand{\fsnull}[1]{}
\newcommand{\old}[1]{}

\usepackage[makeroom]{cancel}
\usepackage[toc,page]{appendix}
\usepackage[colorlinks=true,citecolor=blue,linkcolor=magenta]{hyperref}

\usepackage{tikz}
\tikzset{every picture/.style=remember picture}

\usepackage[utf8]{inputenc}
\usepackage{graphicx}
\usepackage{xcolor}
\usepackage{amsmath}
\usepackage{amsthm}
\usepackage{bm}
\usepackage{bbm}
\usepackage{comment}
\usepackage{appendix}
\usepackage{mathdots}
\usepackage{lipsum}
\usepackage{verbatim}
\usepackage{natbib}
\usepackage{nccmath}
\usepackage{amsfonts} 
\newcommand{\poly}{\operatorname{poly}}

\newcommand{\BC}{\mathcal{B}}
\newcommand{\CC}{\mathcal{C}}

\newcommand{\EC}{\mathcal{E}}
\newcommand{\FC}{\mathcal{F}}
\newcommand{\GC}{\mathcal{G}}
\newcommand{\MC}{\mathcal{M}}
\newcommand{\NC}{\mathcal{N}}
\newcommand{\OC}{\mathcal{O}}
\newcommand{\PC}{\mathcal{P}}

\newcommand{\Var}{{\rm Var}}

\renewcommand{\leq}{\leqslant}

\renewcommand{\vec}[1]{\boldsymbol{#1}}  % Bold vectors instead of arrow vectors

\newcommand*{\id}{\openone}

\newcommand*{\g}{\mathfrak{g}}

\newcommand{\bs}{\textsf{BS}}

\newcommand{\thv}{\vec{\theta}}

\def\be{\begin{equation}}
\def\ee{\end{equation}}
\def\bs{\begin{split}}
\def\e{\end{split}}
\def\ba{\begin{eqnarray}}
\def\bea{\begin{eqnarray}}

\def\tea{\end{eqnarray}}
\def\ea{\end{eqnarray}}
\def\eea{\end{eqnarray}}

\newcommand\mf[1]{\mathfrak{#1}}
\newcommand\mbb[1]{\mathbb{#1}}

\def\liea{\mathfrak{g}}

\newcommand\Z{\text{Z}}

\newtheorem{theorem}{Theorem}
\newtheorem{lemma}{Lemma}
\newtheorem{suplemma}{Supplemental Lemma}
\newtheorem{supdefinition}{Supplemental Definition}
\newtheorem{suptheorem}{Supplemental Theorem}
\newtheorem{supproposition}{Supplemental Proposition}
\newtheorem{supexample}{Supplemental Example}
\newtheorem{corollary}{Corollary}

\usepackage{amssymb}
\usepackage{dsfont}

\def\be{\begin{equation}}
\def\te{\end{equation}}
\def\ee{\end{equation}}
\def\ba{\begin{eqnarray}}
\def\bea{\begin{eqnarray}}

\def\tea{\end{eqnarray}}
\def\ea{\end{eqnarray}}
\def\eea{\end{eqnarray}}

\begin{document}

\title{Showcasing a Barren Plateau Theory Beyond the Dynamical Lie Algebra}

\author{N. L. Diaz}
\affiliation{Information Sciences, Los Alamos National Laboratory, Los Alamos, NM 87545, USA}
\affiliation{Departamento de F\'isica-IFLP/CONICET,
		Universidad Nacional de La Plata, C.C. 67, La Plata 1900, Argentina}

  \author{Diego Garc\'ia-Mart\'in}
\affiliation{Information Sciences, Los Alamos National Laboratory, Los Alamos, NM 87545, USA}

\author{Sujay Kazi}
\affiliation{Department of Electrical and Computer Engineering, Duke University, Durham, NC 27708, USA}

\author{Martin Larocca}
\affiliation{Theoretical Division, Los Alamos National Laboratory, Los Alamos, NM 87545, USA}
\affiliation{Center for Nonlinear Studies, Los Alamos National Laboratory, Los Alamos, NM 87545, USA}

\author{M. Cerezo}
\affiliation{Information Sciences, Los Alamos National Laboratory, Los Alamos, NM 87545, USA}

\begin{abstract}

Barren plateaus have emerged as a pivotal challenge for variational quantum computing. Our  understanding of this phenomenon underwent a transformative shift with the recent introduction of a Lie algebraic theory capable of  explaining most sources of barren plateaus.  However, this theory requires either initial states or observables that lie in the circuit's Lie algebra.
Focusing on parametrized matchgate circuits, in this work we are able to go beyond this assumption and provide an exact formula for the loss function variance  that is valid for arbitrary input states and measurements. 
Our results reveal that new phenomena emerge when the Lie algebra constraint is relaxed. For instance,  we find that the variance does not necessarily vanish inversely with the Lie algebra's dimension. Instead, this measure of expressiveness is replaced by a generalized expressiveness quantity: The dimension of the Lie group modules.  By characterizing the operators in these modules as products of Majorana operators, we can introduce a precise notion of generalized globality and show that measuring generalized-global operators leads to   barren plateaus. Our work also provides operational meaning to the generalized entanglement as we connect it with known fermionic entanglement measures, and show that it  satisfies a monogamy relation. 
Finally, while parameterized matchgate circuits 
are not efficiently simulable in general, our results suggest that the  structure allowing for trainability may also lead to classical simulability.

\end{abstract}

\maketitle
\textit{Introduction.}   The development of near-term quantum hardware  raised an important question: \textit{What can we achieve with these devices?} Variational quantum computing schemes, such as variational algorithms~\cite{cerezo2020variationalreview, bharti2021noisy, endo2021hybrid}, and quantum machine learning models~\cite{schuld2015introduction, biamonte2017quantum, havlivcek2019supervised}, stand as promising candidates to answer this question. 
At its core, variational quantum computing attempts to solve a  problem by formulating it as an optimization task. These methods operate by evolving an initial state through a parametrized circuit and estimating the expectation value of a given observable at the circuit's output. 
The model's parameters are iteratively adjusted to minimize a loss function quantifying the extent to which the problem has been solved.

Ever since the seminal work in Ref.~\cite{mcclean2018barren}, which first showed that variational quantum computing schemes can exhibit Barren Plateaus (BPs) 
--an exponentially concentrated loss function--
the study of BPs underwent a veritable Cambrian explosion.  Over the past few years, researchers
have started to unravel
the different sources of BPs through case-specific analyses ~\cite{cerezo2020cost,marrero2020entanglement,sharma2020trainability,patti2020entanglement,pesah2020absence,uvarov2020barren,cerezo2020impact,uvarov2020variational,wang2020noise,abbas2020power,arrasmith2021equivalence,larocca2021diagnosing,liu2021presence,holmes2021connecting,zhao2021analyzing,kieferova2021quantum,thanaslip2021subtleties,lee2021towards,shaydulin2021importance,holmes2021barren,leadbeater2021f,martin2022barren,grimsley2022adapt,leone2022practical,sack2022avoiding,kashif2023impact,friedrich2023quantum,garcia2023deep,kulshrestha2022beinit,volkoff2021efficient,kashif2023unified,monbroussou2023trainability}. Recently,  two independent works~\cite{fontana2023theadjoint,ragone2023unified} presented the mathematical foundations for a unified Lie algebraic theory of BPs.  As conceptually explained in~\cite{ragone2023unified} (and as conjectured in~\cite{larocca2021diagnosing}), a central object to BPs is the Dynamical Lie Algebra (DLA) $\g$ which quantifies the ultimate expressiveness of a parametrized circuit.  
While Refs.~\cite{fontana2023theadjoint,ragone2023unified} constituted a quantum leap in our understanding of BPs, the theory developed requires either the initial state or measurement operator to belong in the DLA.

In this work, we present an exact loss variance formula 
for parametrized matchgate circuits~\cite{jozsa2008matchgates,wan2022matchgate,de2013power,oszmaniec2022fermion,cherrat2023quantum} that does not require any assumptions on input states or measurements. By leveraging the decomposition of operator space into group modules, we show that the variance decomposes as a sum of terms that quantify  generalized purities and coherences in the state and measurement (see Fig.~\ref{fig:schematic}) with respect to the different modules.  Importantly, not all variance contributions vanish inverse-polynomially in the DLA's dimension, meaning that the conjecture in~\cite{larocca2021diagnosing} is not valid in this extended setting. Going beyond, we are able to prove that the generalized entanglement appearing in our extended formula for matchgate circuits constitutes fermionic entanglement 
as quantified by \cite{gigena2015entanglement, gigena2020one, gigena2021many}. In particular, this implies a generalized monogamy of entanglement relation. We then present an interpretation of the modules as containing operators with different amounts of generalized globality, and find that generalized-global observables lead to BPs.  We highlight that our formulas contain, as special cases, those obtained in Refs.~\cite{fontana2023theadjoint,ragone2023unified}. However,  by studying arbitrary input states and measurements we are able to reach regimes and insights beyond previous approaches.

\textit{Background}.
We focus on a loss function of the form
\begin{equation}\label{eq:loss}
    \ell_{\thv}(\rho,O)=\Tr[U(\thv)\rho U\ad(\thv) O]\,.
\end{equation}
Here $\rho$ is an $n$-qubit state in $\HC=(\mbb{C}^2)^{\otimes n}$, $O$ an observable (with $\norm{O}_2^2\leq 2^n)$,  and $U(\thv)=\prod_le^{-i\theta_l H_l}$ a parametrized circuit with trainable parameters $\thv$ and Hermitian generators $H_l$ taken from a set $\GC$. The presence of BPs  can be diagnosed by studying the variance 
\begin{equation}\label{eq:var}
    \Var_{\thv}[\ell_{\thv}(\rho,O)]=\mathbb{E}_{\thv}[\ell_{\thv}(\rho,O)^2]-\mathbb{E}_{\thv}[\ell_{\thv}(\rho,O)]^2\,,
\end{equation}
and we will say that  the loss exhibits a BP if  $\Var_{\thv}[\ell_{\thv}(\rho,O)]\in\OC(1/b^n)$ for $b>1$, in which case the loss is exponentially concentrated around its mean~\cite{mcclean2018barren,cerezo2020cost}.

Throughout this paper we consider $U(\thv)$ to be a parametrized matchgate circuit composed of fermionic Gaussian unitaries. 
In particular, we take the circuit generators $\GC=\{Z_i\}_{i=1}^n\cup\{X_iX_{i+1}\}_{i=1}^{n-1}$ producing a DLA $\g={\rm span}_\mathbb{R}\langle i\GC\rangle_{{\rm Lie}}$ as follows~\cite{kokcu2022fixed,wiersema2023classification}
\small
 \begin{align}
     \g&={\rm span}_{\mathbb R}i\{Z_i,\widehat{X_iX_j},\widehat{Y_iY_j},\widehat{X_iY_j},\widehat{Y_iX_j}\}_{1\leq i<j\leq n} \simeq \mathfrak{so}(2n)\nonumber\,,
 \end{align}
 \normalsize
 where  $X_i$, $Y_i$ and $Z_i$ denoting the Pauli operators acting on the $i$-th qubit, and $\widehat{A_iB_j}=A_iZ_iZ_{i+1}\cdots Z_{j-1}B_j$. Moreover, we assume that $U(\thv)$ is deep enough so that it forms a $2$-design over the Lie group $G=e^{\g}$ (see Theorem 2 in Ref.~\cite{ragone2023unified}). This allows us to compute expectation values over the parameter landscape as $\mathbb{E}_{\thv}[U(\thv)^{\otimes t}(\cdot) U\ad(\thv)^{\otimes t} ]=\int_{G}dU U(\thv)^{\otimes t}(\cdot) U\ad(\thv)^{\otimes t} $, for $t\leq 2$, where $dU$ denotes the normalized left-and-right-invariant Haar measure on $G$.

\begin{figure}[t]
    \centering
\includegraphics[width=1\columnwidth]{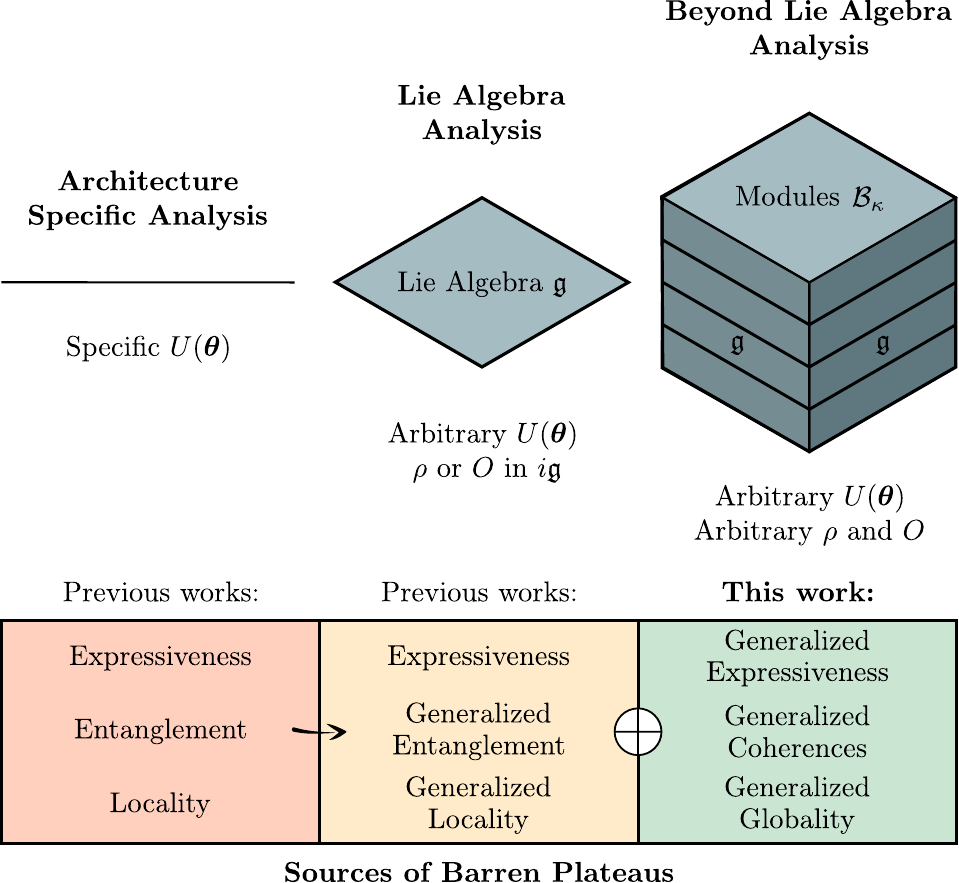}
    \caption{For the past few years researchers have studied BP in case-specific analyses, discovering different causes of loss function concentration     such as the expressiveness of $U(\thv)$, the amount of entanglement in $\rho$, or the locality of $O$. Recently, Ref.~\cite{ragone2023unified} proposed a 
    unification of these BP sources 
    under a Lie algebraic framework 
    involving generalized notions of entanglement and locality. This unified theory requires a strong assumption: either $\rho$ or $O$ must be in $i\g$.
    In this work we go beyond this assumption and develop a theory of BPs for arbitrary initial states and measurements in parametrized matchgate circuits.   
    Our results 
    uncover new sources of concentration such as the generalized module expressiveness, coherences and globality.   }
    \label{fig:schematic}
\end{figure}

 A convenient way to handle the Lie algebra $\g$ is by mapping its Pauli elements to fermionic operators via the Jordan-Wigner transformation~\cite{lieb1961two}.  In particular, we will use the (Hermitian) Majorana fermionic operators,
\begin{equation}
\begin{split}
    c_1&=XI\dots I,\; c_3= ZXI\dots I, \; c_{2n-1}=Z\dots Z X\,, \nonumber\\
        c_2&=YI\dots I,\; c_4= ZYI\dots I, \; \;\; c_{2n}\;\;\;=Z\dots Z Y\,,
\end{split}
\end{equation}
which form a Clifford algebra with quadratic form $    \{c_\mu,c_\nu\}=2\delta_{\mu\nu}$ for $\mu,\nu=1,\dots ,2n$. One can verify that $\g={\rm span}_{\mathbb R}\{c_\mu c_\nu\}_{1\leq \mu<\nu \leq 2n}$ .

\textit{Results}. In what follows we derive exact formulas for the loss function variance. Our main results are based on the following  preliminary lemmas, proved in the Supplemental Information (SI): 
\begin{lemma}\label{lem:modules}
The space of linear operators acting on $n$-qubits, denoted as $\BC$, can be decomposed into subspaces as
\begin{equation} \label{eq:modules}
\mathcal{B}=\bigoplus_{\kappa=0}^{2n} \mathcal{B}_{\kappa}\,,
\end{equation}
with each  $\mathcal{B}_\kappa$ being the linear space, of dimension $\binom{2n}{\kappa}$, spanned by a basis of products of $\kappa$ distinct Majoranas. 
\end{lemma}
Lemma~\ref{lem:modules} allows us to express any operator $M$ acting on $\HC$ as $M=\sum_{\kappa=0}^{2n} M_\kappa$ with $M_\kappa\in\BC_\kappa$ a homogeneous multilinear polynomial of degree $\kappa$ in the Majoranas. In particular, the $M_\kappa$ are the orthogonal projections of $M$ 
into  each $\BC_\kappa$, i.e., $M_\kappa=\sum_{j=1}^{\dim( \BC_\kappa)} \Tr[B_j M] B_j$, where $\{B_j\}_{j=1}^{\dim(\BC_\kappa)}$ is a Hermitian orthonormal basis  for $\BC_\kappa$  with respect to the standard Hilbert--Schmidt inner product. 
Here, we also note that $\BC_0={\rm span}_{\mbb{C}}\id$, while $\BC_{2n}={\rm span}_{\mbb{C}}P$ with $P=Z^{\otimes n}=(-i)^{n}c_1c_2\cdots c_{2n}$ the fermionic parity operator.  
More generally, the basis elements of $\BC_\kappa$ are related to those in $\BC_{2n-\kappa}$ via multiplication by $P$. 

In fact, as shown by the next lemma, the subspaces $\{\BC_{\kappa}\}$ are invariant under the action of $G$, i.e., are $G$-modules.  
\begin{lemma}\label{lem:inv}
Let $M_\kappa\in \BC_\kappa$, then  $\forall U\in G$, $UM_\kappa U\ad\in\BC_\kappa$. Moreover, any pair of Pauli operators in $\BC_\kappa$ are proportional to each other via commutation with elements in $\g$.
\end{lemma}
In the SI we present a general formula for a basis of the commutant of the tensor square representation of any Lie group with associated Pauli-string Lie algebra (such as $G$), which may be of independent interest for the reader. This result, along with the previous lemmas, allow us to compute the variance of the loss function  via the Weingarten calculus~\cite{ragone2022representation,garcia2023deep,mele2023introduction}, and prove the following theorem:
\begin{theorem}\label{theo:main}
The loss function in Eq.~\eqref{eq:loss} has mean
\begin{align}\label{eq:mean}
    \mathbb{E}_{\thv}[\ell_{\thv}(\rho,O)]&\!=\sum_{\kappa=0,2n}\!\langle\rho_\kappa,O_\kappa\rangle_{\id+P}\,,
\end{align}
and variance
\small
\begin{equation}\label{eq:variance}
\begin{split}
    \Var_{\thv}[\ell_{\thv}(\rho,O)]=   \sum_{\kappa=1}^{2n-1} &\frac{\PC_\kappa(\rho)\PC_\kappa(O)+\CC_\kappa(\rho)\CC_\kappa(O)}{\dim(\BC_\kappa)}\,.
\end{split}
\end{equation}
\normalsize
Here we defined the $\kappa$-purity of an operator $M\in\BC$ as $\PC_\kappa(M)=\langle M_\kappa,M_\kappa\rangle_{\id}$, and its $\kappa$-coherence  as $\CC_\kappa(M)=i^{\kappa \,{\rm mod\, 2}}\langle M_\kappa,M_{2n-\kappa}\rangle_{P}$. Moreover, for $M_1,M_2,\Gamma\in\BC$, we have $\langle M_1,M_2\rangle_{\Gamma}=\Tr[\Gamma M_1\ad M_2]$.
\end{theorem}

Theorem~\ref{theo:main} provides an exact closed formula for the loss variance for arbitrary input states and measurements, meaning that we can completely characterize the loss concentration phenomenon for parametrized matchgate circuits. This result shows that the mean of the loss is solely determined by the component of $\rho$ and $O$ in the $0$-th and $2n$-th modules, i.e., in the modules that commute with $G$. To unpack  Eq.~\eqref{eq:variance} we first consider the following case:

\begin{corollary}\label{cor:in-module}
Let $O\in\BC_\kappa$, with $\kappa\neq n$, then $
   \Var_{\thv}[\ell_{\thv}(\rho,O)]=  
\frac{\PC_\kappa(\rho)\PC_\kappa(O)}{\dim(\BC_\kappa)}$.
\end{corollary}
Corollary~\ref{cor:in-module} shows that when $O\in \BC_{\kappa}$ the variance depends solely on three quantities: The $\kappa$-purities of $\rho$ and $O$, and the dimension of $\BC_\kappa$. In particular, for the special case when $\kappa=2$, when $\BC_2={\rm span}_{\mathbb{C}}\,\g$,  we recover the main result of Refs.~\cite{fontana2023theadjoint,ragone2023unified}, whose interpretation we briefly recall. As discussed in~\cite{ragone2023unified}, $\PC_2(\rho)$ quantifies the generalized entanglement in the initial state~\cite{somma2004nature,somma2005quantum}, $\PC_2(O)$  the generalized locality of the measurement, and $\dim(\g)$ the expressiveness of the circuit~\cite{larocca2021diagnosing}. Here, one defines a generalization notion of entanglement and locality relative to a subspace of $\BC$, rather than the standard notion that relies on the subsystem decomposition of $\HC$. Moreover, we recall that $\dim(\g)$ is a ``local'' measure of  expressiveness of the circuit~\cite{larocca2021diagnosing,larocca2021theory}, as it quantifies the number of independent directions in the tangent plane at identity on the group's manifold. We can then
use these results to interpret our more general expression in Corollary~\ref{cor:in-module}. First, we have that the variance still depends on purities  but now with respect to the set of operators in the corresponding module $\BC_{\kappa}$.  In this sense, we will call $\rho$ generalized-unentangled if $\PC_\kappa(\rho)$ is  maximal (i.e., if $\rho$ belongs to the orbit of a simultaneous eigenstate of  a maximally commuting subspace of
$\BC_\kappa$), and generalized-entangled otherwise. 
Remarkably, there is additional operational meaning, inherited by the group defining our circuits: In the SI we show that $\PC_\mf{g}(\rho)$ is essentially the fermionic entanglement entropy introduced in \cite{gigena2015entanglement, gigena2020one}. The latter is a rigorous measure of total entanglement, developed with the aim of properly quantifying correlations between indistinguishable (fermionic) particles. 
In the SI we also discuss how  $\PC_\kappa(\rho)$ seems to provide more general measures of fermionic correlations, generalizing \cite{gigena2021many}, an extension of \cite{gigena2015entanglement, gigena2020one} itself.

\begin{figure*}[t!]
    \centering
    \includegraphics[width=\linewidth]{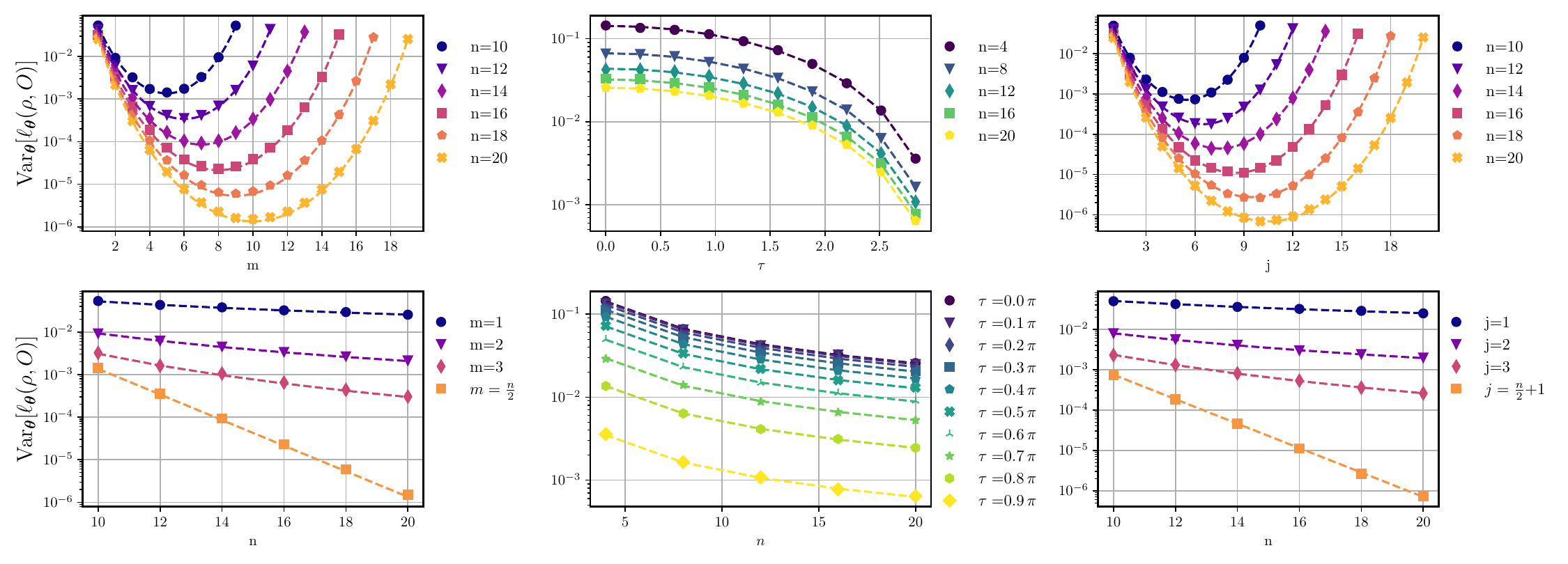}
    \caption{  Variance of the loss function for  
    three  different setups where we vary $\rho$ and $O$. Left:  $\rho=\left(|0\rangle\langle 0|\right)^{\otimes n}$ is a Gaussian state, and $O=Z^{\otimes m}\in \BC_{2m}$.  
    Middle:  $\rho=|\Psi(\tau)\rangle\langle\Psi(\tau)|$ is a magic state with $|\Psi(\tau)\rangle= \big[\frac{|0000\rangle+|0011\rangle+|1100\rangle+e^{i\tau}|1111\rangle}{2}\big]^{\otimes n/4}$, and $O=Z_1\in\BC_2$. Right: $\rho=\ket{\psi}\bra{\psi}$ is a non-fermionic state with $\ket{\psi}=\frac{1}{\sqrt{2}}\left(\ket{0}^{\otimes n} + \ket{1}\ket{0}^{\otimes n-1}\right)$, and $O=X_j\in\BC_{2j-1}$. In the top row  we respectively show the variance scaling as a function of the parameters $m$, $\tau$ and $j$ for different number of qubits while in the bottom row we fix the parameter value and study the scaling with $n$. In all cases, each dot is computed from $10^4$ independent random  parameter values uniformly drawn from $[0,2\pi]$, and the quantum circuit consists of $L=n^2$ layers to ensure good convergence to an approximate $2$-design. Simulations were performed with the open-source library \texttt{Qibo}~\cite{efthymiou2020qibo,efthymiou2022quantum}.  Dashed lines correspond to our derived  analytical expressions. }
    \label{fig:numerics}
\end{figure*}

Next, we will say that $O$ is $\kappa$-generalized pure with respect to $\BC_\kappa$ if $O_\kappa=O$. This implies that more generalized-entangled states, or less generalized-pure measurements lead to smaller variances. 
More interestingly, we now see that the variance is inversely proportional to $\dim(\BC_\kappa)=\binom{2n}{\kappa}$. This reveals that when $\kappa\not\in\{ 2,2n-2\}$, the variance does not satisfy the 
inverse-polynomial-in-$\dim(\mf{g})$ scaling conjectured in~\cite{larocca2021diagnosing}. 
Notably, $\dim(\BC_\kappa)$ quantifies the dimension of the operator subspace explored by the circuit, and hence constitutes a generalized measure of expressiveness. That is, the larger the module, the larger the subspace explored, the more expressive the circuit, and therefore the more concentrated the variance (in accordance with the results in~\cite{holmes2021connecting}). For instance, we can see that as $\kappa$ increases (up to $\kappa=n$), the dimension of the modules quickly become exponentially large, indicating that the loss functions are prone to expressiveness-induced BPs. In fact, we can also provide operational meaning to this result by defining a notion of generalized globality. In particular, recalling that the elements in $\BC_\kappa$ are expressed in terms of products  of $\kappa$ Majoranas, we can always say that an operator is ``less'' generalized-global if $\kappa\,{\rm mod\, }n\in\Theta(1)$ is small, whereas it is ``more'' generalized-global if $\kappa\,{\rm mod\, }n\in\Theta(n)$ (note that we take modulo $n$ to account for the parity symmetry). With this interpretation, measuring a generalized global operator given by a product of many Majoranas  leads to smaller variances (since $\dim(\BC_\kappa)$ is larger). Conceptually, this phenomenon generalizes the results in~\cite{cerezo2020cost} which show that for certain circuit architectures measuring global operators (in the standard sense of acting non-trivially on many qubits) leads to BPs.

Next, let us consider a slightly more general setting.
\begin{corollary}\label{cor:in-moduleskkp}
Let $O\in\BC_\kappa\oplus \BC_{\kappa'}$, with $\kappa,\kappa'\neq n$, then
\begin{equation}\label{eq:o-in-bkbkp}
\begin{split}
   \Var_{\thv}[\ell_{\thv}(\rho,O)]=&  
 \frac{\PC_\kappa(\rho)\PC_\kappa(O)}{\dim(\BC_\kappa)}+\frac{\PC_{\kappa'}(\rho)\PC_{\kappa'}(O)}{\dim(\BC_{\kappa'})}\\
 &+2\delta_{k+k',2n}\frac{\CC_\kappa(\rho)\CC_\kappa(O)}{\dim(\BC_\kappa)} \,.
 \end{split}
\end{equation}
\end{corollary}
The first two terms in Eq.~\eqref{eq:o-in-bkbkp} account for the $\kappa$-purities of $\rho$ and $O$, and the expressiveness in each module. However a covariance appears for the special case when $k'=2n-k$ (as these modules are isomorphic, and thus satisfy the conditions proposed in~\cite{fontana2023theadjoint} for the covariance to be non-zero).  In particular, the relative sign between the basis elements determines whether the $\kappa$-coherences are  positive or negative, and whether this term increases or decreases the variance. Note that in the special case of $\kappa=n$ both $B_j$ and $PB_j$ are in $\BC_n$, meaning that the variance of an operator in $\BC_n$ can always contain $\kappa$-coherences.  In the SI, we show that choosing different bases for the tensor square commutant leads to different interpretations for the terms in $\Var_{\thv}[\ell_{\thv}(\rho,O)]$. For instance, we find that the module $\BC_n$ does not produce coherences in a parity-aware basis, as it can be decomposed into two non-isomorphic sub-modules. We also show that in this basis,  the coherences disappear for Fermionic states and measurements $\forall \kappa$.

Taken together, the previous corollaries allow us to understand  Theorem~\ref{theo:main}. That is, Eq.~\eqref{eq:variance} indicates that the loss variance depends on the generalized purities, entanglement, and expressiveness in each $G$-module,  but that at the same time there also exist contributions arising from the coherences in the initial state and measurement operator between the isomorphic modules connected by parity. 

At this point we highlight the fact that the generalized entanglement across all $G$-modules must satisfy a monogamy of generalized-entanglement relation. Recall that the standard monogamy of entanglement results restrict how much entanglement can be shared among multi-party systems~\cite{coffman2000distributed}. That is, if two parties are highly entangled among themselves, then they cannot be too entangled with the rest. We see that since $\sum_{\kappa=0}^{2n}\PC_{\kappa}(\rho)=\Tr[\rho^2]/2^n$, then  the total amount of generalized  entanglement is bounded by  the standard purity of $\rho$. Hence, if one (or) some of the $\PC_{\kappa}(\rho)$ are small, indicating high generalized entanglement in these $G$-modules, then the rest of the $\kappa$-purities must be large, implying small amounts of generalized entanglement in the rest.

\textit{Illustrative examples}. 
We now present examples capturing different features revealed by our formalism. We first  consider  the case when $\rho$ is a pure fermionic Gaussian state and $O\in \BC_\kappa$ with $\kappa$ even, e.g., $O=Z^{\otimes \kappa/2}$. We can use Theorem \ref{theo:main} to find
$ \Var_{\thv}[\ell_{\thv}(\rho,O)]= \binom{n}{\kappa/2}  \binom{2n}{\kappa}^{-1}$ (see the SI), indicating that the variance decreases as $\kappa\,{\rm mod\, }n$  increases (see Fig.~\ref{fig:numerics} (left, top)). Indeed, we can see from Fig.~\ref{fig:numerics} (left, bottom) that if $\kappa$ does not scale with $n$, then $ \Var_{\thv}[\ell_{\thv}(\rho,O)]\in\OC(1/\poly(n))$ and the loss has no BP. However, if  $\kappa$ scales super-logarithmically with $n$, a BP is induced by the globality of the measurement operator (i.e., by the high expressiveness of the circuit in this module).  

Next, we consider a family of fermionic magic (non-Gaussian) states~ \cite{hebenstreit2019all} given by  ${|\Psi(\tau)\rangle= \big[\frac{|0000\rangle+|0011\rangle+|1100\rangle+e^{i\tau}|1111\rangle}{2}\big]^{\otimes n/4}}$. These states have been shown to be a resource in quantum computation, being able to  promote  
matchgate circuits to universal quantum computation via gadgetization~\cite{reardon2023improved}. They also have an exponential  ``extent''~\cite{dias2023classical,cudby2023gaussian} for any fixed $\tau\in (0,\pi]$, meaning that their expansion in Gaussian states requires exponentially many terms leading to nonsimulable loss functions via Wick's theorem. Interestingly, while the variance decreases for increasing magic parameter $\tau$, this model is trainable for any $\pi-\tau\in \Omega(1/\poly(n))$, which includes states with exponential extent.
In the SI we use recent fermionic entanglement techniques \cite{gigena2015entanglement} to show that for $O\in i\g$, $\mathcal{P}_2$ is 
``additive'' with respect to $n$, yielding $\Var_{\thv}[\ell_{\thv}(\rho,O)]= 
n\cos^2(\tau/2)\binom{2n}{2}^{-1}$.

To finish, we consider the non-fermionic initial state $|\psi\rangle=\alpha |0\rangle^{\otimes n}+\beta |1\rangle|0\rangle^{\otimes n-1}$ and non-fermionic measure operator $O=X_j\in\BC_{2j-1}$ (by non-fermionic we mean  not commuting with $P$). In this case we find $\Var_{\thv}[C(\thv)]=4\alpha^2\beta^2 \binom{n-1}{j-1}\binom{2n}{2j-1}^{-1}$. While the operators $X_j$ are clearly local in the conventional sense, they have generalize globality $(2j-1)\,{\rm mod\, }\,n$, again verifying that variance decreases with generalized globality (see  Fig.~\ref{fig:numerics} (right)). 
This also implies the following surprising result: A loss function with  $O=Z_{ n/2}\in\BC_2$ (assume $n$ even) has no BP, while one with $O=X_{ n/2}\in\BC_{n-1}$ has BPs. Since one can map $Z_{ n/2}\to X_{ n/2}$ via a rotation around the $Y$ axis on qubit $n/2$, we can see that adding \emph{just one} single qubit rotation at the end of the circuit can transform a loss with no-BPs into one with BPs.

\textit{Discussion and Outlook}. As variational quantum computing continues to be an increasingly popular candidate for practical quantum applications, BPs stand as a fundamental challenge that may curtail its utility. In this context, our work provides a crucial step towards a general understating of the BP phenomenon. By focusing on parametrized matchgate circuits we can go further than, and generalize, previous results requiring certain assumptions on the initial state and measurement operator. Our results not only confirm previous intuitions (like the existence of covariances from isomorphic modules proposed in Ref.~\cite{fontana2023theadjoint}), but also reveal new BP physics such as the generalized globality and the  module expressiveness. We thus hope that our results will inspire the community to take the final steps towards a more general ultimate BP theory. 

At a more fundamental level, we prove an intrinsic connection between generalized entanglement and known fermionic entanglement measures \cite{gigena2015entanglement, gigena2020one,gigena2021many}. On its own, this connection is remarkable, as the underlying physical system in our framework is composed of (distinguishable) qubits on a quantum computer, and not of actual fermions. Moreover, this connection uncovers the alarming fact that variances become smaller when the quantum resources (measured by fermionic entanglement) are increased (see remarks on the SI). 
Such observation also raises a critical question: \textit{Does the inherent structure 
that precludes the presence of BPs in a variational model (a requisite for trainability) simultaneously render it classically simulable?} Answering this question is crucial in the pursuit for variational quantum advantage. While our results provide the basis for addressing the problem rigorously, they also show that the question is deeply connected  
with challenging foundational problems, such as the proper definition and characterization of quantum correlations 
for indistinguishable particles and their connection with computational hardness.

\section*{Acknowledgments}

The authors would like to thank Nicolás Gigena, Raúl Rossignoli, Bojko N. Bakalov, Sam Slezak, and  Dylan Herman for fruitful discussions. N.L.D,  M.C. and D.G.M. were supported by the Laboratory Directed Research and Development (LDRD) program of Los Alamos National Laboratory (LANL) under project numbers 20210116DR, 20230049DR and 20230527ECR, respectively. N.L.D. was initially supported by CONICET  Argentina. M.L. also acknowledges support by the Center for Nonlinear Studies at LANL by the U.S. DOE, Office of Science, Office of Advanced Scientific Computing Research, under the Accelerated Research in Quantum Computing (ARQC) program. M.C. was also supported by LANL's ASC Beyond Moore’s Law project. S.K. was supported by the Department of Defense (DoD) through the National Defense Science \& Engineering Graduate (NDSEG) Fellowship Program.

\bibliography{quantum}

%apsrev4-2.bst 2019-01-14 (MD) hand-edited version of apsrev4-1.bst
%Control: key (0)
%Control: author (8) initials jnrlst
%Control: editor formatted (1) identically to author
%Control: production of article title (0) allowed
%Control: page (0) single
%Control: year (1) truncated
%Control: production of eprint (0) enabled
\begin{thebibliography}{68}%
\makeatletter
\providecommand \@ifxundefined [1]{%
 \@ifx{#1\undefined}
}%
\providecommand \@ifnum [1]{%
 \ifnum #1\expandafter \@firstoftwo
 \else \expandafter \@secondoftwo
 \fi
}%
\providecommand \@ifx [1]{%
 \ifx #1\expandafter \@firstoftwo
 \else \expandafter \@secondoftwo
 \fi
}%
\providecommand \natexlab [1]{#1}%
\providecommand \enquote  [1]{``#1''}%
\providecommand \bibnamefont  [1]{#1}%
\providecommand \bibfnamefont [1]{#1}%
\providecommand \citenamefont [1]{#1}%
\providecommand \href@noop [0]{\@secondoftwo}%
\providecommand \href [0]{\begingroup \@sanitize@url \@href}%
\providecommand \@href[1]{\@@startlink{#1}\@@href}%
\providecommand \@@href[1]{\endgroup#1\@@endlink}%
\providecommand \@sanitize@url [0]{\catcode `\\12\catcode `\$12\catcode
  `\&12\catcode `\#12\catcode `\^12\catcode `\_12\catcode `\%12\relax}%
\providecommand \@@startlink[1]{}%
\providecommand \@@endlink[0]{}%
\providecommand \url  [0]{\begingroup\@sanitize@url \@url }%
\providecommand \@url [1]{\endgroup\@href {#1}{\urlprefix }}%
\providecommand \urlprefix  [0]{URL }%
\providecommand \Eprint [0]{\href }%
\providecommand \doibase [0]{https://doi.org/}%
\providecommand \selectlanguage [0]{\@gobble}%
\providecommand \bibinfo  [0]{\@secondoftwo}%
\providecommand \bibfield  [0]{\@secondoftwo}%
\providecommand \translation [1]{[#1]}%
\providecommand \BibitemOpen [0]{}%
\providecommand \bibitemStop [0]{}%
\providecommand \bibitemNoStop [0]{.\EOS\space}%
\providecommand \EOS [0]{\spacefactor3000\relax}%
\providecommand \BibitemShut  [1]{\csname bibitem#1\endcsname}%
\let\auto@bib@innerbib\@empty
%</preamble>
\bibitem [{\citenamefont {Cerezo}\ \emph
  {et~al.}(2021{\natexlab{a}})\citenamefont {Cerezo}, \citenamefont
  {Arrasmith}, \citenamefont {Babbush}, \citenamefont {Benjamin}, \citenamefont
  {Endo}, \citenamefont {Fujii}, \citenamefont {McClean}, \citenamefont
  {Mitarai}, \citenamefont {Yuan}, \citenamefont {Cincio},\ and\ \citenamefont
  {Coles}}]{cerezo2020variationalreview}%
  \BibitemOpen
  \bibfield  {author} {\bibinfo {author} {\bibfnamefont {M.}~\bibnamefont
  {Cerezo}}, \bibinfo {author} {\bibfnamefont {A.}~\bibnamefont {Arrasmith}},
  \bibinfo {author} {\bibfnamefont {R.}~\bibnamefont {Babbush}}, \bibinfo
  {author} {\bibfnamefont {S.~C.}\ \bibnamefont {Benjamin}}, \bibinfo {author}
  {\bibfnamefont {S.}~\bibnamefont {Endo}}, \bibinfo {author} {\bibfnamefont
  {K.}~\bibnamefont {Fujii}}, \bibinfo {author} {\bibfnamefont {J.~R.}\
  \bibnamefont {McClean}}, \bibinfo {author} {\bibfnamefont {K.}~\bibnamefont
  {Mitarai}}, \bibinfo {author} {\bibfnamefont {X.}~\bibnamefont {Yuan}},
  \bibinfo {author} {\bibfnamefont {L.}~\bibnamefont {Cincio}},\ and\ \bibinfo
  {author} {\bibfnamefont {P.~J.}\ \bibnamefont {Coles}},\ }\bibfield  {title}
  {\bibinfo {title} {Variational quantum algorithms},\ }\href
  {https://doi.org/10.1038/s42254-021-00348-9} {\bibfield  {journal} {\bibinfo
  {journal} {Nature Reviews Physics}\ }\textbf {\bibinfo {volume} {3}},\
  \bibinfo {pages} {625–644} (\bibinfo {year}
  {2021}{\natexlab{a}})}\BibitemShut {NoStop}%
\bibitem [{\citenamefont {Bharti}\ \emph {et~al.}(2022)\citenamefont {Bharti},
  \citenamefont {Cervera-Lierta}, \citenamefont {Kyaw}, \citenamefont {Haug},
  \citenamefont {Alperin-Lea}, \citenamefont {Anand}, \citenamefont {Degroote},
  \citenamefont {Heimonen}, \citenamefont {Kottmann}, \citenamefont {Menke}
  \emph {et~al.}}]{bharti2021noisy}%
  \BibitemOpen
  \bibfield  {author} {\bibinfo {author} {\bibfnamefont {K.}~\bibnamefont
  {Bharti}}, \bibinfo {author} {\bibfnamefont {A.}~\bibnamefont
  {Cervera-Lierta}}, \bibinfo {author} {\bibfnamefont {T.~H.}\ \bibnamefont
  {Kyaw}}, \bibinfo {author} {\bibfnamefont {T.}~\bibnamefont {Haug}}, \bibinfo
  {author} {\bibfnamefont {S.}~\bibnamefont {Alperin-Lea}}, \bibinfo {author}
  {\bibfnamefont {A.}~\bibnamefont {Anand}}, \bibinfo {author} {\bibfnamefont
  {M.}~\bibnamefont {Degroote}}, \bibinfo {author} {\bibfnamefont
  {H.}~\bibnamefont {Heimonen}}, \bibinfo {author} {\bibfnamefont {J.~S.}\
  \bibnamefont {Kottmann}}, \bibinfo {author} {\bibfnamefont {T.}~\bibnamefont
  {Menke}}, \emph {et~al.},\ }\bibfield  {title} {\bibinfo {title} {Noisy
  intermediate-scale quantum algorithms},\ }\href
  {https://doi.org/10.1103/RevModPhys.94.015004} {\bibfield  {journal}
  {\bibinfo  {journal} {Reviews of Modern Physics}\ }\textbf {\bibinfo {volume}
  {94}},\ \bibinfo {pages} {015004} (\bibinfo {year} {2022})}\BibitemShut
  {NoStop}%
\bibitem [{\citenamefont {Endo}\ \emph {et~al.}(2021)\citenamefont {Endo},
  \citenamefont {Cai}, \citenamefont {Benjamin},\ and\ \citenamefont
  {Yuan}}]{endo2021hybrid}%
  \BibitemOpen
  \bibfield  {author} {\bibinfo {author} {\bibfnamefont {S.}~\bibnamefont
  {Endo}}, \bibinfo {author} {\bibfnamefont {Z.}~\bibnamefont {Cai}}, \bibinfo
  {author} {\bibfnamefont {S.~C.}\ \bibnamefont {Benjamin}},\ and\ \bibinfo
  {author} {\bibfnamefont {X.}~\bibnamefont {Yuan}},\ }\bibfield  {title}
  {\bibinfo {title} {Hybrid quantum-classical algorithms and quantum error
  mitigation},\ }\href {https://doi.org/10.7566/JPSJ.90.032001} {\bibfield
  {journal} {\bibinfo  {journal} {Journal of the Physical Society of Japan}\
  }\textbf {\bibinfo {volume} {90}},\ \bibinfo {pages} {032001} (\bibinfo
  {year} {2021})}\BibitemShut {NoStop}%
\bibitem [{\citenamefont {Schuld}\ \emph {et~al.}(2015)\citenamefont {Schuld},
  \citenamefont {Sinayskiy},\ and\ \citenamefont
  {Petruccione}}]{schuld2015introduction}%
  \BibitemOpen
  \bibfield  {author} {\bibinfo {author} {\bibfnamefont {M.}~\bibnamefont
  {Schuld}}, \bibinfo {author} {\bibfnamefont {I.}~\bibnamefont {Sinayskiy}},\
  and\ \bibinfo {author} {\bibfnamefont {F.}~\bibnamefont {Petruccione}},\
  }\bibfield  {title} {\bibinfo {title} {An introduction to quantum machine
  learning},\ }\href {https://doi.org/10.1080/00107514.2014.964942} {\bibfield
  {journal} {\bibinfo  {journal} {Contemporary Physics}\ }\textbf {\bibinfo
  {volume} {56}},\ \bibinfo {pages} {172} (\bibinfo {year} {2015})}\BibitemShut
  {NoStop}%
\bibitem [{\citenamefont {Biamonte}\ \emph {et~al.}(2017)\citenamefont
  {Biamonte}, \citenamefont {Wittek}, \citenamefont {Pancotti}, \citenamefont
  {Rebentrost}, \citenamefont {Wiebe},\ and\ \citenamefont
  {Lloyd}}]{biamonte2017quantum}%
  \BibitemOpen
  \bibfield  {author} {\bibinfo {author} {\bibfnamefont {J.}~\bibnamefont
  {Biamonte}}, \bibinfo {author} {\bibfnamefont {P.}~\bibnamefont {Wittek}},
  \bibinfo {author} {\bibfnamefont {N.}~\bibnamefont {Pancotti}}, \bibinfo
  {author} {\bibfnamefont {P.}~\bibnamefont {Rebentrost}}, \bibinfo {author}
  {\bibfnamefont {N.}~\bibnamefont {Wiebe}},\ and\ \bibinfo {author}
  {\bibfnamefont {S.}~\bibnamefont {Lloyd}},\ }\bibfield  {title} {\bibinfo
  {title} {Quantum machine learning},\ }\href
  {https://doi.org/10.1038/nature23474} {\bibfield  {journal} {\bibinfo
  {journal} {Nature}\ }\textbf {\bibinfo {volume} {549}},\ \bibinfo {pages}
  {195} (\bibinfo {year} {2017})}\BibitemShut {NoStop}%
\bibitem [{\citenamefont {Havl{\'\i}{\v{c}}ek}\ \emph
  {et~al.}(2019)\citenamefont {Havl{\'\i}{\v{c}}ek}, \citenamefont
  {C{\'o}rcoles}, \citenamefont {Temme}, \citenamefont {Harrow}, \citenamefont
  {Kandala}, \citenamefont {Chow},\ and\ \citenamefont
  {Gambetta}}]{havlivcek2019supervised}%
  \BibitemOpen
  \bibfield  {author} {\bibinfo {author} {\bibfnamefont {V.}~\bibnamefont
  {Havl{\'\i}{\v{c}}ek}}, \bibinfo {author} {\bibfnamefont {A.~D.}\
  \bibnamefont {C{\'o}rcoles}}, \bibinfo {author} {\bibfnamefont
  {K.}~\bibnamefont {Temme}}, \bibinfo {author} {\bibfnamefont {A.~W.}\
  \bibnamefont {Harrow}}, \bibinfo {author} {\bibfnamefont {A.}~\bibnamefont
  {Kandala}}, \bibinfo {author} {\bibfnamefont {J.~M.}\ \bibnamefont {Chow}},\
  and\ \bibinfo {author} {\bibfnamefont {J.~M.}\ \bibnamefont {Gambetta}},\
  }\bibfield  {title} {\bibinfo {title} {Supervised learning with
  quantum-enhanced feature spaces},\ }\href
  {https://doi.org/10.1038/s41586-019-0980-2} {\bibfield  {journal} {\bibinfo
  {journal} {Nature}\ }\textbf {\bibinfo {volume} {567}},\ \bibinfo {pages}
  {209} (\bibinfo {year} {2019})}\BibitemShut {NoStop}%
\bibitem [{\citenamefont {McClean}\ \emph {et~al.}(2018)\citenamefont
  {McClean}, \citenamefont {Boixo}, \citenamefont {Smelyanskiy}, \citenamefont
  {Babbush},\ and\ \citenamefont {Neven}}]{mcclean2018barren}%
  \BibitemOpen
  \bibfield  {author} {\bibinfo {author} {\bibfnamefont {J.~R.}\ \bibnamefont
  {McClean}}, \bibinfo {author} {\bibfnamefont {S.}~\bibnamefont {Boixo}},
  \bibinfo {author} {\bibfnamefont {V.~N.}\ \bibnamefont {Smelyanskiy}},
  \bibinfo {author} {\bibfnamefont {R.}~\bibnamefont {Babbush}},\ and\ \bibinfo
  {author} {\bibfnamefont {H.}~\bibnamefont {Neven}},\ }\bibfield  {title}
  {\bibinfo {title} {Barren plateaus in quantum neural network training
  landscapes},\ }\href {https://doi.org/10.1038/s41467-018-07090-4} {\bibfield
  {journal} {\bibinfo  {journal} {Nature {C}ommunications}\ }\textbf {\bibinfo
  {volume} {9}},\ \bibinfo {pages} {1} (\bibinfo {year} {2018})}\BibitemShut
  {NoStop}%
\bibitem [{\citenamefont {Cerezo}\ \emph
  {et~al.}(2021{\natexlab{b}})\citenamefont {Cerezo}, \citenamefont {Sone},
  \citenamefont {Volkoff}, \citenamefont {Cincio},\ and\ \citenamefont
  {Coles}}]{cerezo2020cost}%
  \BibitemOpen
  \bibfield  {author} {\bibinfo {author} {\bibfnamefont {M.}~\bibnamefont
  {Cerezo}}, \bibinfo {author} {\bibfnamefont {A.}~\bibnamefont {Sone}},
  \bibinfo {author} {\bibfnamefont {T.}~\bibnamefont {Volkoff}}, \bibinfo
  {author} {\bibfnamefont {L.}~\bibnamefont {Cincio}},\ and\ \bibinfo {author}
  {\bibfnamefont {P.~J.}\ \bibnamefont {Coles}},\ }\bibfield  {title} {\bibinfo
  {title} {Cost function dependent barren plateaus in shallow parametrized
  quantum circuits},\ }\href {https://doi.org/10.1038/s41467-021-21728-w}
  {\bibfield  {journal} {\bibinfo  {journal} {Nature {C}ommunications}\
  }\textbf {\bibinfo {volume} {12}},\ \bibinfo {pages} {1} (\bibinfo {year}
  {2021}{\natexlab{b}})}\BibitemShut {NoStop}%
\bibitem [{\citenamefont {Marrero}\ \emph {et~al.}(2021)\citenamefont
  {Marrero}, \citenamefont {Kieferov{\'a}},\ and\ \citenamefont
  {Wiebe}}]{marrero2020entanglement}%
  \BibitemOpen
  \bibfield  {author} {\bibinfo {author} {\bibfnamefont {C.~O.}\ \bibnamefont
  {Marrero}}, \bibinfo {author} {\bibfnamefont {M.}~\bibnamefont
  {Kieferov{\'a}}},\ and\ \bibinfo {author} {\bibfnamefont {N.}~\bibnamefont
  {Wiebe}},\ }\bibfield  {title} {\bibinfo {title} {Entanglement-induced barren
  plateaus},\ }\href {https://doi.org/10.1103/PRXQuantum.2.040316} {\bibfield
  {journal} {\bibinfo  {journal} {PRX Quantum}\ }\textbf {\bibinfo {volume}
  {2}},\ \bibinfo {pages} {040316} (\bibinfo {year} {2021})}\BibitemShut
  {NoStop}%
\bibitem [{\citenamefont {Sharma}\ \emph {et~al.}(2022)\citenamefont {Sharma},
  \citenamefont {Cerezo}, \citenamefont {Cincio},\ and\ \citenamefont
  {Coles}}]{sharma2020trainability}%
  \BibitemOpen
  \bibfield  {author} {\bibinfo {author} {\bibfnamefont {K.}~\bibnamefont
  {Sharma}}, \bibinfo {author} {\bibfnamefont {M.}~\bibnamefont {Cerezo}},
  \bibinfo {author} {\bibfnamefont {L.}~\bibnamefont {Cincio}},\ and\ \bibinfo
  {author} {\bibfnamefont {P.~J.}\ \bibnamefont {Coles}},\ }\bibfield  {title}
  {\bibinfo {title} {Trainability of dissipative perceptron-based quantum
  neural networks},\ }\href {https://doi.org/10.1103/PhysRevLett.128.180505}
  {\bibfield  {journal} {\bibinfo  {journal} {Physical Review Letters}\
  }\textbf {\bibinfo {volume} {128}},\ \bibinfo {pages} {180505} (\bibinfo
  {year} {2022})}\BibitemShut {NoStop}%
\bibitem [{\citenamefont {Patti}\ \emph {et~al.}(2021)\citenamefont {Patti},
  \citenamefont {Najafi}, \citenamefont {Gao},\ and\ \citenamefont
  {Yelin}}]{patti2020entanglement}%
  \BibitemOpen
  \bibfield  {author} {\bibinfo {author} {\bibfnamefont {T.~L.}\ \bibnamefont
  {Patti}}, \bibinfo {author} {\bibfnamefont {K.}~\bibnamefont {Najafi}},
  \bibinfo {author} {\bibfnamefont {X.}~\bibnamefont {Gao}},\ and\ \bibinfo
  {author} {\bibfnamefont {S.~F.}\ \bibnamefont {Yelin}},\ }\bibfield  {title}
  {\bibinfo {title} {Entanglement devised barren plateau mitigation},\ }\href
  {https://doi.org/10.1103/PhysRevResearch.3.033090} {\bibfield  {journal}
  {\bibinfo  {journal} {Physical Review Research}\ }\textbf {\bibinfo {volume}
  {3}},\ \bibinfo {pages} {033090} (\bibinfo {year} {2021})}\BibitemShut
  {NoStop}%
\bibitem [{\citenamefont {Pesah}\ \emph {et~al.}(2021)\citenamefont {Pesah},
  \citenamefont {Cerezo}, \citenamefont {Wang}, \citenamefont {Volkoff},
  \citenamefont {Sornborger},\ and\ \citenamefont {Coles}}]{pesah2020absence}%
  \BibitemOpen
  \bibfield  {author} {\bibinfo {author} {\bibfnamefont {A.}~\bibnamefont
  {Pesah}}, \bibinfo {author} {\bibfnamefont {M.}~\bibnamefont {Cerezo}},
  \bibinfo {author} {\bibfnamefont {S.}~\bibnamefont {Wang}}, \bibinfo {author}
  {\bibfnamefont {T.}~\bibnamefont {Volkoff}}, \bibinfo {author} {\bibfnamefont
  {A.~T.}\ \bibnamefont {Sornborger}},\ and\ \bibinfo {author} {\bibfnamefont
  {P.~J.}\ \bibnamefont {Coles}},\ }\bibfield  {title} {\bibinfo {title}
  {Absence of barren plateaus in quantum convolutional neural networks},\
  }\href {https://doi.org/10.1103/PhysRevX.11.041011} {\bibfield  {journal}
  {\bibinfo  {journal} {Physical Review X}\ }\textbf {\bibinfo {volume} {11}},\
  \bibinfo {pages} {041011} (\bibinfo {year} {2021})}\BibitemShut {NoStop}%
\bibitem [{\citenamefont {Uvarov}\ and\ \citenamefont
  {Biamonte}(2021)}]{uvarov2020barren}%
  \BibitemOpen
  \bibfield  {author} {\bibinfo {author} {\bibfnamefont {A.}~\bibnamefont
  {Uvarov}}\ and\ \bibinfo {author} {\bibfnamefont {J.~D.}\ \bibnamefont
  {Biamonte}},\ }\bibfield  {title} {\bibinfo {title} {On barren plateaus and
  cost function locality in variational quantum algorithms},\ }\href
  {https://doi.org/10.1088/1751-8121/abfac7} {\bibfield  {journal} {\bibinfo
  {journal} {Journal of Physics A: Mathematical and Theoretical}\ }\textbf
  {\bibinfo {volume} {54}},\ \bibinfo {pages} {245301} (\bibinfo {year}
  {2021})}\BibitemShut {NoStop}%
\bibitem [{\citenamefont {Cerezo}\ and\ \citenamefont
  {Coles}(2021)}]{cerezo2020impact}%
  \BibitemOpen
  \bibfield  {author} {\bibinfo {author} {\bibfnamefont {M.}~\bibnamefont
  {Cerezo}}\ and\ \bibinfo {author} {\bibfnamefont {P.~J.}\ \bibnamefont
  {Coles}},\ }\bibfield  {title} {\bibinfo {title} {Higher order derivatives of
  quantum neural networks with barren plateaus},\ }\href
  {https://doi.org/10.1088/2058-9565/abf51a} {\bibfield  {journal} {\bibinfo
  {journal} {Quantum Science and Technology}\ }\textbf {\bibinfo {volume}
  {6}},\ \bibinfo {pages} {035006} (\bibinfo {year} {2021})}\BibitemShut
  {NoStop}%
\bibitem [{\citenamefont {Uvarov}\ \emph {et~al.}(2020)\citenamefont {Uvarov},
  \citenamefont {Biamonte},\ and\ \citenamefont
  {Yudin}}]{uvarov2020variational}%
  \BibitemOpen
  \bibfield  {author} {\bibinfo {author} {\bibfnamefont {A.}~\bibnamefont
  {Uvarov}}, \bibinfo {author} {\bibfnamefont {J.~D.}\ \bibnamefont
  {Biamonte}},\ and\ \bibinfo {author} {\bibfnamefont {D.}~\bibnamefont
  {Yudin}},\ }\bibfield  {title} {\bibinfo {title} {Variational quantum
  eigensolver for frustrated quantum systems},\ }\href
  {https://doi.org/10.1103/PhysRevB.102.075104} {\bibfield  {journal} {\bibinfo
   {journal} {Physical Review B}\ }\textbf {\bibinfo {volume} {102}},\ \bibinfo
  {pages} {075104} (\bibinfo {year} {2020})}\BibitemShut {NoStop}%
\bibitem [{\citenamefont {Wang}\ \emph {et~al.}(2021)\citenamefont {Wang},
  \citenamefont {Fontana}, \citenamefont {Cerezo}, \citenamefont {Sharma},
  \citenamefont {Sone}, \citenamefont {Cincio},\ and\ \citenamefont
  {Coles}}]{wang2020noise}%
  \BibitemOpen
  \bibfield  {author} {\bibinfo {author} {\bibfnamefont {S.}~\bibnamefont
  {Wang}}, \bibinfo {author} {\bibfnamefont {E.}~\bibnamefont {Fontana}},
  \bibinfo {author} {\bibfnamefont {M.}~\bibnamefont {Cerezo}}, \bibinfo
  {author} {\bibfnamefont {K.}~\bibnamefont {Sharma}}, \bibinfo {author}
  {\bibfnamefont {A.}~\bibnamefont {Sone}}, \bibinfo {author} {\bibfnamefont
  {L.}~\bibnamefont {Cincio}},\ and\ \bibinfo {author} {\bibfnamefont {P.~J.}\
  \bibnamefont {Coles}},\ }\bibfield  {title} {\bibinfo {title} {Noise-induced
  barren plateaus in variational quantum algorithms},\ }\href
  {https://doi.org/10.1038/s41467-021-27045-6} {\bibfield  {journal} {\bibinfo
  {journal} {Nature {C}ommunications}\ }\textbf {\bibinfo {volume} {12}},\
  \bibinfo {pages} {1} (\bibinfo {year} {2021})}\BibitemShut {NoStop}%
\bibitem [{\citenamefont {Abbas}\ \emph {et~al.}(2021)\citenamefont {Abbas},
  \citenamefont {Sutter}, \citenamefont {Zoufal}, \citenamefont {Lucchi},
  \citenamefont {Figalli},\ and\ \citenamefont {Woerner}}]{abbas2020power}%
  \BibitemOpen
  \bibfield  {author} {\bibinfo {author} {\bibfnamefont {A.}~\bibnamefont
  {Abbas}}, \bibinfo {author} {\bibfnamefont {D.}~\bibnamefont {Sutter}},
  \bibinfo {author} {\bibfnamefont {C.}~\bibnamefont {Zoufal}}, \bibinfo
  {author} {\bibfnamefont {A.}~\bibnamefont {Lucchi}}, \bibinfo {author}
  {\bibfnamefont {A.}~\bibnamefont {Figalli}},\ and\ \bibinfo {author}
  {\bibfnamefont {S.}~\bibnamefont {Woerner}},\ }\bibfield  {title} {\bibinfo
  {title} {The power of quantum neural networks},\ }\href
  {https://doi.org/10.1038/s43588-021-00084-1} {\bibfield  {journal} {\bibinfo
  {journal} {Nature Computational Science}\ }\textbf {\bibinfo {volume} {1}},\
  \bibinfo {pages} {403} (\bibinfo {year} {2021})}\BibitemShut {NoStop}%
\bibitem [{\citenamefont {Arrasmith}\ \emph {et~al.}(2022)\citenamefont
  {Arrasmith}, \citenamefont {Holmes}, \citenamefont {Cerezo},\ and\
  \citenamefont {Coles}}]{arrasmith2021equivalence}%
  \BibitemOpen
  \bibfield  {author} {\bibinfo {author} {\bibfnamefont {A.}~\bibnamefont
  {Arrasmith}}, \bibinfo {author} {\bibfnamefont {Z.}~\bibnamefont {Holmes}},
  \bibinfo {author} {\bibfnamefont {M.}~\bibnamefont {Cerezo}},\ and\ \bibinfo
  {author} {\bibfnamefont {P.~J.}\ \bibnamefont {Coles}},\ }\bibfield  {title}
  {\bibinfo {title} {Equivalence of quantum barren plateaus to cost
  concentration and narrow gorges},\ }\href
  {https://doi.org/10.1088/2058-9565/ac7d06} {\bibfield  {journal} {\bibinfo
  {journal} {Quantum Science and Technology}\ }\textbf {\bibinfo {volume}
  {7}},\ \bibinfo {pages} {045015} (\bibinfo {year} {2022})}\BibitemShut
  {NoStop}%
\bibitem [{\citenamefont {Larocca}\ \emph {et~al.}(2022)\citenamefont
  {Larocca}, \citenamefont {Czarnik}, \citenamefont {Sharma}, \citenamefont
  {Muraleedharan}, \citenamefont {Coles},\ and\ \citenamefont
  {Cerezo}}]{larocca2021diagnosing}%
  \BibitemOpen
  \bibfield  {author} {\bibinfo {author} {\bibfnamefont {M.}~\bibnamefont
  {Larocca}}, \bibinfo {author} {\bibfnamefont {P.}~\bibnamefont {Czarnik}},
  \bibinfo {author} {\bibfnamefont {K.}~\bibnamefont {Sharma}}, \bibinfo
  {author} {\bibfnamefont {G.}~\bibnamefont {Muraleedharan}}, \bibinfo {author}
  {\bibfnamefont {P.~J.}\ \bibnamefont {Coles}},\ and\ \bibinfo {author}
  {\bibfnamefont {M.}~\bibnamefont {Cerezo}},\ }\bibfield  {title} {\bibinfo
  {title} {Diagnosing {B}arren {P}lateaus with {T}ools from {Q}uantum {O}ptimal
  {C}ontrol},\ }\href {https://doi.org/10.22331/q-2022-09-29-824} {\bibfield
  {journal} {\bibinfo  {journal} {{Quantum}}\ }\textbf {\bibinfo {volume}
  {6}},\ \bibinfo {pages} {824} (\bibinfo {year} {2022})}\BibitemShut {NoStop}%
\bibitem [{\citenamefont {Liu}\ \emph {et~al.}(2022)\citenamefont {Liu},
  \citenamefont {Yu}, \citenamefont {Duan},\ and\ \citenamefont
  {Deng}}]{liu2021presence}%
  \BibitemOpen
  \bibfield  {author} {\bibinfo {author} {\bibfnamefont {Z.}~\bibnamefont
  {Liu}}, \bibinfo {author} {\bibfnamefont {L.-W.}\ \bibnamefont {Yu}},
  \bibinfo {author} {\bibfnamefont {L.-M.}\ \bibnamefont {Duan}},\ and\
  \bibinfo {author} {\bibfnamefont {D.-L.}\ \bibnamefont {Deng}},\ }\bibfield
  {title} {\bibinfo {title} {The presence and absence of barren plateaus in
  tensor-network based machine learning},\ }\href
  {https://doi.org/10.1103/PhysRevLett.129.270501} {\bibfield  {journal}
  {\bibinfo  {journal} {Physical Review Letters}\ }\textbf {\bibinfo {volume}
  {129}},\ \bibinfo {pages} {270501} (\bibinfo {year} {2022})}\BibitemShut
  {NoStop}%
\bibitem [{\citenamefont {Holmes}\ \emph {et~al.}(2022)\citenamefont {Holmes},
  \citenamefont {Sharma}, \citenamefont {Cerezo},\ and\ \citenamefont
  {Coles}}]{holmes2021connecting}%
  \BibitemOpen
  \bibfield  {author} {\bibinfo {author} {\bibfnamefont {Z.}~\bibnamefont
  {Holmes}}, \bibinfo {author} {\bibfnamefont {K.}~\bibnamefont {Sharma}},
  \bibinfo {author} {\bibfnamefont {M.}~\bibnamefont {Cerezo}},\ and\ \bibinfo
  {author} {\bibfnamefont {P.~J.}\ \bibnamefont {Coles}},\ }\bibfield  {title}
  {\bibinfo {title} {Connecting ansatz expressibility to gradient magnitudes
  and barren plateaus},\ }\href {https://doi.org/10.1103/PRXQuantum.3.010313}
  {\bibfield  {journal} {\bibinfo  {journal} {PRX Quantum}\ }\textbf {\bibinfo
  {volume} {3}},\ \bibinfo {pages} {010313} (\bibinfo {year}
  {2022})}\BibitemShut {NoStop}%
\bibitem [{\citenamefont {Zhao}\ and\ \citenamefont
  {Gao}(2021)}]{zhao2021analyzing}%
  \BibitemOpen
  \bibfield  {author} {\bibinfo {author} {\bibfnamefont {C.}~\bibnamefont
  {Zhao}}\ and\ \bibinfo {author} {\bibfnamefont {X.-S.}\ \bibnamefont {Gao}},\
  }\bibfield  {title} {\bibinfo {title} {Analyzing the barren plateau
  phenomenon in training quantum neural networks with the {ZX}-calculus},\
  }\href {https://doi.org/10.22331/q-2021-06-04-466} {\bibfield  {journal}
  {\bibinfo  {journal} {{Quantum}}\ }\textbf {\bibinfo {volume} {5}},\ \bibinfo
  {pages} {466} (\bibinfo {year} {2021})}\BibitemShut {NoStop}%
\bibitem [{\citenamefont {Kieferova}\ \emph {et~al.}(2021)\citenamefont
  {Kieferova}, \citenamefont {Carlos},\ and\ \citenamefont
  {Wiebe}}]{kieferova2021quantum}%
  \BibitemOpen
  \bibfield  {author} {\bibinfo {author} {\bibfnamefont {M.}~\bibnamefont
  {Kieferova}}, \bibinfo {author} {\bibfnamefont {O.~M.}\ \bibnamefont
  {Carlos}},\ and\ \bibinfo {author} {\bibfnamefont {N.}~\bibnamefont
  {Wiebe}},\ }\bibfield  {title} {\bibinfo {title} {Quantum generative training
  using r\'{e}nyi divergences},\ }\href {https://arxiv.org/abs/2106.09567}
  {\bibfield  {journal} {\bibinfo  {journal} {arXiv preprint arXiv:2106.09567}\
  } (\bibinfo {year} {2021})}\BibitemShut {NoStop}%
\bibitem [{\citenamefont {Thanaslip}\ \emph {et~al.}(2023)\citenamefont
  {Thanaslip}, \citenamefont {Wang}, \citenamefont {Nghiem}, \citenamefont
  {Coles},\ and\ \citenamefont {Cerezo}}]{thanaslip2021subtleties}%
  \BibitemOpen
  \bibfield  {author} {\bibinfo {author} {\bibfnamefont {S.}~\bibnamefont
  {Thanaslip}}, \bibinfo {author} {\bibfnamefont {S.}~\bibnamefont {Wang}},
  \bibinfo {author} {\bibfnamefont {N.~A.}\ \bibnamefont {Nghiem}}, \bibinfo
  {author} {\bibfnamefont {P.~J.}\ \bibnamefont {Coles}},\ and\ \bibinfo
  {author} {\bibfnamefont {M.}~\bibnamefont {Cerezo}},\ }\bibfield  {title}
  {\bibinfo {title} {Subtleties in the trainability of quantum machine learning
  models},\ }\href {https://doi.org/10.1007/s42484-023-00103-6} {\bibfield
  {journal} {\bibinfo  {journal} {Quantum Machine Intelligence}\ }\textbf
  {\bibinfo {volume} {5}},\ \bibinfo {pages} {21} (\bibinfo {year}
  {2023})}\BibitemShut {NoStop}%
\bibitem [{\citenamefont {Lee}\ \emph {et~al.}(2021)\citenamefont {Lee},
  \citenamefont {Magann}, \citenamefont {Rabitz},\ and\ \citenamefont
  {Arenz}}]{lee2021towards}%
  \BibitemOpen
  \bibfield  {author} {\bibinfo {author} {\bibfnamefont {J.}~\bibnamefont
  {Lee}}, \bibinfo {author} {\bibfnamefont {A.~B.}\ \bibnamefont {Magann}},
  \bibinfo {author} {\bibfnamefont {H.~A.}\ \bibnamefont {Rabitz}},\ and\
  \bibinfo {author} {\bibfnamefont {C.}~\bibnamefont {Arenz}},\ }\bibfield
  {title} {\bibinfo {title} {Progress toward favorable landscapes in quantum
  combinatorial optimization},\ }\href
  {https://doi.org/10.1103/PhysRevA.104.032401} {\bibfield  {journal} {\bibinfo
   {journal} {Physical Review A}\ }\textbf {\bibinfo {volume} {104}},\ \bibinfo
  {pages} {032401} (\bibinfo {year} {2021})}\BibitemShut {NoStop}%
\bibitem [{\citenamefont {Shaydulin}\ and\ \citenamefont
  {Wild}(2022)}]{shaydulin2021importance}%
  \BibitemOpen
  \bibfield  {author} {\bibinfo {author} {\bibfnamefont {R.}~\bibnamefont
  {Shaydulin}}\ and\ \bibinfo {author} {\bibfnamefont {S.~M.}\ \bibnamefont
  {Wild}},\ }\bibfield  {title} {\bibinfo {title} {Importance of kernel
  bandwidth in quantum machine learning},\ }\href
  {https://doi.org/10.1103/PhysRevA.106.042407} {\bibfield  {journal} {\bibinfo
   {journal} {Physical Review A}\ }\textbf {\bibinfo {volume} {106}},\ \bibinfo
  {pages} {042407} (\bibinfo {year} {2022})}\BibitemShut {NoStop}%
\bibitem [{\citenamefont {Holmes}\ \emph {et~al.}(2021)\citenamefont {Holmes},
  \citenamefont {Arrasmith}, \citenamefont {Yan}, \citenamefont {Coles},
  \citenamefont {Albrecht},\ and\ \citenamefont
  {Sornborger}}]{holmes2021barren}%
  \BibitemOpen
  \bibfield  {author} {\bibinfo {author} {\bibfnamefont {Z.}~\bibnamefont
  {Holmes}}, \bibinfo {author} {\bibfnamefont {A.}~\bibnamefont {Arrasmith}},
  \bibinfo {author} {\bibfnamefont {B.}~\bibnamefont {Yan}}, \bibinfo {author}
  {\bibfnamefont {P.~J.}\ \bibnamefont {Coles}}, \bibinfo {author}
  {\bibfnamefont {A.}~\bibnamefont {Albrecht}},\ and\ \bibinfo {author}
  {\bibfnamefont {A.~T.}\ \bibnamefont {Sornborger}},\ }\bibfield  {title}
  {\bibinfo {title} {Barren plateaus preclude learning scramblers},\ }\href
  {https://doi.org/10.1103/PhysRevLett.126.190501} {\bibfield  {journal}
  {\bibinfo  {journal} {Physical Review Letters}\ }\textbf {\bibinfo {volume}
  {126}},\ \bibinfo {pages} {190501} (\bibinfo {year} {2021})}\BibitemShut
  {NoStop}%
\bibitem [{\citenamefont {Leadbeater}\ \emph {et~al.}(2021)\citenamefont
  {Leadbeater}, \citenamefont {Sharrock}, \citenamefont {Coyle},\ and\
  \citenamefont {Benedetti}}]{leadbeater2021f}%
  \BibitemOpen
  \bibfield  {author} {\bibinfo {author} {\bibfnamefont {C.}~\bibnamefont
  {Leadbeater}}, \bibinfo {author} {\bibfnamefont {L.}~\bibnamefont
  {Sharrock}}, \bibinfo {author} {\bibfnamefont {B.}~\bibnamefont {Coyle}},\
  and\ \bibinfo {author} {\bibfnamefont {M.}~\bibnamefont {Benedetti}},\
  }\bibfield  {title} {\bibinfo {title} {F-divergences and cost function
  locality in generative modelling with quantum circuits},\ }\href
  {https://doi.org/10.3390/e23101281} {\bibfield  {journal} {\bibinfo
  {journal} {Entropy}\ }\textbf {\bibinfo {volume} {23}},\ \bibinfo {pages}
  {1281} (\bibinfo {year} {2021})}\BibitemShut {NoStop}%
\bibitem [{\citenamefont {Mart{\'\i}n}\ \emph {et~al.}(2022)\citenamefont
  {Mart{\'\i}n}, \citenamefont {Plekhanov},\ and\ \citenamefont
  {Lubasch}}]{martin2022barren}%
  \BibitemOpen
  \bibfield  {author} {\bibinfo {author} {\bibfnamefont {E.~C.}\ \bibnamefont
  {Mart{\'\i}n}}, \bibinfo {author} {\bibfnamefont {K.}~\bibnamefont
  {Plekhanov}},\ and\ \bibinfo {author} {\bibfnamefont {M.}~\bibnamefont
  {Lubasch}},\ }\bibfield  {title} {\bibinfo {title} {Barren plateaus in
  quantum tensor network optimization},\ }\href
  {https://arxiv.org/abs/2209.00292} {\bibfield  {journal} {\bibinfo  {journal}
  {arXiv preprint arXiv:2209.00292}\ } (\bibinfo {year} {2022})}\BibitemShut
  {NoStop}%
\bibitem [{\citenamefont {Grimsley}\ \emph {et~al.}(2023)\citenamefont
  {Grimsley}, \citenamefont {Mayhall}, \citenamefont {Barron}, \citenamefont
  {Barnes},\ and\ \citenamefont {Economou}}]{grimsley2022adapt}%
  \BibitemOpen
  \bibfield  {author} {\bibinfo {author} {\bibfnamefont {H.~R.}\ \bibnamefont
  {Grimsley}}, \bibinfo {author} {\bibfnamefont {N.~J.}\ \bibnamefont
  {Mayhall}}, \bibinfo {author} {\bibfnamefont {G.~S.}\ \bibnamefont {Barron}},
  \bibinfo {author} {\bibfnamefont {E.}~\bibnamefont {Barnes}},\ and\ \bibinfo
  {author} {\bibfnamefont {S.~E.}\ \bibnamefont {Economou}},\ }\bibfield
  {title} {\bibinfo {title} {Adaptive, problem-tailored variational quantum
  eigensolver mitigates rough parameter landscapes and barren plateaus},\
  }\href {https://doi.org/10.1038/s41534-023-00681-0} {\bibfield  {journal}
  {\bibinfo  {journal} {npj Quantum Information}\ }\textbf {\bibinfo {volume}
  {9}},\ \bibinfo {pages} {19} (\bibinfo {year} {2023})}\BibitemShut {NoStop}%
\bibitem [{\citenamefont {Leone}\ \emph {et~al.}(2022)\citenamefont {Leone},
  \citenamefont {Oliviero}, \citenamefont {Cincio},\ and\ \citenamefont
  {Cerezo}}]{leone2022practical}%
  \BibitemOpen
  \bibfield  {author} {\bibinfo {author} {\bibfnamefont {L.}~\bibnamefont
  {Leone}}, \bibinfo {author} {\bibfnamefont {S.~F.}\ \bibnamefont {Oliviero}},
  \bibinfo {author} {\bibfnamefont {L.}~\bibnamefont {Cincio}},\ and\ \bibinfo
  {author} {\bibfnamefont {M.}~\bibnamefont {Cerezo}},\ }\bibfield  {title}
  {\bibinfo {title} {On the practical usefulness of the hardware efficient
  ansatz},\ }\href {https://arxiv.org/abs/2211.01477} {\bibfield  {journal}
  {\bibinfo  {journal} {arXiv preprint arXiv:2211.01477}\ } (\bibinfo {year}
  {2022})}\BibitemShut {NoStop}%
\bibitem [{\citenamefont {Sack}\ \emph {et~al.}(2022)\citenamefont {Sack},
  \citenamefont {Medina}, \citenamefont {Michailidis}, \citenamefont {Kueng},\
  and\ \citenamefont {Serbyn}}]{sack2022avoiding}%
  \BibitemOpen
  \bibfield  {author} {\bibinfo {author} {\bibfnamefont {S.~H.}\ \bibnamefont
  {Sack}}, \bibinfo {author} {\bibfnamefont {R.~A.}\ \bibnamefont {Medina}},
  \bibinfo {author} {\bibfnamefont {A.~A.}\ \bibnamefont {Michailidis}},
  \bibinfo {author} {\bibfnamefont {R.}~\bibnamefont {Kueng}},\ and\ \bibinfo
  {author} {\bibfnamefont {M.}~\bibnamefont {Serbyn}},\ }\bibfield  {title}
  {\bibinfo {title} {Avoiding barren plateaus using classical shadows},\ }\href
  {https://doi.org/10.1103/PRXQuantum.3.020365} {\bibfield  {journal} {\bibinfo
   {journal} {PRX Quantum}\ }\textbf {\bibinfo {volume} {3}},\ \bibinfo {pages}
  {020365} (\bibinfo {year} {2022})}\BibitemShut {NoStop}%
\bibitem [{\citenamefont {Kashif}\ and\ \citenamefont
  {Al-Kuwari}(2023{\natexlab{a}})}]{kashif2023impact}%
  \BibitemOpen
  \bibfield  {author} {\bibinfo {author} {\bibfnamefont {M.}~\bibnamefont
  {Kashif}}\ and\ \bibinfo {author} {\bibfnamefont {S.}~\bibnamefont
  {Al-Kuwari}},\ }\bibfield  {title} {\bibinfo {title} {The impact of cost
  function globality and locality in hybrid quantum neural networks on nisq
  devices},\ }\href {https://doi.org/10.1088/2632-2153/acb12f} {\bibfield
  {journal} {\bibinfo  {journal} {Machine Learning: Science and Technology}\
  }\textbf {\bibinfo {volume} {4}},\ \bibinfo {pages} {015004} (\bibinfo {year}
  {2023}{\natexlab{a}})}\BibitemShut {NoStop}%
\bibitem [{\citenamefont {Friedrich}\ and\ \citenamefont
  {Maziero}(2023)}]{friedrich2023quantum}%
  \BibitemOpen
  \bibfield  {author} {\bibinfo {author} {\bibfnamefont {L.}~\bibnamefont
  {Friedrich}}\ and\ \bibinfo {author} {\bibfnamefont {J.}~\bibnamefont
  {Maziero}},\ }\bibfield  {title} {\bibinfo {title} {Quantum neural network
  cost function concentration dependency on the parametrization expressivity},\
  }\href {https://doi.org/10.1038/s41598-023-37003-5} {\bibfield  {journal}
  {\bibinfo  {journal} {Scientific Reports}\ }\textbf {\bibinfo {volume}
  {13}},\ \bibinfo {pages} {9978} (\bibinfo {year} {2023})}\BibitemShut
  {NoStop}%
\bibitem [{\citenamefont {Garc{\'\i}a-Mart{\'\i}n}\ \emph
  {et~al.}(2023)\citenamefont {Garc{\'\i}a-Mart{\'\i}n}, \citenamefont
  {Larocca},\ and\ \citenamefont {Cerezo}}]{garcia2023deep}%
  \BibitemOpen
  \bibfield  {author} {\bibinfo {author} {\bibfnamefont {D.}~\bibnamefont
  {Garc{\'\i}a-Mart{\'\i}n}}, \bibinfo {author} {\bibfnamefont
  {M.}~\bibnamefont {Larocca}},\ and\ \bibinfo {author} {\bibfnamefont
  {M.}~\bibnamefont {Cerezo}},\ }\bibfield  {title} {\bibinfo {title} {Deep
  quantum neural networks form gaussian processes},\ }\href
  {https://arxiv.org/abs/2305.09957} {\bibfield  {journal} {\bibinfo  {journal}
  {arXiv preprint arXiv:2305.09957}\ } (\bibinfo {year} {2023})}\BibitemShut
  {NoStop}%
\bibitem [{\citenamefont {Kulshrestha}\ and\ \citenamefont
  {Safro}(2022)}]{kulshrestha2022beinit}%
  \BibitemOpen
  \bibfield  {author} {\bibinfo {author} {\bibfnamefont {A.}~\bibnamefont
  {Kulshrestha}}\ and\ \bibinfo {author} {\bibfnamefont {I.}~\bibnamefont
  {Safro}},\ }\bibfield  {title} {\bibinfo {title} {Beinit: Avoiding barren
  plateaus in variational quantum algorithms},\ }in\ \href
  {https://doi.org/10.1109/QCE53715.2022.00039} {\emph {\bibinfo {booktitle}
  {2022 IEEE International Conference on Quantum Computing and Engineering
  (QCE)}}}\ (\bibinfo {organization} {IEEE},\ \bibinfo {year} {2022})\ pp.\
  \bibinfo {pages} {197--203}\BibitemShut {NoStop}%
\bibitem [{\citenamefont {Volkoff}(2021)}]{volkoff2021efficient}%
  \BibitemOpen
  \bibfield  {author} {\bibinfo {author} {\bibfnamefont {T.~J.}\ \bibnamefont
  {Volkoff}},\ }\bibfield  {title} {\bibinfo {title} {Efficient trainability of
  linear optical modules in quantum optical neural networks},\ }\href
  {https://doi.org/10.1007/s10946-021-09958-1} {\bibfield  {journal} {\bibinfo
  {journal} {Journal of Russian Laser Research}\ }\textbf {\bibinfo {volume}
  {42}},\ \bibinfo {pages} {250} (\bibinfo {year} {2021})}\BibitemShut
  {NoStop}%
\bibitem [{\citenamefont {Kashif}\ and\ \citenamefont
  {Al-Kuwari}(2023{\natexlab{b}})}]{kashif2023unified}%
  \BibitemOpen
  \bibfield  {author} {\bibinfo {author} {\bibfnamefont {M.}~\bibnamefont
  {Kashif}}\ and\ \bibinfo {author} {\bibfnamefont {S.}~\bibnamefont
  {Al-Kuwari}},\ }\bibfield  {title} {\bibinfo {title} {The unified effect of
  data encoding, ansatz expressibility and entanglement on the trainability of
  hqnns},\ }\href {https://doi.org/10.1080/17445760.2023.2231163} {\bibfield
  {journal} {\bibinfo  {journal} {International Journal of Parallel, Emergent
  and Distributed Systems}\ }\textbf {\bibinfo {volume} {38}},\ \bibinfo
  {pages} {362} (\bibinfo {year} {2023}{\natexlab{b}})}\BibitemShut {NoStop}%
\bibitem [{\citenamefont {Monbroussou}\ \emph {et~al.}(2023)\citenamefont
  {Monbroussou}, \citenamefont {Landman}, \citenamefont {Grilo}, \citenamefont
  {Kukla},\ and\ \citenamefont {Kashefi}}]{monbroussou2023trainability}%
  \BibitemOpen
  \bibfield  {author} {\bibinfo {author} {\bibfnamefont {L.}~\bibnamefont
  {Monbroussou}}, \bibinfo {author} {\bibfnamefont {J.}~\bibnamefont
  {Landman}}, \bibinfo {author} {\bibfnamefont {A.~B.}\ \bibnamefont {Grilo}},
  \bibinfo {author} {\bibfnamefont {R.}~\bibnamefont {Kukla}},\ and\ \bibinfo
  {author} {\bibfnamefont {E.}~\bibnamefont {Kashefi}},\ }\bibfield  {title}
  {\bibinfo {title} {Trainability and expressivity of hamming-weight preserving
  quantum circuits for machine learning},\ }\href
  {https://arxiv.org/abs/2309.15547} {\bibfield  {journal} {\bibinfo  {journal}
  {arXiv preprint arXiv:2309.15547}\ } (\bibinfo {year} {2023})}\BibitemShut
  {NoStop}%
\bibitem [{\citenamefont {Fontana}\ \emph {et~al.}(2023)\citenamefont
  {Fontana}, \citenamefont {Herman}, \citenamefont {Chakrabarti}, \citenamefont
  {Kumar}, \citenamefont {Yalovetzky}, \citenamefont {Heredge}, \citenamefont
  {Hari~Sureshbabu},\ and\ \citenamefont {Pistoia}}]{fontana2023theadjoint}%
  \BibitemOpen
  \bibfield  {author} {\bibinfo {author} {\bibfnamefont {E.}~\bibnamefont
  {Fontana}}, \bibinfo {author} {\bibfnamefont {D.}~\bibnamefont {Herman}},
  \bibinfo {author} {\bibfnamefont {S.}~\bibnamefont {Chakrabarti}}, \bibinfo
  {author} {\bibfnamefont {N.}~\bibnamefont {Kumar}}, \bibinfo {author}
  {\bibfnamefont {R.}~\bibnamefont {Yalovetzky}}, \bibinfo {author}
  {\bibfnamefont {J.}~\bibnamefont {Heredge}}, \bibinfo {author} {\bibfnamefont
  {S.}~\bibnamefont {Hari~Sureshbabu}},\ and\ \bibinfo {author} {\bibfnamefont
  {M.}~\bibnamefont {Pistoia}},\ }\bibfield  {title} {\bibinfo {title} {The
  adjoint is all you need: Characterizing barren plateaus in quantum
  ans\"atze},\ }\href {https://arxiv.org/abs/2309.07902} {\bibfield  {journal}
  {\bibinfo  {journal} {arXiv preprint arXiv:2309.07902}\ } (\bibinfo {year}
  {2023})}\BibitemShut {NoStop}%
\bibitem [{\citenamefont {Ragone}\ \emph {et~al.}(2023)\citenamefont {Ragone},
  \citenamefont {Bakalov}, \citenamefont {Sauvage}, \citenamefont {Kemper},
  \citenamefont {Marrero}, \citenamefont {Larocca},\ and\ \citenamefont
  {Cerezo}}]{ragone2023unified}%
  \BibitemOpen
  \bibfield  {author} {\bibinfo {author} {\bibfnamefont {M.}~\bibnamefont
  {Ragone}}, \bibinfo {author} {\bibfnamefont {B.~N.}\ \bibnamefont {Bakalov}},
  \bibinfo {author} {\bibfnamefont {F.}~\bibnamefont {Sauvage}}, \bibinfo
  {author} {\bibfnamefont {A.~F.}\ \bibnamefont {Kemper}}, \bibinfo {author}
  {\bibfnamefont {C.~O.}\ \bibnamefont {Marrero}}, \bibinfo {author}
  {\bibfnamefont {M.}~\bibnamefont {Larocca}},\ and\ \bibinfo {author}
  {\bibfnamefont {M.}~\bibnamefont {Cerezo}},\ }\bibfield  {title} {\bibinfo
  {title} {A unified theory of barren plateaus for deep parametrized quantum
  circuits},\ }\href {https://arxiv.org/abs/2309.09342} {\bibfield  {journal}
  {\bibinfo  {journal} {arXiv preprint arXiv:2309.09342}\ } (\bibinfo {year}
  {2023})}\BibitemShut {NoStop}%
\bibitem [{\citenamefont {Jozsa}\ and\ \citenamefont
  {Miyake}(2008)}]{jozsa2008matchgates}%
  \BibitemOpen
  \bibfield  {author} {\bibinfo {author} {\bibfnamefont {R.}~\bibnamefont
  {Jozsa}}\ and\ \bibinfo {author} {\bibfnamefont {A.}~\bibnamefont {Miyake}},\
  }\bibfield  {title} {\bibinfo {title} {Matchgates and classical simulation of
  quantum circuits},\ }\href {https://doi.org/10.1098/rspa.2008.0189}
  {\bibfield  {journal} {\bibinfo  {journal} {Proceedings of the Royal Society
  A: Mathematical, Physical and Engineering Sciences}\ }\textbf {\bibinfo
  {volume} {464}},\ \bibinfo {pages} {3089} (\bibinfo {year}
  {2008})}\BibitemShut {NoStop}%
\bibitem [{\citenamefont {Wan}\ \emph {et~al.}(2022)\citenamefont {Wan},
  \citenamefont {Huggins}, \citenamefont {Lee},\ and\ \citenamefont
  {Babbush}}]{wan2022matchgate}%
  \BibitemOpen
  \bibfield  {author} {\bibinfo {author} {\bibfnamefont {K.}~\bibnamefont
  {Wan}}, \bibinfo {author} {\bibfnamefont {W.~J.}\ \bibnamefont {Huggins}},
  \bibinfo {author} {\bibfnamefont {J.}~\bibnamefont {Lee}},\ and\ \bibinfo
  {author} {\bibfnamefont {R.}~\bibnamefont {Babbush}},\ }\bibfield  {title}
  {\bibinfo {title} {Matchgate shadows for fermionic quantum simulation},\
  }\href {https://arxiv.org/abs/2207.13723} {\bibfield  {journal} {\bibinfo
  {journal} {arXiv preprint arXiv:2207.13723}\ } (\bibinfo {year}
  {2022})}\BibitemShut {NoStop}%
\bibitem [{\citenamefont {De~Melo}\ \emph {et~al.}(2013)\citenamefont
  {De~Melo}, \citenamefont {{\'C}wikli{\'n}ski},\ and\ \citenamefont
  {Terhal}}]{de2013power}%
  \BibitemOpen
  \bibfield  {author} {\bibinfo {author} {\bibfnamefont {F.}~\bibnamefont
  {De~Melo}}, \bibinfo {author} {\bibfnamefont {P.}~\bibnamefont
  {{\'C}wikli{\'n}ski}},\ and\ \bibinfo {author} {\bibfnamefont {B.~M.}\
  \bibnamefont {Terhal}},\ }\bibfield  {title} {\bibinfo {title} {The power of
  noisy fermionic quantum computation},\ }\href
  {https://doi.org/https://doi.org/10.1088/1367-2630/15/1/013015} {\bibfield
  {journal} {\bibinfo  {journal} {New Journal of Physics}\ }\textbf {\bibinfo
  {volume} {15}},\ \bibinfo {pages} {013015} (\bibinfo {year}
  {2013})}\BibitemShut {NoStop}%
\bibitem [{\citenamefont {Oszmaniec}\ \emph {et~al.}(2022)\citenamefont
  {Oszmaniec}, \citenamefont {Dangniam}, \citenamefont {Morales},\ and\
  \citenamefont {Zimbor{\'a}s}}]{oszmaniec2022fermion}%
  \BibitemOpen
  \bibfield  {author} {\bibinfo {author} {\bibfnamefont {M.}~\bibnamefont
  {Oszmaniec}}, \bibinfo {author} {\bibfnamefont {N.}~\bibnamefont {Dangniam}},
  \bibinfo {author} {\bibfnamefont {M.~E.}\ \bibnamefont {Morales}},\ and\
  \bibinfo {author} {\bibfnamefont {Z.}~\bibnamefont {Zimbor{\'a}s}},\
  }\bibfield  {title} {\bibinfo {title} {Fermion sampling: a robust quantum
  computational advantage scheme using fermionic linear optics and magic input
  states},\ }\href {https://doi.org/10.1103/PRXQuantum.3.020328} {\bibfield
  {journal} {\bibinfo  {journal} {PRX Quantum}\ }\textbf {\bibinfo {volume}
  {3}},\ \bibinfo {pages} {020328} (\bibinfo {year} {2022})}\BibitemShut
  {NoStop}%
\bibitem [{\citenamefont {Cherrat}\ \emph {et~al.}(2023)\citenamefont
  {Cherrat}, \citenamefont {Raj}, \citenamefont {Kerenidis}, \citenamefont
  {Shekhar}, \citenamefont {Wood}, \citenamefont {Dee}, \citenamefont
  {Chakrabarti}, \citenamefont {Chen}, \citenamefont {Herman}, \citenamefont
  {Hu} \emph {et~al.}}]{cherrat2023quantum}%
  \BibitemOpen
  \bibfield  {author} {\bibinfo {author} {\bibfnamefont {E.~A.}\ \bibnamefont
  {Cherrat}}, \bibinfo {author} {\bibfnamefont {S.}~\bibnamefont {Raj}},
  \bibinfo {author} {\bibfnamefont {I.}~\bibnamefont {Kerenidis}}, \bibinfo
  {author} {\bibfnamefont {A.}~\bibnamefont {Shekhar}}, \bibinfo {author}
  {\bibfnamefont {B.}~\bibnamefont {Wood}}, \bibinfo {author} {\bibfnamefont
  {J.}~\bibnamefont {Dee}}, \bibinfo {author} {\bibfnamefont {S.}~\bibnamefont
  {Chakrabarti}}, \bibinfo {author} {\bibfnamefont {R.}~\bibnamefont {Chen}},
  \bibinfo {author} {\bibfnamefont {D.}~\bibnamefont {Herman}}, \bibinfo
  {author} {\bibfnamefont {S.}~\bibnamefont {Hu}}, \emph {et~al.},\ }\bibfield
  {title} {\bibinfo {title} {Quantum deep hedging},\ }\href
  {https://arxiv.org/abs/2303.16585} {\bibfield  {journal} {\bibinfo  {journal}
  {arXiv preprint arXiv:2303.16585}\ } (\bibinfo {year} {2023})}\BibitemShut
  {NoStop}%
\bibitem [{\citenamefont {Gigena}\ and\ \citenamefont
  {Rossignoli}(2015)}]{gigena2015entanglement}%
  \BibitemOpen
  \bibfield  {author} {\bibinfo {author} {\bibfnamefont {N.}~\bibnamefont
  {Gigena}}\ and\ \bibinfo {author} {\bibfnamefont {R.}~\bibnamefont
  {Rossignoli}},\ }\bibfield  {title} {\bibinfo {title} {Entanglement in
  fermion systems},\ }\href {https://doi.org/10.1103/PhysRevA.92.042326}
  {\bibfield  {journal} {\bibinfo  {journal} {Physical Review A}\ }\textbf
  {\bibinfo {volume} {92}},\ \bibinfo {pages} {042326} (\bibinfo {year}
  {2015})}\BibitemShut {NoStop}%
\bibitem [{\citenamefont {Gigena}\ \emph {et~al.}(2020)\citenamefont {Gigena},
  \citenamefont {Di~Tullio},\ and\ \citenamefont {Rossignoli}}]{gigena2020one}%
  \BibitemOpen
  \bibfield  {author} {\bibinfo {author} {\bibfnamefont {N.}~\bibnamefont
  {Gigena}}, \bibinfo {author} {\bibfnamefont {M.}~\bibnamefont {Di~Tullio}},\
  and\ \bibinfo {author} {\bibfnamefont {R.}~\bibnamefont {Rossignoli}},\
  }\bibfield  {title} {\bibinfo {title} {One-body entanglement as a quantum
  resource in fermionic systems},\ }\href
  {https://doi.org/10.1103/PhysRevA.102.042410} {\bibfield  {journal} {\bibinfo
   {journal} {Physical Review A}\ }\textbf {\bibinfo {volume} {102}},\ \bibinfo
  {pages} {042410} (\bibinfo {year} {2020})}\BibitemShut {NoStop}%
\bibitem [{\citenamefont {Gigena}\ \emph {et~al.}(2021)\citenamefont {Gigena},
  \citenamefont {Di~Tullio},\ and\ \citenamefont
  {Rossignoli}}]{gigena2021many}%
  \BibitemOpen
  \bibfield  {author} {\bibinfo {author} {\bibfnamefont {N.}~\bibnamefont
  {Gigena}}, \bibinfo {author} {\bibfnamefont {M.}~\bibnamefont {Di~Tullio}},\
  and\ \bibinfo {author} {\bibfnamefont {R.}~\bibnamefont {Rossignoli}},\
  }\bibfield  {title} {\bibinfo {title} {Many-body entanglement in fermion
  systems},\ }\href {https://doi.org/10.1103/PhysRevA.103.052424} {\bibfield
  {journal} {\bibinfo  {journal} {Physical Review A}\ }\textbf {\bibinfo
  {volume} {103}},\ \bibinfo {pages} {052424} (\bibinfo {year}
  {2021})}\BibitemShut {NoStop}%
\bibitem [{\citenamefont {K{\"o}kc{\"u}}\ \emph {et~al.}(2022)\citenamefont
  {K{\"o}kc{\"u}}, \citenamefont {Steckmann}, \citenamefont {Wang},
  \citenamefont {Freericks}, \citenamefont {Dumitrescu},\ and\ \citenamefont
  {Kemper}}]{kokcu2022fixed}%
  \BibitemOpen
  \bibfield  {author} {\bibinfo {author} {\bibfnamefont {E.}~\bibnamefont
  {K{\"o}kc{\"u}}}, \bibinfo {author} {\bibfnamefont {T.}~\bibnamefont
  {Steckmann}}, \bibinfo {author} {\bibfnamefont {Y.}~\bibnamefont {Wang}},
  \bibinfo {author} {\bibfnamefont {J.}~\bibnamefont {Freericks}}, \bibinfo
  {author} {\bibfnamefont {E.~F.}\ \bibnamefont {Dumitrescu}},\ and\ \bibinfo
  {author} {\bibfnamefont {A.~F.}\ \bibnamefont {Kemper}},\ }\bibfield  {title}
  {\bibinfo {title} {Fixed depth hamiltonian simulation via cartan
  decomposition},\ }\href {https://doi.org/10.1103/PhysRevLett.129.070501}
  {\bibfield  {journal} {\bibinfo  {journal} {Physical Review Letters}\
  }\textbf {\bibinfo {volume} {129}},\ \bibinfo {pages} {070501} (\bibinfo
  {year} {2022})}\BibitemShut {NoStop}%
\bibitem [{\citenamefont {Wiersema}\ \emph {et~al.}(2023)\citenamefont
  {Wiersema}, \citenamefont {K{\"o}kc{\"u}}, \citenamefont {Kemper},\ and\
  \citenamefont {Bakalov}}]{wiersema2023classification}%
  \BibitemOpen
  \bibfield  {author} {\bibinfo {author} {\bibfnamefont {R.}~\bibnamefont
  {Wiersema}}, \bibinfo {author} {\bibfnamefont {E.}~\bibnamefont
  {K{\"o}kc{\"u}}}, \bibinfo {author} {\bibfnamefont {A.~F.}\ \bibnamefont
  {Kemper}},\ and\ \bibinfo {author} {\bibfnamefont {B.~N.}\ \bibnamefont
  {Bakalov}},\ }\bibfield  {title} {\bibinfo {title} {Classification of
  dynamical lie algebras for translation-invariant 2-local spin systems in one
  dimension},\ }\href {https://arxiv.org/abs/2309.05690} {\bibfield  {journal}
  {\bibinfo  {journal} {arXiv preprint arXiv:2309.05690}\ } (\bibinfo {year}
  {2023})}\BibitemShut {NoStop}%
\bibitem [{\citenamefont {Lieb}\ \emph {et~al.}(1961)\citenamefont {Lieb},
  \citenamefont {Schultz},\ and\ \citenamefont {Mattis}}]{lieb1961two}%
  \BibitemOpen
  \bibfield  {author} {\bibinfo {author} {\bibfnamefont {E.}~\bibnamefont
  {Lieb}}, \bibinfo {author} {\bibfnamefont {T.}~\bibnamefont {Schultz}},\ and\
  \bibinfo {author} {\bibfnamefont {D.}~\bibnamefont {Mattis}},\ }\bibfield
  {title} {\bibinfo {title} {Two soluble models of an antiferromagnetic
  chain},\ }\href {https://doi.org/10.1016/0003-4916(61)90115-4} {\bibfield
  {journal} {\bibinfo  {journal} {Annals of Physics}\ }\textbf {\bibinfo
  {volume} {16}},\ \bibinfo {pages} {407} (\bibinfo {year} {1961})}\BibitemShut
  {NoStop}%
\bibitem [{\citenamefont {Ragone}\ \emph {et~al.}(2022)\citenamefont {Ragone},
  \citenamefont {Nguyen}, \citenamefont {Schatzki}, \citenamefont {Braccia},
  \citenamefont {Larocca}, \citenamefont {Sauvage}, \citenamefont {Coles},\
  and\ \citenamefont {Cerezo}}]{ragone2022representation}%
  \BibitemOpen
  \bibfield  {author} {\bibinfo {author} {\bibfnamefont {M.}~\bibnamefont
  {Ragone}}, \bibinfo {author} {\bibfnamefont {Q.~T.}\ \bibnamefont {Nguyen}},
  \bibinfo {author} {\bibfnamefont {L.}~\bibnamefont {Schatzki}}, \bibinfo
  {author} {\bibfnamefont {P.}~\bibnamefont {Braccia}}, \bibinfo {author}
  {\bibfnamefont {M.}~\bibnamefont {Larocca}}, \bibinfo {author} {\bibfnamefont
  {F.}~\bibnamefont {Sauvage}}, \bibinfo {author} {\bibfnamefont {P.~J.}\
  \bibnamefont {Coles}},\ and\ \bibinfo {author} {\bibfnamefont
  {M.}~\bibnamefont {Cerezo}},\ }\bibfield  {title} {\bibinfo {title}
  {Representation theory for geometric quantum machine learning},\ }\href
  {https://arxiv.org/abs/2210.07980} {\bibfield  {journal} {\bibinfo  {journal}
  {arXiv preprint arXiv:2210.07980}\ } (\bibinfo {year} {2022})}\BibitemShut
  {NoStop}%
\bibitem [{\citenamefont {Mele}(2023)}]{mele2023introduction}%
  \BibitemOpen
  \bibfield  {author} {\bibinfo {author} {\bibfnamefont {A.~A.}\ \bibnamefont
  {Mele}},\ }\bibfield  {title} {\bibinfo {title} {Introduction to haar measure
  tools in quantum information: A beginner's tutorial},\ }\href
  {https://arxiv.org/abs/2307.08956} {\bibfield  {journal} {\bibinfo  {journal}
  {arXiv preprint arXiv:2307.08956}\ } (\bibinfo {year} {2023})}\BibitemShut
  {NoStop}%
\bibitem [{\citenamefont {Somma}\ \emph {et~al.}(2004)\citenamefont {Somma},
  \citenamefont {Ortiz}, \citenamefont {Barnum}, \citenamefont {Knill},\ and\
  \citenamefont {Viola}}]{somma2004nature}%
  \BibitemOpen
  \bibfield  {author} {\bibinfo {author} {\bibfnamefont {R.}~\bibnamefont
  {Somma}}, \bibinfo {author} {\bibfnamefont {G.}~\bibnamefont {Ortiz}},
  \bibinfo {author} {\bibfnamefont {H.}~\bibnamefont {Barnum}}, \bibinfo
  {author} {\bibfnamefont {E.}~\bibnamefont {Knill}},\ and\ \bibinfo {author}
  {\bibfnamefont {L.}~\bibnamefont {Viola}},\ }\bibfield  {title} {\bibinfo
  {title} {Nature and measure of entanglement in quantum phase transitions},\
  }\href {https://doi.org/10.1103/PhysRevA.70.042311} {\bibfield  {journal}
  {\bibinfo  {journal} {Physical Review A}\ }\textbf {\bibinfo {volume} {70}},\
  \bibinfo {pages} {042311} (\bibinfo {year} {2004})}\BibitemShut {NoStop}%
\bibitem [{\citenamefont {Somma}(2005)}]{somma2005quantum}%
  \BibitemOpen
  \bibfield  {author} {\bibinfo {author} {\bibfnamefont {R.~D.}\ \bibnamefont
  {Somma}},\ }\bibfield  {title} {\bibinfo {title} {Quantum computation,
  complexity, and many-body physics},\ }\href
  {https://arxiv.org/abs/quant-ph/0512209} {\bibfield  {journal} {\bibinfo
  {journal} {arXiv preprint quant-ph/0512209}\ } (\bibinfo {year}
  {2005})}\BibitemShut {NoStop}%
\bibitem [{\citenamefont {Larocca}\ \emph {et~al.}(2023)\citenamefont
  {Larocca}, \citenamefont {Ju}, \citenamefont {García-Martín}, \citenamefont
  {Coles},\ and\ \citenamefont {Cerezo}}]{larocca2021theory}%
  \BibitemOpen
  \bibfield  {author} {\bibinfo {author} {\bibfnamefont {M.}~\bibnamefont
  {Larocca}}, \bibinfo {author} {\bibfnamefont {N.}~\bibnamefont {Ju}},
  \bibinfo {author} {\bibfnamefont {D.}~\bibnamefont {García-Martín}},
  \bibinfo {author} {\bibfnamefont {P.~J.}\ \bibnamefont {Coles}},\ and\
  \bibinfo {author} {\bibfnamefont {M.}~\bibnamefont {Cerezo}},\ }\bibfield
  {title} {\bibinfo {title} {Theory of overparametrization in quantum neural
  networks},\ }\href
  {https://doi.org/https://doi.org/10.1038/s43588-023-00467-6} {\bibfield
  {journal} {\bibinfo  {journal} {Nature Computational Science}\ }\textbf
  {\bibinfo {volume} {3}},\ \bibinfo {pages} {542} (\bibinfo {year}
  {2023})}\BibitemShut {NoStop}%
\bibitem [{\citenamefont {Efthymiou}\ \emph {et~al.}(2021)\citenamefont
  {Efthymiou}, \citenamefont {Ramos-Calderer}, \citenamefont {Bravo-Prieto},
  \citenamefont {P{\'e}rez-Salinas}, \citenamefont {Garc{\'\i}a-Mart{\'\i}n},
  \citenamefont {Garcia-Saez}, \citenamefont {Latorre},\ and\ \citenamefont
  {Carrazza}}]{efthymiou2020qibo}%
  \BibitemOpen
  \bibfield  {author} {\bibinfo {author} {\bibfnamefont {S.}~\bibnamefont
  {Efthymiou}}, \bibinfo {author} {\bibfnamefont {S.}~\bibnamefont
  {Ramos-Calderer}}, \bibinfo {author} {\bibfnamefont {C.}~\bibnamefont
  {Bravo-Prieto}}, \bibinfo {author} {\bibfnamefont {A.}~\bibnamefont
  {P{\'e}rez-Salinas}}, \bibinfo {author} {\bibfnamefont {D.}~\bibnamefont
  {Garc{\'\i}a-Mart{\'\i}n}}, \bibinfo {author} {\bibfnamefont
  {A.}~\bibnamefont {Garcia-Saez}}, \bibinfo {author} {\bibfnamefont {J.~I.}\
  \bibnamefont {Latorre}},\ and\ \bibinfo {author} {\bibfnamefont
  {S.}~\bibnamefont {Carrazza}},\ }\bibfield  {title} {\bibinfo {title} {Qibo:
  a framework for quantum simulation with hardware acceleration},\ }\href
  {https://doi.org/10.1088/2058-9565/ac39f5} {\bibfield  {journal} {\bibinfo
  {journal} {Quantum Science and Technology}\ }\textbf {\bibinfo {volume}
  {7}},\ \bibinfo {pages} {015018} (\bibinfo {year} {2021})}\BibitemShut
  {NoStop}%
\bibitem [{\citenamefont {Efthymiou}\ \emph {et~al.}(2022)\citenamefont
  {Efthymiou}, \citenamefont {Lazzarin}, \citenamefont {Pasquale},\ and\
  \citenamefont {Carrazza}}]{efthymiou2022quantum}%
  \BibitemOpen
  \bibfield  {author} {\bibinfo {author} {\bibfnamefont {S.}~\bibnamefont
  {Efthymiou}}, \bibinfo {author} {\bibfnamefont {M.}~\bibnamefont {Lazzarin}},
  \bibinfo {author} {\bibfnamefont {A.}~\bibnamefont {Pasquale}},\ and\
  \bibinfo {author} {\bibfnamefont {S.}~\bibnamefont {Carrazza}},\ }\bibfield
  {title} {\bibinfo {title} {Quantum simulation with just-in-time
  compilation},\ }\href {https://doi.org/10.22331/q-2022-09-22-814} {\bibfield
  {journal} {\bibinfo  {journal} {Quantum}\ }\textbf {\bibinfo {volume} {6}},\
  \bibinfo {pages} {814} (\bibinfo {year} {2022})}\BibitemShut {NoStop}%
\bibitem [{\citenamefont {Coffman}\ \emph {et~al.}(2000)\citenamefont
  {Coffman}, \citenamefont {Kundu},\ and\ \citenamefont
  {Wootters}}]{coffman2000distributed}%
  \BibitemOpen
  \bibfield  {author} {\bibinfo {author} {\bibfnamefont {V.}~\bibnamefont
  {Coffman}}, \bibinfo {author} {\bibfnamefont {J.}~\bibnamefont {Kundu}},\
  and\ \bibinfo {author} {\bibfnamefont {W.~K.}\ \bibnamefont {Wootters}},\
  }\bibfield  {title} {\bibinfo {title} {Distributed entanglement},\ }\href
  {https://doi.org/10.1103/PhysRevA.61.052306} {\bibfield  {journal} {\bibinfo
  {journal} {Physical Review A}\ }\textbf {\bibinfo {volume} {61}},\ \bibinfo
  {pages} {052306} (\bibinfo {year} {2000})}\BibitemShut {NoStop}%
\bibitem [{\citenamefont {Hebenstreit}\ \emph {et~al.}(2019)\citenamefont
  {Hebenstreit}, \citenamefont {Jozsa}, \citenamefont {Kraus}, \citenamefont
  {Strelchuk},\ and\ \citenamefont {Yoganathan}}]{hebenstreit2019all}%
  \BibitemOpen
  \bibfield  {author} {\bibinfo {author} {\bibfnamefont {M.}~\bibnamefont
  {Hebenstreit}}, \bibinfo {author} {\bibfnamefont {R.}~\bibnamefont {Jozsa}},
  \bibinfo {author} {\bibfnamefont {B.}~\bibnamefont {Kraus}}, \bibinfo
  {author} {\bibfnamefont {S.}~\bibnamefont {Strelchuk}},\ and\ \bibinfo
  {author} {\bibfnamefont {M.}~\bibnamefont {Yoganathan}},\ }\bibfield  {title}
  {\bibinfo {title} {All pure fermionic non-gaussian states are magic states
  for matchgate computations},\ }\href
  {https://doi.org/10.1103/PhysRevLett.123.080503} {\bibfield  {journal}
  {\bibinfo  {journal} {Physical review letters}\ }\textbf {\bibinfo {volume}
  {123}},\ \bibinfo {pages} {080503} (\bibinfo {year} {2019})}\BibitemShut
  {NoStop}%
\bibitem [{\citenamefont {Reardon-Smith}\ \emph {et~al.}(2023)\citenamefont
  {Reardon-Smith}, \citenamefont {Oszmaniec},\ and\ \citenamefont
  {Korzekwa}}]{reardon2023improved}%
  \BibitemOpen
  \bibfield  {author} {\bibinfo {author} {\bibfnamefont {O.}~\bibnamefont
  {Reardon-Smith}}, \bibinfo {author} {\bibfnamefont {M.}~\bibnamefont
  {Oszmaniec}},\ and\ \bibinfo {author} {\bibfnamefont {K.}~\bibnamefont
  {Korzekwa}},\ }\bibfield  {title} {\bibinfo {title} {Improved simulation of
  quantum circuits dominated by free fermionic operations},\ }\href
  {https://doi.org/10.48550/arXiv.2307.12702} {\bibfield  {journal} {\bibinfo
  {journal} {arXiv preprint arXiv:2307.12702}\ } (\bibinfo {year}
  {2023})}\BibitemShut {NoStop}%
\bibitem [{\citenamefont {Dias}\ and\ \citenamefont
  {Koenig}(2023)}]{dias2023classical}%
  \BibitemOpen
  \bibfield  {author} {\bibinfo {author} {\bibfnamefont {B.}~\bibnamefont
  {Dias}}\ and\ \bibinfo {author} {\bibfnamefont {R.}~\bibnamefont {Koenig}},\
  }\bibfield  {title} {\bibinfo {title} {Classical simulation of non-gaussian
  fermionic circuits},\ }\href {https://arxiv.org/abs/2307.12912} {\bibfield
  {journal} {\bibinfo  {journal} {arXiv preprint arXiv:2307.12912}\ } (\bibinfo
  {year} {2023})}\BibitemShut {NoStop}%
\bibitem [{\citenamefont {Cudby}\ and\ \citenamefont
  {Strelchuk}(2023)}]{cudby2023gaussian}%
  \BibitemOpen
  \bibfield  {author} {\bibinfo {author} {\bibfnamefont {J.}~\bibnamefont
  {Cudby}}\ and\ \bibinfo {author} {\bibfnamefont {S.}~\bibnamefont
  {Strelchuk}},\ }\bibfield  {title} {\bibinfo {title} {Gaussian decomposition
  of magic states for matchgate computations},\ }\href
  {https://arxiv.org/abs/2307.12654} {\bibfield  {journal} {\bibinfo  {journal}
  {arXiv preprint arXiv:2307.12654}\ } (\bibinfo {year} {2023})}\BibitemShut
  {NoStop}%
\bibitem [{\citenamefont {Bratteli}\ and\ \citenamefont
  {Robinson}(2012)}]{bratteli2012operator}%
  \BibitemOpen
  \bibfield  {author} {\bibinfo {author} {\bibfnamefont {O.}~\bibnamefont
  {Bratteli}}\ and\ \bibinfo {author} {\bibfnamefont {D.~W.}\ \bibnamefont
  {Robinson}},\ }\href@noop {} {\emph {\bibinfo {title} {Operator algebras and
  quantum statistical mechanics: Volume 1: C*-and W*-Algebras. Symmetry Groups.
  Decomposition of States}}}\ (\bibinfo  {publisher} {Springer Science \&
  Business Media},\ \bibinfo {year} {2012})\BibitemShut {NoStop}%
\bibitem [{\citenamefont {De~Pasquale}\ and\ \citenamefont
  {Facchi}(2009)}]{de2009x}%
  \BibitemOpen
  \bibfield  {author} {\bibinfo {author} {\bibfnamefont {A.}~\bibnamefont
  {De~Pasquale}}\ and\ \bibinfo {author} {\bibfnamefont {P.}~\bibnamefont
  {Facchi}},\ }\bibfield  {title} {\bibinfo {title} {X y model on the circle:
  Diagonalization, spectrum, and forerunners of the quantum phase transition},\
  }\href {https://doi.org/10.1103/PhysRevA.80.032102} {\bibfield  {journal}
  {\bibinfo  {journal} {Physical Review A}\ }\textbf {\bibinfo {volume} {80}},\
  \bibinfo {pages} {032102} (\bibinfo {year} {2009})}\BibitemShut {NoStop}%
\bibitem [{\citenamefont {Bravyi}(2005)}]{bravyi2004lagrangian}%
  \BibitemOpen
  \bibfield  {author} {\bibinfo {author} {\bibfnamefont {S.}~\bibnamefont
  {Bravyi}},\ }\bibfield  {title} {\bibinfo {title} {Lagrangian representation
  for fermionic linear optics},\ }\href
  {https://doi.org/10.5555/2011637.2011640} {\bibfield  {journal} {\bibinfo
  {journal} {Quantum Info. Comput.}\ }\textbf {\bibinfo {volume} {5}},\
  \bibinfo {pages} {216–238} (\bibinfo {year} {2005})}\BibitemShut {NoStop}%
\bibitem [{\citenamefont {Goh}\ \emph {et~al.}(2023)\citenamefont {Goh},
  \citenamefont {Larocca}, \citenamefont {Cincio}, \citenamefont {Cerezo},\
  and\ \citenamefont {Sauvage}}]{goh2023lie}%
  \BibitemOpen
  \bibfield  {author} {\bibinfo {author} {\bibfnamefont {M.~L.}\ \bibnamefont
  {Goh}}, \bibinfo {author} {\bibfnamefont {M.}~\bibnamefont {Larocca}},
  \bibinfo {author} {\bibfnamefont {L.}~\bibnamefont {Cincio}}, \bibinfo
  {author} {\bibfnamefont {M.}~\bibnamefont {Cerezo}},\ and\ \bibinfo {author}
  {\bibfnamefont {F.}~\bibnamefont {Sauvage}},\ }\bibfield  {title} {\bibinfo
  {title} {Lie-algebraic classical simulations for variational quantum
  computing},\ }\href {https://arxiv.org/abs/2308.01432} {\bibfield  {journal}
  {\bibinfo  {journal} {arXiv preprint arXiv:2308.01432}\ } (\bibinfo {year}
  {2023})}\BibitemShut {NoStop}%
\end{thebibliography}%

\clearpage
\newpage

\onecolumngrid
\renewcommand\figurename{Supplemental Figure}

\section*   {Supplemental Information for ``Showcasing a Barren Plateau Theory Beyond the Dynamical Lie Algebra'' }

\section{Preliminaries}

In this section, we introduce definitions and useful results that will be needed throughout the rest of this Supplemental Information. 
We begin by recalling that we will be dealing with parametrized quantum circuits of the form
\begin{equation} \label{sup-eq:pqc}
    U(\thv) = \prod_l e^{-i\theta_l H_l}\,,
\end{equation}
where $\thv$ is a vector containing real parameters, and the $H_l$ are Hermitian operators, called generators, that are taken from a set $\GC$. These generators determine the Dynamical Lie Algebra 
 (DLA) $\mathfrak{g}$ of the circuit, defined below.

\begin{supdefinition}[Dynamical Lie Algebra]\label{sup-def:dynamical_lie_algebra} Given a set of Hermitian generators $\GC$, the dynamical Lie algebra $\mathfrak{g}$ is the subspace of operator space spanned by the repeated nested commutators of the elements in $i\GC$. That is
\begin{equation}\label{sup-eq:liea_def}
\liea={\rm span}_\mathbb{R}\left\langle i\GC \right\rangle_{Lie}\,,
\end{equation}
where $\left\langle i\GC \right\rangle_{Lie}$ denotes the Lie closure of $i\GC$.
\end{supdefinition}

The importance of the DLA stems from the fact that all unitaries of the form~\eqref{sup-eq:pqc} belong to a Lie group $G$ that is obtained via exponentiation of the algebra, i.e. $G=e^{\liea}$. Hence, the dimension of $\liea$ determines the ultimate expressivity of the parametrized quantum circuit~\cite{larocca2021diagnosing}. 

In this work, we are concerned with the DLA of parametrized matchgate circuits, i.e., circuits whose gates can be mapped to fermionic Gaussian unitaties. In particular, we take the  Lie algebra  $\g$ spanned by the nested commutators of the circuit's generators $H_l\in\GC=\{Z_1,\dots, Z_n, X_1 X_2,\dots, X_{n-1} X_{n}\}$,  where $n$ is the number of qubits and we use $X_i, Y_i, Z_i$ to denote the corresponding Pauli operator acting on the $i$-th qubit. This Lie algebra is given by~\cite{kokcu2022fixed,wiersema2023classification}
 \begin{equation} \label{sup-eq:dla}
     \g={\rm span}_{\mathbb R}i\{Z_i,\widehat{X_iX_j},\widehat{Y_iY_j},\widehat{X_iY_j},\widehat{Y_iX_j}\}_{1\leq i<j\leq n} \simeq \mathfrak{so}(2n)\,,
 \end{equation}
 \normalsize
where $\widehat{A_iB_j}=A_iZ_iZ_{i+1}\cdots Z_{j-1}B_j$ and $ \mathfrak{so}(2n)$ is the orthogonal Lie algebra. In order to work with this algebra, it is convenient to introduce the following $2n$ Majorana operators~\cite{jozsa2008matchgates}
\begin{equation}\begin{split} \label{sup-eq:majoranas}
    c_1&=XI\dots I,\; c_3= ZXI\dots I, \; c_{2n-1}=Z\dots Z X \\
        c_2&=YI\dots I,\; c_4= ZYI\dots I, \; \;\; c_{2n}\;\;\;=Z\dots Z Y\,. \end{split}
\end{equation}
We note that the Majoranas  are all Pauli strings that anti-commute among themselves, i.e.,
\begin{equation}\label{sup-eq:majo-anit-comm}
    \{c_\mu,c_\nu\}=2\delta_{\mu\nu}\;\;\;\; \mu,\nu=1,\dots ,2n\,.
\end{equation}
Via these operators, a free-fermionic Hamiltonian is defined as $H=i\sum_{\mu<\nu}h_{\mu\nu}c_\mu c_\nu$ with $h_{\mu\nu}$ a real anti-symmetric matrix of size $2n\times 2n$. A Gaussian unitary can then be written as $U=e^{iH}$ and satisfies that $U^\dag c_\mu U=\sum_\nu R_{\mu \nu}c_\nu$, with $R_{\mu \nu}=e^{2h_{\mu\nu}}$ a matrix in the fundamental representation of $\mathfrak{so}(2n)$, in agreement with the aforementioned isomorphism (this statement is proven in Sec.~\ref{sup-sec:lemma2}).

Under this map, we can obtain all the basis elements of $\g$ in~\eqref{sup-eq:dla} via the product of two Majoranas. More formally, we can prove the following Supplemental Proposition.

\begin{supproposition} \label{sup-prop:algebra} The dynamical Lie algebra in Eq.~\eqref{sup-eq:dla} is given by
\begin{equation}
    \g = {\rm span}_{\mathbb{R}}\{c_\mu c_\nu\}_{1\leq\mu<\nu\leq2n}\,.
\end{equation}

\end{supproposition}

\setcounter{figure}{0}
\begin{figure}[t!]
    \centering
\includegraphics[width=0.8\linewidth]{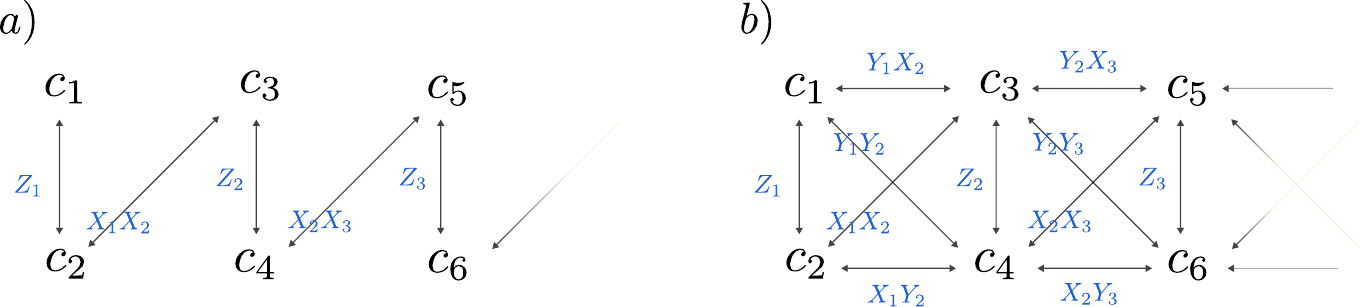}
    \caption{{\bf Majorana graph}. The vertices of the Majorana graph represent all the $c_\mu$ operators, and the edges represent the product of the operators at the connected vertices (ignoring phase factors for simplicity). As explained in the main text,  we can ``move'' through paths in the graph, and we interpret such movement as equivalent to taking the nested commutator of the generators associated to the traversed edges. In this way, a basis for the dynamical Lie algebra in Eq.~\eqref{sup-eq:dla} can be obtained. For example, by employing the generators in $a)$ one can construct all the local elements in $b)$. 
    Moreover, this procedure allows us to show that the basis of the dynamical Lie algebra consists of all possible products of two distinct Majorana operators, namely the complete graph.} 
    \label{sup-fig:majorana-graph}
\end{figure}

\begin{proof}
Let us first notice that the generators in $i\GC$ can be represented as follows,
\begin{equation}\label{sup-eq:maj-products} \begin{split}
    iZ_i&=c_{2i-1}c_{2i} \,, \\
    iX_iX_{i+1}&=c_{2i}c_{2i+1}\,.
  \end{split}
\end{equation}
In order to show how to obtain the rest of the basis elements in $\g$, we find it useful to introduce a pictorial representation where we arrange all the $c_\mu$ in a graph with two rows, the upper one containing as vertices the Majorana operators with odd $\mu$, and the lower one containing those with even $\mu$, as shown in Supp. Fig.~\ref{sup-fig:majorana-graph}. Then we draw edges between those vertices connecting two Majoranas that appear in the products in~\eqref{sup-eq:maj-products}. These edges represent the product of the two Majorana operators. We refer to this graph as the ``Majorana graph''. We now note that the following commutation relation holds,
\begin{equation} \label{sup-eq:maj-commutator}
    [ c_{\nu_1} c_\mu,  c_\mu c_{\nu_2}]=2\, c_{\nu_1}c_{\nu_2}\,,
\end{equation}
where the indices $\mu,\nu_1,\nu_2$ are all different from each other, and we used that $c_\mu^2=\id$ . Furthermore, we have
\begin{equation}\label{sup-eq:maj-zerocommutator}
    [ c_{\nu_1} c_{\nu_2},  c_{\nu_3} c_{\nu_4}]=0\,,
\end{equation}
whenever $\nu_1,\nu_2,\nu_3,\nu_4$ are all different. Next, let us take the path determined by two adjacent edges in the Majorana graph (e.g., $(c_2,c_3)$ and $(c_3,c_4)$) and say that traversing this path is equivalent to taking the commutator between the products of two Majoranas that each edge represent (e.g. $[c_2c_3,c_3c_4]$). From Eq.~\eqref{sup-eq:maj-commutator}, we know that this commutator is proportional to the product of the operators at the beginning and at the end of the path (e.g. $ c_2c_4$). Hence, this product belongs to $\g$ since, by definition, $\g$ is the span over the  real numbers of the nested commutators of the generators in $i\GC$ (which are precisely the edges in the Majorana graph). Now, we take a path consisting of an arbitrary number of adjacent edges, and we say that traversing it is equivalent to taking the nested commutators between the operators that each edge represent. But using~\eqref{sup-eq:maj-commutator} iteratively, we find that the nested commutator represented by the path is again proportional to the product of the Majoranas located at the beginning and at the end of the path. Moreover, if we were to select two non-adjacent edges and take the commutator between the operators attached to the edges, Eq.~\eqref{sup-eq:maj-zerocommutator} tells us that this commutator is zero.
Hence, we conclude that all nested commutators of the generators in $i\GC$ are proportional to products of two distinct Majoranas. Furthermore, since the Majorana graph is connected, there always exists a path between any two vertices, so all possible products of two Majoranas belong to $\g$. We thus obtain that $\dim(\g)=\binom{2n}{2}=n(2n-1)$.

\end{proof}

\section{Proof of Lemma I}

We here prove  Lemma~\ref{lem:modules}, for which we first recall the definition of the direct sum of vector spaces.

\begin{supdefinition}[Direct sum of vector spaces] \label{sup-def:direct_sum} Let $V$ be a vector space over a field $F$, and let $V_1, V_2,\dots,V_m$ be subspaces of $V$. Then, $V$ is said to be the direct sum of $V_1,V_2,\dots,V_m$, i.e.
\begin{equation}
    V=\bigoplus_{s=1}^m V_s\,,
\end{equation}
if and only if $V=\sum_{s=1}^m V_s$ and $V_s\cap V_{s'}=\{0\}$ whenever $s\neq s'$.
\end{supdefinition}
We now restate  Lemma~\ref{lem:modules} for convenience. We also refer the reader to ~\cite{wan2022matchgate} for an alternative version of the proof. 
\setcounter{lemma}{0}
\begin{lemma}\label{sup-lem:modules}
The space of linear operators acting on $n$-qubits, denoted as $\BC$, can be decomposed into subspaces as
\begin{equation} 
\mathcal{B}=\bigoplus_{\kappa=0}^{2n} \mathcal{B}_{\kappa}\,,
\end{equation}
with each  $\mathcal{B}_\kappa$ being the linear space, of dimension $\binom{2n}{\kappa}$, spanned by a basis of products of $\kappa$ distinct Majoranas.
 \end{lemma}

\begin{proof}
 Let us consider two distinct products of Majorana operators, i.e., $c_{\nu_1}\cdots c_{\nu_\kappa}$ and $c_{\nu'_1}\cdots c_{\nu'_{\kappa'}}$, where $\nu_1<\dots<\nu_\kappa$ and $\nu'_1<\dots<\nu'_{\kappa'}$. We ask ourselves the question of when can $c_{\nu_1}\cdots c_{\nu_\kappa} c_{\nu'_1}\cdots c_{\nu'_{\kappa'}}$ be proportional to the identity. 
 We begin by making the following observation: Given two different Majoranas from~\eqref{sup-eq:majoranas}, they exactly differ at a single qubit index. Hence, the product of two different Majoranas cannot be the identity matrix. Now, can the product of three different Majoranas be proportional to the identity? It is straightforward to see that we can choose a third Majorana such that we obtain $\pm\id$ at the qubit index where the former two Majoranas differ. However, this third Majorana is different from the previous two at a new qubit index. So the product of three distinct Majoranas cannot be proportional to the identity. Iterating this line of reasoning leads to the conclusion that the product of any number of different Majorana operators cannot be proportional to the identity matrix.
Thus, it follows that
 \begin{equation} \label{sup-eq:ortho-maj-prod}
     \Tr[c_{\nu_1}\cdots c_{\nu_\kappa} c_{\nu'_1}\cdots c_{\nu'_{\kappa'}}] =(-1)^{\lfloor \frac{\kappa}{2}\rfloor} d \, \delta_{\kappa \kappa'} \delta_{\nu_1 \nu'_1}\cdots\delta_{\nu_\kappa \nu'_\kappa}\,,
 \end{equation}
 where $d=2^n$ is the dimension of the Hilbert space of $n$ qubits.
Equation~\eqref{sup-eq:ortho-maj-prod} implies that $\BC_\kappa\cap\BC_{\kappa'}=\{0\}$ whenever $\kappa\neq\kappa'$ (this follows from the orthogonality implied by the $\delta_{\kappa\kappa'}$ factor). Furthermore, since $\BC_\kappa$ is the subspace of $\BC$ spanned by a basis of products of $\kappa$ different Majoranas (and all these products are orthogonal among themselves according to~\eqref{sup-eq:ortho-maj-prod}), its dimension is given by the number of all possible such (ordered) products, i.e., $\dim(\BC_\kappa)=\binom{2n}{\kappa}$. We then have $\sum_{\kappa=0}^{2n} \dim(\BC_\kappa)= 2^{2n}=d^2$, where $d=2^n$. This means that $\BC= \sum_{\kappa=0}^{2n} \BC_\kappa$, concluding the proof.
     
\end{proof}

\section{Proof of Lemma 2}
\label{sup-sec:lemma2}

In this Section, we prove Lemma~\ref{lem:inv}, which we recall for convenience.

\begin{lemma}\label{sup-lem:inv}
Let $M_\kappa\in \BC_\kappa$, then  $\forall U\in G$, $UM_\kappa U\ad\in\BC_\kappa$. Moreover, any pair of Pauli operators in $\BC_\kappa$ are proportional to each other via commutation with elements in $\g$.
\end{lemma}

\begin{proof}
Let us first choose $M_1=c_\mu$. Since $U\in G$, then it can always be written as $e^{-it H}$, where $iH\in\g$ and $t$ is a real parameter. We then have
\begin{equation}
    \frac{d (U c_\mu U\ad)}{dt} = [c_\mu, iH]\,.
\end{equation}
In order to evaluate the commutator $[c_\mu, iH]$, we note that
\begin{equation}
    [c_\mu, c_{\nu_1}c_{\nu_2}] =  \begin{cases} &0 \qquad\quad\; {\rm if \;\mu\neq \nu_1,\nu_2} \\ &2c_{\nu_2} \qquad {\rm if \;\mu=\nu_1} \\ &\!\!\!\!-2 c_{\nu_1} \qquad {\rm if\; \mu=\nu_2}  \end{cases}\,,
\end{equation}
where $c_{\nu_1}c_{\nu_2}\in\g$ (i.e., $\nu_1\neq\nu_2$). Therefore, since $iH=\sum_{\mu<\nu} h_{\mu\nu} c_\mu c_\nu$ with $h_{\mu\nu}$ a real anti-symmetric tensor (as proven in Supplemental Proposition~\ref{sup-prop:algebra}), we obtain
\begin{equation}
    \frac{d (U c_\mu U\ad)}{dt} = 2 \sum_{\nu=1}^{2n} h_{\mu\nu} U c_\nu U\ad \,,
\end{equation}
and therefore it follows that for $t=1$,
\begin{equation}
    U c_\mu U\ad = \sum_{\nu=1}^{2n} e^{2t\, h_{\mu\nu}} c_\nu\,.
\end{equation}
Thus, we conclude that $U c_\mu U\ad\in\BC_1$ for every $U\in G$.
Now, let us look at an operator $M_\kappa\in\BC_\kappa$. Then,
\begin{equation}
    U M_\kappa U\ad = U \left(\sum_{\nu_1< \cdots<\nu_\kappa} M_{\nu_1,\dots,\nu_\kappa} c_{\nu_1}\cdots c_{\nu_\kappa} \right) U\ad = \sum_{\nu_1< \cdots<\nu_\kappa}  M_{\nu_1\dots\nu_\kappa} U c_{\nu_1}U\ad U c_{\nu_2} U\ad\cdots U c_\kappa U\ad\,,
\end{equation}
where $M_{\nu_1,\dots,\nu_\kappa}$ is an anti-symmetric tensor.
Using that $U c_\mu U\ad\in\BC_1$ for every $U\in G$, it readily follows that $UM_\kappa U\ad\in\BC_\kappa$.

In order to prove the second part of the lemma, we simply need to compute the commutators between a product of $\kappa$ distinct Majoranas and a Pauli element of $\g$. These are given by
\begin{equation} \label{sup-eq:comm-algebra}
    [c_\mu c_\nu, c_{\nu_1}\cdots c_{\nu_\kappa}]= \begin{cases}
        &\!\!\!\!0 \qquad\qquad\qquad\qquad\qquad\qquad\quad {\rm if \;\mu,\nu\notin \{\nu_1,\dots,\nu_\kappa\} \; or\; \mu,\nu\in \{\nu_1,\dots,\nu_\kappa\}}\\ \pm & \!\!\!\!2 c_\mu  c_{\nu_1}\cdots c_{\nu_{i-1}}c_{\nu_{i+1}}\cdots c_{\nu_\kappa} \qquad\; {\rm if \;\mu\notin \{\nu_1,\dots,\nu_\kappa\} \; and \; \nu=\nu_i} \\ \pm&  \!\!\!\!2 c_\nu c_{\nu_1}\cdots c_{\nu_{i-1}}c_{\nu_{i+1}}\cdots c_{\nu_\kappa} \qquad\; {\rm if \;\mu=\nu_i \; and \; \nu\notin \{\nu_1,\dots,\nu_\kappa\}}
    \end{cases}\,.
\end{equation}
From Eq.~\eqref{sup-eq:comm-algebra} it is clear that given an operator $c_{\nu_1}\cdots c_{\nu_\kappa}\in\BC_\kappa$, it can be transformed into any other product of $\kappa$ distinct Majoranas via commutation with elements in $\g$.

\end{proof}

\section{Basis of the  commutant for the tensor square of a Lie group with associated Pauli-string Lie algebra}

As discussed in the main text, computing the variance of the loss function requires the evaluation of quantities such as $\mathbb{E}_{\thv}[U(\thv)^{\otimes t}(\cdot) U\ad(\thv)^{\otimes t} ]$, for $t\leq 2$. If we assume that the circuit is deep enough so that it forms a $2$-design over the Lie group $G$ (see Theorem 2 in~\cite{ragone2022representation}), then we can estimate such expectation values as 
\begin{equation}
    \mathbb{E}_{\thv}[U(\thv)^{\otimes t}(\cdot) U\ad(\thv)^{\otimes t} ]=\int_{G}dU U(\thv)^{\otimes t}(\cdot) U\ad(\thv)^{\otimes t} \,,
\end{equation}
where $dU$ denotes the normalized left-and-right-invariant Haar measure on $G$.  Moreover, it is well known that a map of the form $\EC^{(t)}(\cdot)=\int_{G}dU U(\thv)^{\otimes t}(\cdot) U\ad(\thv)^{\otimes t}$ is a projector into the commutant of the $t$-th fold tensor representation of $G$. That is, 
\begin{equation}
    \EC^{(t)}:\BC\rightarrow {\rm comm}(G^{\otimes t})\,,
\end{equation}
where 
\begin{equation}
    {\rm comm}(G^{\otimes t})=\{A\in \BC^{\otimes t}\,\,|\,\, [A,U^{\otimes t}]=0\,\,\forall U\in G\}\,.
\end{equation}
Moreover, given a Hermitian and orthonormal basis  $\{B^{(t)}_\eta\}_{\eta=1}^{\dim({\rm comm}(G^{\otimes t}))}$ of ${\rm comm}(G^{\otimes t})$ (with respect to the standard Hilbert--Schmidt inner product) we can use the Weingarten calculus to find that for any operator $A\in\BC^{\otimes t}$~\cite{ragone2022representation,garcia2023deep,mele2023introduction}
\begin{equation}\label{eq:weingarten}
    \EC^{(t)}(A)=\sum_{\eta=1}^{\dim({\rm comm}(G^{\otimes t}))}\Tr[B^{(t)}_\eta A]B^{(t)}_\eta\,.
\end{equation}
In what follows, we will refer to the elements of ${\rm comm}(G)$ as \textit{linear symmetries}, while to those of ${\rm comm}(G^{\otimes 2})$ as \textit{quadratic symmetries}. 

From the previous, we can see that evaluating Haar random averages over a group, requires knowledge of a basis for the tensor-fold representation's commutant. In this section, we present a method to find a basis for ${\rm comm}(G^{\otimes 2})$, i.e., for all the quadratic symmetries of any Pauli string DLA. We recall that the DLA associated to the matchgate circuit is precisely of this form, and hence, a special case of our results. The only assumption that we will make, is that we already know  the basis for  ${\rm comm}(G)$, i.e., for all the linear symmetries. For $\g$ being the matchgate circuit's DLA, then we know that ${\rm comm}(G)={\rm span}_{\mathbb{C}}\{\id,P\}$.  We begin by defining some basic terminology and observing some simple facts.

\begin{supdefinition}[Linear symmetries]  Given a set of generators $\GC$ of Lie algebra $\g$, as in $\eqref{sup-eq:liea_def}$, its linear symmetries are all the linear operators $L$ such that $[H_l, L]=0 \quad\forall H_l\in\GC$.
\end{supdefinition}

We note that the linear symmetries form a vector subspace of operators that commute with the generators of the algebra (and hence, with the algebra itself and the Lie group), i.e., if $L_1$ and $L_2$ are linear symmetries it follows that $\alpha L_1 + \beta L_2$ is also a linear symmetry for arbitrary $\alpha,\beta\in\mathbb{C}$.

\begin{supdefinition}[Quadratic symmetries] Given a set of generators $\GC$ of Lie algebra $\g$, as in $\eqref{sup-eq:liea_def}$, its quadratic symmetries are all the linear operators $Q$ such that $[H_l\otimes\id+\id\otimes H_l, Q]=0 \quad\forall H_l\in\GC$.
    
\end{supdefinition}

We remark that if $Q$ is a quadratic symmetry, then it follows that $[U^{\otimes 2},Q]=0\quad\forall U\in G$. Hence, the subspace spanned by the quadratic symmetries is also referred to as the commutant of the tensor-two representation of $G$. We now introduce the concept of a Pauli-string DLA.

\begin{supdefinition}[Pauli-string DLA]
\label{sup-def:pauli-str-dla}
A Lie algebra $\mf{g}$ is a Pauli-string DLA if $\mf{g}$ has a set of generators that are all Pauli strings.
    
\end{supdefinition}

A Pauli-string DLA $\mf{g}$ is much easier to understand than most DLAs. For example, the commutator of any two Pauli strings is either zero or another Pauli string (up to a factor of $\pm 2i$), so it immediately follows that any non-zero nested commutator will be proportional to a single Pauli string. Hence, there exists a basis of $\mf{g}$ consisting entirely of Pauli strings. Similarly, it is easy to see that the space of linear symmetries of a Pauli-string DLA has a basis consisting only of Pauli strings. We make another useful observation with a trivial lemma.

\begin{suplemma}
\label{sup-lem:pauli-str-commutator-bidirectional}
Suppose that $S_1, S_2, S_3$ are three Pauli strings. Then $[S_1, S_3] = \pm 2iS_2$ if and only if $[S_2, S_3] = \mp 2iS_1$.
\end{suplemma}

\begin{proof}
We proceed as follows,
\begin{equation}\begin{split}
    [S_1, S_3] = \pm 2iS_2 &\iff S_1S_3 = \pm iS_2\\
    &\iff S_1 = \pm iS_2S_3 \\
    &\iff S_2S_3 = \mp iS_1 \\
    &\iff [S_2, S_3] = \mp 2iS_1\,. \end{split}
\end{equation}
\end{proof}

Supplemental Lemma~\ref{sup-lem:pauli-str-commutator-bidirectional} shows that the action of taking the commutator with $S_3$ is bi-directional, as it maps $S_1$ to $S_2$ and $S_2$ to $S_1$ (up to constant factors). This observation motivates the following definition.

\begin{supdefinition}[Commutator graph]
\label{sup-def:commutator-graph}
Given a set of generators $\GC$ of Pauli strings on $n$ qubits, the commutator graph associated with $\GC$ is a graph with vertices given by the $4^n$ Pauli strings on $n$ qubits, and edges connecting two Pauli strings $S_1, S_2$ if and only if there exists a Pauli string $S_3\in\GC$ such that $[S_1, S_3]\propto S_2$, or equivalently, $[S_2, S_3]\propto S_1$.
\end{supdefinition}

In order to better grasp the concept of a commutator graph, we provide an example of what this graph looks like for the matchgate circuit DLA~\eqref{sup-eq:dla} on $n=3$ qubits.

\begin{supexample}[Commutator graph]
\label{ex:P3-commutator-graph}
The generating set for the matchgate circuit DLA~\eqref{sup-eq:dla} on $n=3$ qubits consists of the Pauli strings $ZII, IZI, IIZ, XXI, IXX$, where, e.g., $ZII$ stands for $\Z\otimes \id \otimes\id$. The commutator graph has $4^3 = 64$ vertices corresponding to all the $3$-qubit Pauli strings. To label the edges, we color them \textbf{\textcolor{red}{red}} for $ZII$, \textbf{\textcolor{orange}{orange}} for $IZI$, \textbf{\textcolor{green}{green}} for $IIZ$, \textbf{\textcolor{blue}{blue}} for $XXI$, and \textbf{\textcolor{violet}{violet}} for $IXX$. Then the commutator graph looks as follows:

\begin{center}  
    \begin{tikzpicture}[node distance={15mm}, thick, main/.style = {circle, draw}]
        \node[main] (ZZY) [scale=0.75] {$ZZY$};
        \node[main] (ZZX) [left of = ZZY, scale=0.75] {$ZZX$};
        \node[main] (ZYI) [left of = ZZX, scale=0.75] {$ZYI$};
        \node[main] (ZXI) [left of = ZYI, scale=0.75] {$ZXI$};
        \node[main] (YII) [left of = ZXI, scale=0.75] {$YII$};
        \node[main] (XII) [left of = YII, scale=0.75] {$XII$};
        \draw[green] (ZZY) -- (ZZX);
        \draw[violet] (ZZX) -- (ZYI);
        \draw[orange] (ZYI) -- (ZXI);
        \draw[blue] (ZXI) -- (YII);
        \draw[red] (YII) -- (XII);
        \node[main] (III) [left of = XII, xshift=-50mm, scale=0.75] {$III$};
    \end{tikzpicture}
\end{center}

\begin{center}
    \begin{tikzpicture}[node distance={15mm}, thick, main/.style = {circle, draw}]
        \node[main] (IIX) [scale=0.75] {$IIX$};
        \node[main] (IIY) [left of = IIX, scale=0.75] {$IIY$};
        \node[main] (IXZ) [left of = IIY, scale=0.75] {$IXZ$};
        \node[main] (IYZ) [left of = IXZ, scale=0.75] {$IYZ$};
        \node[main] (XZZ) [left of = IYZ, scale=0.75] {$XZZ$};
        \node[main] (YZZ) [left of = XZZ, scale=0.75] {$YZZ$};
        \draw[green] (IIX) -- (IIY);
        \draw[violet] (IIY) -- (IXZ);
        \draw[orange] (IXZ) -- (IYZ);
        \draw[blue] (IYZ) -- (XZZ);
        \draw[red] (XZZ) -- (YZZ);
        \node[main] (ZZZ) [left of = YZZ, , xshift=-50mm, scale=0.75] {$ZZZ$};
    \end{tikzpicture}
\end{center}
\vspace{.2cm}

\begin{center}
    \begin{tikzpicture}[node distance={15mm}, thick, main/.style = {circle, draw}]
        \node[main] (ZII) [scale=0.75] {$ZII$};
        \node[main] (YXI) [below of = ZII, scale=0.75] {$YXI$};
        \node[main] (XXI) [right of = YXI, scale=0.75] {$XXI$};
        \node[main] (YYI) [below of = YXI, scale=0.75] {$YYI$};
        \node[main] (XYI) [right of = YYI, scale=0.75] {$XYI$};
        \node[main] (IZI) [right of = XYI, scale=0.75] {$IZI$};
        \node[main] (YZX) [below of = YYI, scale=0.75] {$YZX$};
        \node[main] (XZX) [right of = YZX, scale=0.75] {$XZX$};
        \node[main] (IYX) [right of = XZX, scale=0.75] {$IYX$};
        \node[main] (IXX) [right of = IYX, scale=0.75] {$IXX$};
        \node[main] (YZY) [below of = YZX, scale=0.75] {$YZY$};
        \node[main] (XZY) [right of = YZY, scale=0.75] {$XZY$};
        \node[main] (IYY) [right of = XZY, scale=0.75] {$IYY$};
        \node[main] (IXY) [right of = IYY, scale=0.75] {$IXY$};
        \node[main] (IIZ) [right of = IXY, scale=0.75] {$IIZ$};
        \draw[blue] (ZII) -- (YXI);
        \draw[red] (YXI) -- (XXI);
        \draw[orange] (YXI) -- (YYI);
        \draw[orange] (XXI) -- (XYI);
        \draw[red] (YYI) -- (XYI);
        \draw[blue] (XYI) -- (IZI);
        \draw[violet] (YYI) -- (YZX);
        \draw[violet] (XYI) -- (XZX);
        \draw[violet] (IZI) -- (IYX);
        \draw[red] (YZX) -- (XZX);
        \draw[blue] (XZX) -- (IYX);
        \draw[orange] (IYX) -- (IXX);
        \draw[green] (YZX) -- (YZY);
        \draw[green] (XZX) -- (XZY);
        \draw[green] (IYX) -- (IYY);
        \draw[green] (IXX) -- (IXY);
        \draw[red] (YZY) -- (XZY);
        \draw[blue] (XZY) -- (IYY);
        \draw[orange] (IYY) -- (IXY);
        \draw[violet] (IXY) -- (IIZ);
    \end{tikzpicture}\hspace{1cm}
    \begin{tikzpicture}[node distance={15mm}, thick, main/.style = {circle, draw}]
        \node[main] (IZZ) [scale=0.75] {$IZZ$};
        \node[main] (XYZ) [below of = IZZ, scale=0.75] {$XYZ$};
        \node[main] (YYZ) [right of = XYZ, scale=0.75] {$YYZ$};
        \node[main] (XXZ) [below of = XYZ, scale=0.75] {$XXZ$};
        \node[main] (YXZ) [right of = XXZ, scale=0.75] {$YXZ$};
        \node[main] (ZIZ) [right of = YXZ, scale=0.75] {$ZIZ$};
        \node[main] (XIY) [below of = XXZ, scale=0.75] {$XIY$};
        \node[main] (YIY) [right of = XIY, scale=0.75] {$YIY$};
        \node[main] (ZXY) [right of = YIY, scale=0.75] {$ZXY$};
        \node[main] (ZYY) [right of = ZXY, scale=0.75] {$ZYY$};
        \node[main] (XIX) [below of = XIY, scale=0.75] {$XIX$};
        \node[main] (YIX) [right of = XIX, scale=0.75] {$YIX$};
        \node[main] (ZXX) [right of = YIX, scale=0.75] {$ZXX$};
        \node[main] (ZYX) [right of = ZXX, scale=0.75] {$ZYX$};
        \node[main] (ZZI) [right of = ZYX, scale=0.75] {$ZZI$};
        \draw[blue] (IZZ) -- (XYZ);
        \draw[red] (XYZ) -- (YYZ);
        \draw[orange] (XYZ) -- (XXZ);
        \draw[orange] (YYZ) -- (YXZ);
        \draw[red] (XXZ) -- (YXZ);
        \draw[blue] (YXZ) -- (ZIZ);
        \draw[violet] (XXZ) -- (XIY);
        \draw[violet] (YXZ) -- (YIY);
        \draw[violet] (ZIZ) -- (ZXY);
        \draw[red] (XIY) -- (YIY);
        \draw[blue] (YIY) -- (ZXY);
        \draw[orange] (ZXY) -- (ZYY);
        \draw[green] (XIY) -- (XIX);
        \draw[green] (YIY) -- (YIX);
        \draw[green] (ZXY) -- (ZXX);
        \draw[green] (ZYY) -- (ZYX);
        \draw[red] (XIX) -- (YIX);
        \draw[blue] (YIX) -- (ZXX);
        \draw[orange] (ZXX) -- (ZYX);
        \draw[violet] (ZYX) -- (ZZI);
    \end{tikzpicture}
\end{center}

\begin{center}
    \begin{tikzpicture}[node distance={15mm}, thick, main/.style = {circle, draw}]
        \node[main] (IXI) [scale=0.75] {$IXI$};
        
        \node[main] (IYI) [right of = IXI, scale=0.75] {$IYI$};
        
        \node[main] (IZX) [above right of = IYI, scale=0.75]  {$IZX$};
        \node[main] (XZI) [below right of = IYI, scale=0.75] {$XZI$};
        
        \node[main] (IZY) [above right of = IZX, scale=0.75] {$IZY$};
        \node[main] (XYX) [below right of = IZX, scale=0.75] {$XYX$};
        \node[main] (YZI) [below right of = XZI, scale=0.75] {$YZI$};
        
        \node[main] (XYY) [right of = IZY, scale=0.75] {$XYY$};
        \node[main] (XXX) [right of = XYX, scale=0.75] {$XXX$};
        \node[main] (YYX) [right of = YZI, scale=0.75] {$YYX$};
        
        \node[main] (XXY) [right of = XYY, scale=0.75] {$XXY$};
        \node[main] (YYY) [right of = XXX, scale=0.75] {$YYY$};
        \node[main] (YXX) [right of = YYX, scale=0.75] {$YXX$};
        
        \node[main] (XIZ) [right of = XXY, scale=0.75] {$XIZ$};
        \node[main] (YXY) [right of = YYY, scale=0.75] {$YXY$};
        \node[main] (ZIX) [right of = YXX, scale=0.75] {$ZIX$};
        
        \node[main] (YIZ) [above right of = YXY, scale=0.75] {$YIZ$};
        \node[main] (ZIY) [below right of = YXY, scale=0.75] {$ZIY$};
        
        \node[main] (ZXZ) [below right of = YIZ, scale=0.75] {$ZXZ$};
        
        \node[main] (ZYZ) [right of = ZXZ, scale=0.75] {$ZYZ$};
        
        \draw[orange] (IXI) -- (IYI);
        \draw[violet] (IYI) -- (IZX);
        \draw[blue] (IYI) -- (XZI);
        \draw[blue] (IZX) -- (XYX);
        \draw[green] (IZX) -- (IZY);
        \draw[violet] (XZI) -- (XYX);
        \draw[red] (XZI) -- (YZI);
        \draw[blue] (IZY) -- (XYY);
        \draw[green] (XYX) -- (XYY);
        \draw[orange] (XYX) -- (XXX);
        \draw[red] (XYX) -- (YYX);
        \draw[violet] (YZI) -- (YYX);
        \draw[red] (XXX) -- (YXX);
        \draw[green] (XXX) -- (XXY);
        \draw[orange] (XYY) -- (XXY);
        \draw[red] (XYY) -- (YYY);
        \draw[orange] (YYX) -- (YXX);
        \draw[green] (YYX) -- (YYY);
        \draw[red] (XXY) -- (YXY);
        \draw[violet] (XXY) -- (XIZ);
        \draw[orange] (YYY) -- (YXY);
        \draw[green] (YXX) -- (YXY);
        \draw[blue] (YXX) -- (ZIX);
        \draw[red] (XIZ) -- (YIZ);
        \draw[violet] (YXY) -- (YIZ);
        \draw[blue] (YXY) -- (ZIY);
        \draw[green] (ZIX) -- (ZIY);
        \draw[blue] (YIZ) -- (ZXZ);
        \draw[violet] (ZIY) -- (ZXZ);
        \draw[orange] (ZXZ) -- (ZYZ);
    \end{tikzpicture}
\end{center}

In particular, notice that the commutator graph has $7$ connected components, with sizes $1, 6, 15, 20, 15, 6, 1$. One can observe that the connected components can be indexed from $0$ to $6$, and that component $\kappa$ contains all the $\binom{6}{\kappa}$ Pauli strings that are the product of exactly $\kappa$ distinct Majorana operators. These correspond precisely to the $G$-modules.
\end{supexample}

We are now ready to state the key theorem, which provides a basis for the quadratic symmetries of an arbitrary Pauli-string DLA.

\begin{suptheorem}
\label{sup-thm:pauli-str-dla-quad-sym}
Consider a Pauli-string DLA with generating set $\GC$, Pauli-string basis of linear symmetries $\{L_j\}_{j=1}^{J}$, and commutator graph with connected components $\{C_\kappa\}_{\kappa=1}^{K}$. Then the space of quadratic symmetries of $\GC$ has dimension $JK$ and has a basis given by the operators
\begin{equation}\label{sup-eq:quadratic-Pauli-symmetries}
    Q_\kappa^j = \sum_{S\in C_\kappa}S\otimes L_jS
\end{equation}
for $1\le j\le J$ and $1\le \kappa\le K$\,.
\end{suptheorem}

To prove Supplemental Theorem~\ref{sup-thm:pauli-str-dla-quad-sym}, we first show that the $Q_\kappa^j$ operators are indeed quadratic symmetries.

\begin{suplemma}
\label{sup-lem:quad-sym-verification}
Each $Q_\kappa^j$ is a quadratic symmetry of $\GC$. In other words, $Q_\kappa^j$ commutes with $H_l\otimes \id + \id\otimes H_l$ for every generating Pauli string $H_l\in\GC$.
\end{suplemma}

\begin{proof}
Let $S$ be some $n$-qubit Pauli string. If $S$ commutes with $H_l$, then $S\otimes L_jS$ commutes with both $H_l\otimes \id$ and $\id\otimes H_l$. If $S$ anti-commutes with $H_l$, then let $\phi(S)$ denote the Pauli string such that $[H_l,S] = \pm 2i\phi(S)$. Then by Supplemental Lemma~\ref{sup-lem:pauli-str-commutator-bidirectional}, $[H_l,\phi(S)] = \mp 2iS$, so
\begin{equation} \begin{split} \label{sup-eq:hl-comm}
    [H_l\otimes \id+ \id\otimes H_l, S\otimes L_jS + \phi(S)\otimes L_j\phi(S)] 
    = \,\, & \pm 2i\, \phi(S)\otimes L_jS \mp 2i\, S\otimes L_j\phi(S) \\
    & \pm 2i\, S\otimes L_j\phi(S) \mp 2i\, \phi(S)\otimes L_jS \\
    = \,\, & 0\,. \end{split}
\end{equation}

Now, for any generator $H_l$, we can split $Q_\kappa^j$ as
\begin{equation} \label{sup-eq:Q-split}
    Q_\kappa^j = \sum_{\substack{S\in C_\kappa\\ [S,H_l]=0}} S\otimes L_jS + \sum_{\substack{S\in C_\kappa\\ \{S,H_l\}=0}} S\otimes L_jS\,.
\end{equation}
The first sum in~\eqref{sup-eq:Q-split} obviously commutes with $H_l\otimes \id+ \id\otimes H_l$, while the second sum contains both $S$ and $\phi(S)$ for every anti-commuting $S$ (this follows from Supplemental Lemma \ref{sup-lem:pauli-str-commutator-bidirectional} and Supplemental Definition~\ref{sup-def:commutator-graph}). Hence, using~\eqref{sup-eq:hl-comm} we have that $Q_\kappa^j$ commutes with $H_l\otimes \id + \id\otimes H_l$ for any generator $H_l\in\GC$.
\end{proof}

We will now prove that there are no other quadratic symmetries. To do this, let us assume assume that the general operator $Q=\sum_{S,S'} \tilde{c}(S,S') \, S\otimes S'$  (with $\tilde{c}(S,S')$ some complex coefficients) is a quadratic symmetry. We will prove that $Q$ necessarily lies in the span of the $Q^j_\kappa$ operators. To accomplish this, it suffices to prove the following:
\begin{enumerate}[i)]
    \item If a $2n$-qubit Pauli string $S\otimes S'$ is such that $\tilde{c}(S,S')\neq 0$ for a quadratic symmetry $Q$, then $S'\propto L_jS$ for some linear symmetry $L_j$. This allows us to write $\tilde{c}(S,S') \, S\otimes S' = c(S,j)  \,S\otimes L_j S$.
    \item If $S_1$ and $S_2$ are in the same component $C_\kappa$ of the commutator graph of $\GC$, then $c(S_1,j) = c(S_2,j)$.
\end{enumerate}
We begin by proving the first of these items, which establishes which $2n$-qubit Pauli strings are allowed to appear in the quadratic symmetries at all.

\begin{suplemma}
\label{sup-lem:quad-sym-zero-coeffs}
Suppose that $Q=\sum_{S,S'} \tilde{c}(S,S')\, S\otimes S'$ is a quadratic symmetry of $\GC$, and that $\tilde{c}(S,S')\neq 0$ for $n$-qubit Pauli strings $S$ and $S'$. Then $S'\propto L_j S$ for some linear symmetry $L_j$.
\end{suplemma}

\begin{proof}
Suppose that $\tilde{c}(S,S')\neq 0$ for some $S$, and suppose for the sake of contradiction that there exists a generator $H_l$ that commutes with $S$ but anti-commutes with $S'$. Then,
\begin{equation}
    [H_l\otimes \id + \id\otimes H_l, Q] = 2\tilde{c}(S,S')\,S\otimes H_l S' + \text{(other Pauli-string terms)}\,.
\end{equation}
In order for the above commutator to be zero, an additional term proportional to $S\otimes H_l S'$ must appear to cancel the non-zero contribution $2\tilde{c}(S,S')\,S\otimes H_l S'$. The only possibility is that $Q$ contains a term proportional to $H_l S\otimes H_l S'$ (so that multiplication by $H_l\otimes \id$ leads to $S\otimes H_l S'$). But $H_lS\otimes H_l S'$ commutes with $H_l\otimes \id$, so it results in a null contribution to the commutator.
Hence the above commutator is non-zero, which contradicts $Q$ being a quadratic symmetry. An analogous reasoning produces a contradiction if there were a generator $H_l$ that anti-commutes with $S$ but commutes with $S'$.

We conclude that every generator $H_l$ either commutes with both $S$ and $S'$ or anti-commutes with both $S$ and $S'$. But this means that the product $S'S$ commutes with every generator $H_l$ and thus must be proportional to a linear symmetry $L_j$ (where $L_j$ is itself a Pauli string). We can rearrange this to obtain $S'\propto L_j S$, as desired.
\end{proof}

We now proceed to prove the second item. From the previous result lemma, we can rewrite  $Q=\sum_{S,S'} \tilde{c}(S,S')\, S\otimes S'=\sum_{S,j} \tilde{c}(S,j)\, S\otimes L_j S$, where we defined $c(S,j)$ such that $\tilde{c}(S,S') \, S\otimes S' = c(S,j)  \,S\otimes L_j S$. We begin by showing a restriction on the coefficients $c(S,j)$ for Pauli strings that are adjacent in the commutator graph.

\begin{suplemma}
\label{sup-lem:quad-sym-adjacent-equal-coeffs}
Suppose that $S_1$ and $S_2$ are adjacent Pauli strings in the commutator graph of $\GC$, and suppose that $Q=\sum_{S,S'} \tilde{c}(S,S')\, S\otimes S'$ is a quadratic symmetry of $\GC$. We write $\tilde{c}(S,S') \, S\otimes S' = c(S,j)  \,S\otimes L_j S$. Then, $c(S_1,j) = c(S_2,j)$.
\end{suplemma}

\begin{proof}
Since $S_1$ and $S_2$ are adjacent Pauli strings in the commutator graph of $\GC$, then there is a Pauli-string generator $H_l\in\GC$ such that $[H_l,S_1] = \pm 2iS_2$. Then Supplemental Lemma~\ref{sup-lem:pauli-str-commutator-bidirectional} implies that $[H_l,S_2] = \mp 2i S_1$. We now write $Q = c(S_1,j)\, S_1\otimes L_jS_1+ c(S_2,j)\,S_2\otimes L_jS_2 + \text{(other terms)}$, and we wish to show that $c(S_1,j) = c(S_2,j)$. We can directly compute
\begin{equation}
    [H_l\otimes \id + \id\otimes H_l, Q] = \pm 2i(c(S_1,j) - c(S_2,j))S_2\otimes L_jS_1 \pm 2i(c(S_1,j) - c(S_2,j))S_1\otimes L_jS_2 + \text{(other Pauli-string terms)}.
\end{equation}
Let us focus on the $S_2\otimes L_j S_1$ term in the above expression (the argument can be carried out in the same way for the $S_1\otimes L_jS_2$ term). In particular, $S_2\otimes L_jS_1$ can appear in the above expression in exactly two ways: from the $S_1\otimes L_jS_1$ term in $Q$ (after multiplication by $H_l\otimes I$), and from the $S_2\otimes L_jS_2$ term in $Q$ (after multiplication by $I\otimes H_l$). Thus it is not possible for the above expression to be zero unless $c(S_1,j) = c(S_2,j)$.
\end{proof}

We can now extend Supplemental Lemma \ref{sup-lem:quad-sym-adjacent-equal-coeffs} to relate coefficients for Pauli strings that are in the same connected component, which finalizes the proof of the second item.

\begin{suplemma}
\label{sup-lem:quad-sym-component-equal-coeffs}
Suppose that $S_1$ and $S_2$ are Pauli strings in the same connected component of the commutator graph of $\GC$ (which we denote with index $\kappa$), and suppose that $Q=\sum_{S,S'} \tilde{c}(S,S')\, S\otimes S'$ is a quadratic symmetry of $\GC$. We write $\tilde{c}(S,S') \, S\otimes S' = c(S,j)  \,S\otimes L_j S$. Then, $c(S_1,j) = c(S_2,j)=c(\kappa,j)$.
\end{suplemma}

\begin{proof}
Since $S_1$ and $S_2$ are in the same connected component, there exists a path of length $L$ connecting them in the commutator graph: $S_1 P_1 P_2\dots P_{L-1} S_2$ (with $P_0 = S_1$ and $P_L = S_2$). Now we just repeatedly apply Supplemental Lemma~\ref{sup-lem:quad-sym-adjacent-equal-coeffs}. For $P_l$ a Pauli string in such path and $0\le l\le L-1$, the coefficient of $P_l\otimes L_j P_l$ equals the coefficient of $P_{l+1}\otimes L_j P_{l+1}$. Together these imply that the coefficient of $S_1\otimes L_j S_1$ must equal the coefficient of $S_2\otimes L_j S_2$, i.e. $c(S_1,j)=c(S_2,j)$.
\end{proof}

This completes the proof of Supplemental Theorem~\ref{sup-thm:pauli-str-dla-quad-sym}, as it shows that the following chain of equalities hold:
\begin{align}
    Q&=\sum_{S,S'} \tilde{c}(S,S')\, S\otimes S'\nonumber\\
    &=\sum_{S,j} c(S,j)\, S\otimes L_j S \quad \quad (\text{Supplemental Lemma~\ref{sup-lem:quad-sym-zero-coeffs}})\nonumber\\
    &=\sum_{S\in C_\kappa,j} c(\kappa,j)\, S\otimes L_j S \quad \quad (\text{Supplemental Lemma~\ref{sup-lem:quad-sym-component-equal-coeffs}})\nonumber\\
    &=\sum_{\kappa,j} c(\kappa,j)\, Q_\kappa^j\,.\nonumber
\end{align}

We conclude this section with a trivial example to show how the theorem works.

\begin{supexample}[Quadratic symmetries]
\label{ex:trivial-quad-sym}
Consider the $1$-dimensional Lie algebra on $n=1$ qubit, $\mf{g} = \text{span}_{\mathbb{R}}\{iX\}$. The linear symmetries are $L_0 = I$ and $L_1 = X$. The commutator graph has the $4$ single-qubit Pauli strings $I,X,Y,Z$. We connect $Y$ and $Z$ with an edge, but $I$ and $X$ are isolated. Thus the connected components are $C_0 = \{I\}$, $C_1 = \{Y,Z\}$, and $C_2 = \{X\}$. We now use Supplemental Theorem~\ref{sup-thm:pauli-str-dla-quad-sym} to write down the quadratic symmetries:
\begin{equation}\begin{split}
    Q^0_0 &= I\otimes I \,,\\
    Q^0_1 &= Y\otimes Y + Z\otimes Z \,,\\
    Q^0_2 &= X\otimes X \,,\\
    Q^1_0 &= I\otimes(XI) = I\otimes X \,,\\
    Q^1_1 &= Y\otimes(XY) + Z\otimes(XZ) = i(Y\otimes Z - Z\otimes Y) \,,\\
    Q^1_2 &= X\otimes(XX) = X\otimes I\,.\end{split}
\end{equation}
It is easy to verify that these $6$ operators indeed form a basis for the quadratic symmetries of $\mf{g}$, which in this case is simply the commutant of $X\otimes\id + \id\otimes X$.
\end{supexample}

\section{Quadratic symmetries of parametrized matchgate circuits}

Equipped with Supplemental Theorem \ref{sup-thm:pauli-str-dla-quad-sym}, we return to the special case of the matchgate circuit's DLA. We first introduce some necessary notation. We denote by $[2n]$ the set of integers $\{1,\dots,2n\}$. Then, we write $\vec{s}\in\binom{[2n]}{\kappa}$ to refer to a subset $\vec{s}\subseteq [2n]$ of size $\kappa$, and $\vec{\bar{s}}\subseteq [2n]$ to refer to its complement. Let $\nu_1,\dots,\nu_\kappa$ (with $\nu_1<\cdots < \nu_\kappa$) be the elements in $\vec{s}$. Then, we write $c^{\vec{s}}$ to indicate the product of $\kappa$ distinct Majoranas $c_{\nu_1}\cdots c_{\nu_\kappa}$. Furthermore, we employ the notation $\pi(\vec{s})=\sum_{j=1}^\kappa \nu_{j}$.
We can now prove the following Supplemental Theorem.

\begin{suptheorem}\label{sup-eq:ff-quadratic}
    A Hermitian orthonormal basis for the quadratic symmetries of the matchgate circuit's DLA in Eq.~\eqref{sup-eq:dla} is given by
    \begin{equation} \label{sup-eq:comm-basis} \begin{split}
        &Q_\kappa^0 = \NC_\kappa \sum_{\vec{s}\in\binom{[2n]}{\kappa}} c^{\vec{s}}\otimes c^{\vec{s}} \,,\\ &Q_\kappa^1 = \NC_\kappa\,  (-i)^n i^{\kappa\; {\rm mod}\;2} \sum_{\vec{s}\in\binom{[2n]}{\kappa}} (-1)^{\pi(\vec{s})}c^{\vec{s}}\otimes c^{\vec{\bar{s}}} \,,\end{split}
    \end{equation}
    for integers $0\le \kappa\le 2n$, and $\NC_\kappa =  \left(d \sqrt{\binom{2n}{\kappa}}\right)^{-1}  $.
\end{suptheorem}

\begin{proof}
    We can leverage Lemma~\ref{sup-lem:inv} to  show that the connected components $C_\kappa$ of the commutator graph of the generators $\GC$ of the matchgate circuit's DLA correspond to the Pauli operators in the $\BC_\kappa$ subspaces introduced in Lemma~\ref{sup-lem:modules}. Indeed, Lemma~\ref{sup-lem:inv} states that any pair of Pauli operators in $\BC_\kappa$ can be transformed into each other via commutation with elements in $\g$. Hence all Pauli operators in $\BC_\kappa$ belong to the same connected component of the commuator graph. Moreover, Eq.~\eqref{sup-eq:comm-algebra} implies that Pauli operators in different $\BC_\kappa$ modules cannot be connected via commutation with elements in $\g$.
    On the other hand, it is well known that the only linear symmetries of the matchgate generators $\GC$  are spanned by $I^{\otimes n}$ and the parity operator $P=Z^{\otimes n}$. Hence, Eq.~\eqref{sup-eq:comm-basis} follows from directly applying Eq.~\eqref{sup-eq:quadratic-Pauli-symmetries} and adding a normalization factor $\NC_\kappa$ and a global factor $(-1)^n i^{\kappa\; {\rm mod}\;2}$ to render $Q_\kappa^1$ Hermitian. Indeed, we have
    \begin{equation}
    \Tr[Q^0_\kappa Q^0_{\kappa'}] = \NC_\kappa^2 \Tr\left[\sum_{\vec{s}\in\binom{[2n]}{\kappa}}\sum_{\vec{s}'\in\binom{[2n]}{\kappa}}  c^{\vec{s}} c^{\vec{s}'} \otimes c^{\vec{s}} c^{\vec{s}'}\right]  = \NC_\kappa^2 \,d^2 \sum_{\vec{s}\in\binom{[2n]}{\kappa}}\sum_{\vec{s}'\in\binom{[2n]}{\kappa}}\delta_{\kappa \kappa'} \delta_{\vec{s} \vec{s}'} = \delta_{\kappa \kappa'} \NC_\kappa^2 \,d^2 \sum_{\vec{s}\in\binom{[2n]}{\kappa}} 1 = \delta_{\kappa \kappa'}\,,
\end{equation}
\begin{equation} \begin{split}
    \Tr[Q^1_\kappa Q^1_{\kappa'}] &=(-1)^{n+(\kappa\;{\rm mod\;}2)} \NC_\kappa^2 \Tr\left[\sum_{\vec{s}\in\binom{[2n]}{\kappa}}\sum_{\vec{s}'\in\binom{[2n]}{\kappa}} (-1)^{\pi(\vec{s})+\pi(\vec{s}')}c^{\vec{s}} c^{\vec{s}'} \otimes c^{\vec{\bar{s}}} c^{\vec{\bar{s}}'}\right]  \\& = (-1)^{n+(\kappa\;{\rm mod\;}2)} \NC_\kappa^2 \,d^2 \sum_{\vec{s}\in\binom{[2n]}{\kappa}}\sum_{\vec{s}'\in\binom{[2n]}{\kappa}} (-1)^{\lfloor \frac{\kappa}{2}\rfloor}(-1)^{\lfloor \frac{2n-\kappa}{2}\rfloor}\delta_{\kappa \kappa'} \delta_{\vec{s} \vec{s}'} \\&= \delta_{\kappa \kappa'} \NC_\kappa^2 \,d^2 \sum_{\vec{s}\in\binom{[2n]}{\kappa}} 1 \\&= \delta_{\kappa \kappa'}\,,\end{split}
\end{equation}
\begin{equation}
    \Tr[Q^0_\kappa Q^1_{\kappa'}]  = (-i)^n i^{\kappa\; {\rm mod}\;2} \NC_\kappa^2 \Tr\left[\sum_{\vec{s}\in\binom{[2n]}{\kappa}}\sum_{\vec{s}'\in\binom{[2n]}{\kappa}} (-1)^{\pi(\vec{s}')} c^{\vec{s}} c^{\vec{s}'} \otimes c^{\vec{s}} c^{\vec{\bar{s}}'}\right] =0 \,,
\end{equation}
where we used Eq.~\eqref{sup-eq:ortho-maj-prod}. Moreover, it holds that
\begin{equation}
    \left(Q^0_\kappa\right)^\dagger = \NC_\kappa \sum_{\vec{s}\in\binom{[2n]}{\kappa}} (c^{\vec{s}})^\dagger\otimes (c^{\vec{s}})^\dagger = \NC_\kappa \sum_{\vec{s}\in\binom{[2n]}{\kappa}} (-1)^{\lfloor \frac{\kappa}{2}\rfloor} c^{\vec{s}}\otimes (-1)^{\lfloor \frac{\kappa}{2}\rfloor} c^{\vec{s}} = Q^0_\kappa\,,
\end{equation}
\begin{equation} \begin{split}
    \left(Q^1_\kappa\right)^\dagger &= i^n(-i)^{(\kappa\; {\rm mod}\;2)}\NC_\kappa \sum_{\vec{s}\in\binom{[2n]}{\kappa}} (-1)^{\pi(\vec{s})}(c^{\vec{s}})^\dagger\otimes (c^{\vec{\bar{s}}})^\dagger \\&= i^n (-i)^{(\kappa\; {\rm mod}\;2)}\NC_\kappa \sum_{\vec{s}\in\binom{[2n]}{\kappa}} (-1)^{\pi(\vec{s})} (-1)^{\lfloor \frac{\kappa}{2}\rfloor}c^{\vec{s}}\otimes (-1)^{\lfloor \frac{2n-\kappa}{2}\rfloor}c^{\vec{\bar{s}}} = Q_\kappa^1\,.\end{split}
\end{equation}
\end{proof}

Supplemental Theorem~\ref{sup-eq:ff-quadratic} provides the quadratic symmetries for parametrized matchgate circuits. As indicated above, this is a crucial result that will allow us to compute the exact loss variance via Weingarten calculus (see Section~\ref{sup-sec:th1}). We find it convenient to also work in another basis of the quadratic symmetries, related to the subspaces of well-defined parity in operator space $\BC$. In order to characterize said subspaces, we begin by proving the following Supplemental Lemma.

\begin{suplemma}\label{lem:sep-modules} For even $\kappa$, the subspaces $\BC_\kappa \oplus \BC_{2n-\kappa}$ can be decomposed as
    \begin{equation}
    \BC_\kappa \oplus \BC_{2n-\kappa}=\BC_{\kappa}^{e}\oplus \BC_{\kappa}^{o}\,,
\end{equation}
where $\BC_{\kappa}^{e}$ and $\BC_{\kappa}^{o}$ respectively define the subspaces of operators of even and odd parity spanned by
\begin{equation}\label{eq:operators-definite}
    \begin{split}
      &B_{\kappa,\vec{s}}^e=\frac{i^{\lfloor \frac{\kappa}{2}\rfloor}}{\sqrt{2d}}\left(c^{\vec{s}} + P\cdot c^{\vec{s}}\right)=\frac{i^{\lfloor \frac{\kappa}{2}\rfloor}}{\sqrt{2d}}\left(c^{\vec{s}} +  (-1)^{\pi(\vec{s})} c^{\bar{\vec{s}}}\right)\,,\\
      &B_{\kappa,\vec{s}}^o=\frac{i^{\lfloor \frac{\kappa}{2}\rfloor}}{\sqrt{2d}}\left(c^{\vec{s}} - P\cdot c^{\vec{s}}\right) =\frac{i^{\lfloor \frac{\kappa}{2}\rfloor}}{\sqrt{2d}}\left(c^{\vec{s}} -  (-1)^{\pi(\vec{s})} c^{\bar{\vec{s}}}\right)\,,
  \end{split}
\end{equation}
for $\vec{s}\in\binom{[2n]}{\kappa}$ and $\kappa\in\{2,4,\dots,K\}$, where $K=n-1$ or $K=n-2$ depending on whether $n$ is odd or even, respectively.
\end{suplemma}

\begin{proof}
    First, let us show that the operators $B_{\kappa,\vec{s}}^e$ and $B_{\kappa,\vec{s}}^o$ are orthonormal (it is straightforward to show that they are Hermitian when $\kappa$ is even). We begin by writing these operators as $B^p _{\kappa,\vec{s}}=\frac{i^{\lfloor \frac{\kappa}{2}\rfloor}}{\sqrt{2d}}\left(c^{\vec{s}} +(-1)^{\lambda_p}P\cdot c^{\vec{s}}\right)$, for $\lambda_p=0,1$ corresponding to the even ($p=e$) and odd ($p=o$) operators, respectively. We have
\begin{align}\label{eq:Borth}
    \Tr[B^{\lambda_p} _{\kappa,\vec{s}}B^{\lambda_{p'}} _{\kappa',\vec{s}'} ]
    &=\frac{\left(i^{\lfloor \frac{\kappa}{2}\rfloor}\right)^2}{2d}\left(\Tr[c^{\vec{s}} c^{\vec{s}'}]+(-1)^{\lambda_p +\lambda_{p'}}\Tr[P c^{\vec{s}} P c^{\vec{s}'}]+((-1)^{\lambda_p}+(-1)^{\lambda_{p'}}) \Tr[c^{\vec{s}} P c^{\vec{s}'}]\right)\nonumber\\
    &=\frac{1}{2d}\left(d\, \delta_{\vec{s}\vec{s}'}\left(1+(-1)^{\lambda_p +\lambda_{p'}}\right)\right)=\delta_{\vec{s}\vec{s}'} \delta_{\lambda_p\lambda_{p'}}\,.
\end{align}
Here we have used the fact that $[P,c^{\vec{s}}]=0$ if $\kappa$ is even, $P^2=\id$, as well as Eq.~\eqref{sup-eq:ortho-maj-prod}.  This implies that $\BC_\kappa^e \cap \BC_\kappa^o = \{0\}$. Also, it is straightforward to show that $\BC_\kappa\oplus\BC_{2n-\kappa} =\BC_\kappa^e + \BC_\kappa^o $, which follows from
\begin{equation}
    c^{\vec{s}} \propto B_{\kappa,\vec{s}}^e +B_{\kappa,\vec{s}}^o \,,\quad\; c^{\bar{\vec{s}}}\propto B_{\kappa,\vec{s}}^e -B_{\kappa,\vec{s}}^o\,,
\end{equation}
which indicates that we can recover a basis for $B_\kappa$ (consisting of all different products of $\kappa$ distinct Majoranas) and a basis for $B_{2n-\kappa}$ (consisting of all different products of $2n-\kappa$ distinct Majoranas) from the $\BC_\kappa^e, \BC_\kappa^o$ operators. Hence, using Supplemental Definition~\ref{sup-def:direct_sum} we find that $\BC_\kappa\oplus\BC_{2n-\kappa} =\BC_\kappa^e \oplus \BC_\kappa^o $.
Finally, it is easy to show that these operators have definite parity as $P B_{\kappa,\vec{s}}^e = B_{\kappa,\vec{s}}^e P = B_{\kappa,\vec{s}}^e$ and $P B_{\kappa,\vec{s}}^o = B_{\kappa,\vec{s}}^o P = -B_{\kappa,\vec{s}}^o$.
\end{proof}

We are now ready to prove the following result, which provides an alternative basis for the quadratic symmetries of parametrized matchgate circuits.
\begin{suplemma}
    The operators 
\begin{equation} \label{sup-eq:Q-parity-basis}\begin{split}
Q^{++}_{\kappa}&:=\MC_\kappa\sum_{\vec{s}\in\binom{[2n]}{\kappa}} \left(c^{\vec{s}}+Pc^{\vec{s}}\right) \otimes \left(c^{\vec{s}}+ P c^{\vec{s}}\right)\,, \\
Q^{+-}_{\kappa}&:=\MC_\kappa\sum_{\vec{s}\in\binom{[2n]}{\kappa}} \left(c^{\vec{s}}+Pc^{\vec{s}}\right) \otimes \left(c^{\vec{s}}- Pc^{\vec{s}}\right)\,, \\ 
Q^{-+}_{\kappa}&:=\MC_\kappa\sum_{\vec{s}\in\binom{[2n]}{\kappa}} \left(c^{\vec{s}}-Pc^{\vec{s}}\right) \otimes \left(c^{\vec{s}}+ Pc^{\vec{s}}\right)\,, \\
Q^{--}_{\kappa}&:=\MC_\kappa\sum_{\vec{s}\in\binom{[2n]}{\kappa}} \left(c^{\vec{s}}-Pc^{\vec{s}}\right) \otimes \left(c^{\vec{s}}- Pc^{\vec{s}}\right)\,,\end{split}
\end{equation}
where $\kappa\in\{0,1,\dots,n\}$ and $\MC_\kappa=\NC_\kappa / 2^{1+\delta_{\kappa n}/2}$, are an orthonormal basis of the quadratic symmetries of the DLA in Eq.~\eqref{sup-eq:dla}.
\end{suplemma}
\begin{proof}

Let us first rewrite the operators in~\eqref{sup-eq:Q-parity-basis} as 
\begin{equation}
    Q_\kappa^{\lambda\gamma} =\MC_\kappa\sum_{\vec{s}\in\binom{[2n]}{\kappa}}  \left(c^{\vec{s}}+(-1)^\lambda Pc^{\vec{s}}\right) \otimes \left(c^{\vec{s}}+ (-1)^\gamma Pc^{\vec{s}}\right)\,.
\end{equation}
We now prove that these operators 
are orthonormal. Indeed, it follows that when $\kappa\neq n$, we have
\begin{equation}  \label{sup-eq:parity-basis-ortho}\begin{split}
        \Tr[(Q_\kappa^{\lambda\gamma})^\dag Q_{\kappa'}^{\lambda'\gamma'}] &=\MC_\kappa^2 \Tr\left[\sum_{\vec{s}\in\binom{[2n]}{\kappa}} \sum_{\vec{s}'\in\binom{[2n]}{\kappa}}   \left((1+(-1)^{\lambda +\lambda'}) c^{\vec{s}} c^{\vec{s}'}  +( (-1)^\lambda + (-1)^{\lambda'} )c^{\vec{s}}P c^{\vec{s}'} \right)\right. \\ &\qquad \qquad \qquad \qquad\qquad\; \otimes \left.\left((1+(-1)^{\gamma +\gamma'})  c^{\vec{s}} c^{\vec{s}'}   + ( (-1)^\gamma + (-1)^{\gamma'} ) c^{\vec{s}}P c^{\vec{s}'} \right)\right] \\ &= \MC_\kappa^2 \,d^2 \sum_{\vec{s}\in\binom{[2n]}{\kappa}} \sum_{\vec{s}'\in\binom{[2n]}{\kappa}} 4\,\delta_{\kappa \kappa'} \delta_{\vec{s}\vec{s}'} \delta_{\lambda \lambda'} \delta_{\gamma \gamma'} = \delta_{\kappa \kappa'} \delta_{\lambda \lambda'} \delta_{\gamma \gamma'} \NC_\kappa^2 \,d^2 \sum_{\vec{s}\in\binom{[2n]}{\kappa}}  1 = \delta_{\kappa \kappa'} \delta_{\lambda \lambda'} \delta_{\gamma \gamma'} \,. \end{split}
    \end{equation}
Furthermore, when $\kappa= \kappa'=n$ there is an additional non-zero contribution in~\eqref{sup-eq:parity-basis-ortho}, namely
\begin{equation}
   \MC_\kappa^2\Tr\left[\sum_{\vec{s}\in\binom{[2n]}{\kappa}}  ( (-1)^\lambda + (-1)^{\lambda'} )c^{\vec{s}}P c^{\vec{\bar{s}}}\otimes ( (-1)^\gamma + (-1)^{\gamma'} ) c^{\vec{s}}P c^{\vec{\bar{s}}}\right] = \NC_\kappa^2 d^2 \delta_{\lambda \lambda'} \delta_{\gamma \gamma'} (-1)^{\lambda+\gamma}\sum_{\vec{s}\in\binom{[2n]}{\kappa}}  1 =\frac{(-1)^{\lambda+\gamma}}{2}\delta_{\lambda \lambda'} \delta_{\gamma \gamma'}\,.\nonumber
\end{equation}
This completes the proof of the orthonormality of the $Q_\kappa^{\lambda \gamma}$ operators. We note that it also implies that the operators $Q_n^{+-}$ and $Q_n^{-+}$ are equal to the zero operator.

Next, we prove that the $Q_\kappa^{\lambda \gamma}$ are a basis of the commutant of $G^{\otimes 2}$. We have
    \begin{equation} \label{sup-eq:change-of-basis} \begin{split}
        Q_\kappa^0 &= 2^{\delta_{\kappa n}/2}\, \frac{Q_\kappa^{++} + Q_\kappa^{+-} + Q_\kappa^{-+} + Q_\kappa^{--}}{2} \,, \\ Q_{2n-\kappa}^0 &=2^{\delta_{\kappa n}/2}\, \frac{Q_\kappa^{++} - Q_\kappa^{+-} - Q_\kappa^{-+} + Q_\kappa^{--}}{2}\,,  \\ Q_\kappa^1 &= i^{\kappa\; {\rm mod}\;2} \,2^{\delta_{\kappa n}/2}\, \frac{Q_\kappa^{++} - Q_\kappa^{+-} + Q_\kappa^{-+} - Q_\kappa^{--}}{2}\,, \\ Q_{2n-\kappa}^1 &= (-1)^n\, i^{\kappa\; {\rm mod}\;2}\, 2^{\delta_{\kappa n}/2}\,\frac{Q_\kappa^{++} + Q_\kappa^{+-} - Q_\kappa^{-+} - Q_\kappa^{--}}{2}\,. \end{split}
    \end{equation}
That is, we can obtain the $Q_\kappa^0$, $Q_\kappa^1$ operators in Eq.~\eqref{sup-eq:comm-basis} as linear combinations of the $Q_\kappa^{\lambda \gamma}$. Since both the former and the latter sets are orthonormal sets, and they both contain the same number of operators, it immediately follows that their span is identical, i.e., they are both bases for the commutant of $G^{\otimes 2}$.

\end{proof}

\section{Proof of the general Variance expression (Theorem 1)} \label{sup-sec:th1}
In this section, we provide a proof for our main result, that is, Theorem~\ref{theo:main}.
Let us first recall a few basic concepts.
Variational quantum computing schemes employ parameterized quantum circuits and aim at minimizing a loss function of the form
\begin{equation}\label{eq:loss-SI}
    \ell_{\thv}(\rho,O)=\Tr[U(\thv)\rho U\ad(\thv) O]\,,
\end{equation}
where $\rho$ is a quantum state acting on an $n$-qubit  Hilbert space $\HC=(\mathbb{C})^{\otimes 2}$, and $O$ is a Hermitian observable (with $\norm{O}_2^2\leq 2^n)$. As discussed in the main text, in order to study the presence of barren plateaus, we need to be able to compute the variance
\begin{equation}
    \Var_{\thv}[\ell_{\thv}(\rho,O)]=\mathbb{E}_{\thv}[\ell_{\thv}(\rho,O)^2]-\mathbb{E}_{\thv}[\ell_{\thv}(\rho,O)]^2\,.
\end{equation}
In this section we use the quadratic symmetries (i.e. the orthonormal bases of ${\rm comm}(G^{\otimes 2})$) obtained above to compute the variance via the Weingarten calculus. Since we have obtained two different bases for the quadratic symmetries (namely~\eqref{sup-eq:comm-basis} and~\eqref{sup-eq:Q-parity-basis}, which will  give us different insights),  we will  present two separated but equivalent expressions for the variance. First, we recover the result in the main text, and then we present the alternative version.

 \subsection{Variance expression, proof of Theorem 1}
 Let us start with the derivation of Theorem~\ref{theo:main}, which we restate for  convenience  here.
\setcounter{theorem}{0}

\begin{theorem}\label{sup-th:theorem1}
The loss function in Eq.~\eqref{eq:loss-SI} has mean
\begin{align}
    \mathbb{E}_{\thv}[\ell_{\thv}(\rho,O)]&\!=\sum_{\kappa=0,2n}\!\langle\rho_\kappa,O_\kappa\rangle_{\id+P}\,,
\end{align}
and variance
\small
\begin{equation} \label{sup-eq:theo-1}
\begin{split}
    \Var_{\thv}[\ell_{\thv}(\rho,O)]=   \sum_{\kappa=1}^{2n-1} &\frac{\PC_\kappa(\rho)\PC_\kappa(O)+\CC_\kappa(\rho)\CC_\kappa(O)}{\dim(\BC_\kappa)}\,.
\end{split}
\end{equation}
\normalsize
Here we defined the $\kappa$-purity of an operator $M\in\BC$ as $\PC_\kappa(M)=\langle M_\kappa,M_\kappa\rangle_{\id}$, and its $\kappa$-coherence  as $\CC_\kappa(M)=i^{\kappa \,{\rm mod\, 2}}\langle M_\kappa,M_{2n-\kappa}\rangle_{P}$. Moreover, for $M_1,M_2,\Gamma\in\BC$, we have $\langle M_1,M_2\rangle_{\Gamma}=\Tr[\Gamma M_1\ad M_2]$.
\end{theorem}

To prove Theorem \ref{sup-th:theorem1} we are going to introduce several partial results and Supplemental Lemmas.
Let us begin by computing the expectation value of the loss function.
    In this case, knowing that the linear symmetries are spanned by $\id$ and $P$, we can use the Weingarten calculus (see Eq.~\eqref{eq:weingarten}) to find 
\begin{equation}\label{eq:exp-proof}
    \mathbb{E}_{\thv}[\ell_{\thv}(\rho,O)]=\Tr[\rho\frac{\id}{\sqrt{d}}]\Tr[O \frac{\id}{\sqrt{d}}]+\Tr[\rho \frac{P}{\sqrt{d}}]\Tr[O \frac{P}{\sqrt{d}}]=\sum_{\kappa=0,2n} \langle \rho_\kappa, O_\kappa\rangle_{\id+P}\,,
\end{equation}
 where  in the second equality we used that $O_0=\Tr[O \id]\frac{\id}{d}$, $O_{2n}=\Tr[O P]\frac{P}{d}$, and analogously for $\rho$.

 In order to determine the expression for the variance, we first present two useful Supplemental Lemmas.  
\begin{suplemma}
The mean value squared of the loss function~\eqref{eq:loss-SI} takes the form
    \begin{equation} \label{sup-eq:lemma8}
    \mathbb{E}_{\thv}[\ell_{\thv}(\rho,O)]^2= \mathcal{P}_0(\rho)\mathcal{P}_0(O)+\mathcal{C}_0(\rho)\mathcal{C}_0(O)+\mathcal{P}_{2n}(\rho)\mathcal{P}_{2n}(O)+\mathcal{C}_{2n}(\rho)\mathcal{C}_{2n}(O)\,.
\end{equation}
\end{suplemma}
\begin{proof}
From Eq.~\eqref{eq:exp-proof} we have
\begin{equation}
 \mathbb{E}_{\thv}[\ell_{\thv}(\rho,O)]^2=   \frac{(\Tr[\rho\id])^2(\Tr[O \id])^2}{d^2}+\frac{(\Tr[\rho P])^2(\Tr[O P])^2}{d^2}+2 \frac{\Tr[\rho\id]\Tr[\rho P]\Tr[O \id]\Tr[O P]}{d^2}\,.
\end{equation}
Notably, we can relate each term to either generalized purities or coherences. Notice that in fact $\mathcal{P}_0(O)=\Tr[O_0^2]=\Tr[\left(\Tr[O\id]\frac{\id}{d}\right)^2]=\frac{\left(\Tr[O\id]\right)^2}{d}$, with a similar expression for $\kappa=2n$ but replacing $\id$ with $P$ instead. On the other hand, $$\mathcal{C}_{0}(O)=\Tr[P O_{2n}O_0]=\Tr[P\frac{P}{d}\frac{\id}{d}]\Tr[O\id]\Tr[OP]=\frac{\Tr[O\id]\Tr[OP]}{d}=\mathcal{C}_{2n}(O)\,.$$ 
Putting all together one recovers Eq.~\eqref{sup-eq:lemma8}.
\end{proof}

Additional results are needed to evaluate $\mathbb{E}_{\thv}[\ell_{\thv}(\rho,O)^2]$. This motivates the following Supplemental Lemma.

\begin{suplemma}
   The Hilbert-Schmidt inner product between $O^{\otimes 2}$ and the quadratic symmetry  $Q^\lambda_\kappa$ in Eq.~\eqref{sup-eq:comm-basis} can be rewritten as
    \begin{equation}
        \Tr[O^{\otimes 2}Q_\kappa^0]= \frac{(-1)^{\lfloor \frac{\kappa}{2}\rfloor}}{\sqrt{\dim(\BC_\kappa})}  \mathcal{P}_\kappa(O)\,,\;\;\;\;\; \Tr[O^{\otimes 2}Q_\kappa^1]= \frac{(-1)^{\lfloor \frac{\kappa}{2}\rfloor}}{\sqrt{\dim(\BC_\kappa})}  \mathcal{C}_\kappa(O)\,.
    \end{equation}
\end{suplemma}

\begin{proof}
Notice first that the projection of $O$ onto a module $\kappa$ has the form $ O_\kappa=i^{\lfloor \frac{\kappa}{2}\rfloor}\sum_{\nu_1<\dots <\nu_\kappa} O_{\nu_1\dots \nu_\kappa}\frac{c_{\nu_1}\cdots c_{\nu_\kappa}}{\sqrt{d}}$ with $O_{\nu_1\dots \nu_\kappa}\in \mathbb{R}$ an anti-symmetric tensor under the permutation of its indices. This means that
   $\Tr [O_\kappa^{\otimes 2} Q^\lambda_{\kappa'}]=0$ for $\kappa'\neq \kappa$ as a consequence of the orthogonality condition between different modules. This allows us to write 
   \begin{equation} \label{sup-eq:tracec_O2_Q0} \begin{split}
    \Tr[O^{\otimes 2} Q^0_{\kappa}] &=(-1)^{\lfloor \frac{\kappa}{2}\rfloor}\mathcal{N}_\kappa \sum_{\substack{\nu_1< \dots< \nu_\kappa\\ \nu'_1< \dots< \nu'_\kappa}} O_{\nu_1\dots \nu_\kappa}  O_{\nu'_1\dots \nu'_\kappa} \sum_{\vec{s}\in\binom{[2n]}{\kappa}} \Tr[\frac{c_{\nu_1}\cdots c_{\nu_\kappa} \,c^{\vec{s}}}{\sqrt{d}}]  \Tr[\frac{c_{\nu'_1}\cdots c_{\nu'_\kappa} \,c^{\vec{s}}}{\sqrt{d}}] \\ &=(-1)^{\lfloor \frac{\kappa}{2}\rfloor}\mathcal{N}_\kappa \sum_{\substack{\nu_1< \dots< \nu_\kappa\\ \nu'_1< \dots< \nu'_\kappa}} O_{\nu_1\dots \nu_\kappa}  O_{\nu'_1\dots \nu'_\kappa} \sum_{\vec{s}\in\binom{[2n]}{\kappa}} d \,\delta_{\nu_1, \dots, \nu_\kappa; \vec{s}} \delta_{\nu'_1, \dots, \nu'_\kappa; \vec{s}} \\ &= (-1)^{\lfloor \frac{\kappa}{2}\rfloor} d\, \mathcal{N}_\kappa \sum_{\nu_1< \dots< \nu_\kappa} O_{\nu_1\dots \nu_\kappa}^2 = (-1)^{\lfloor \frac{\kappa}{2}\rfloor} d\, \mathcal{N}_\kappa \mathcal{P}_\kappa(O)=\frac{(-1)^{\lfloor \frac{\kappa}{2}\rfloor}}{\sqrt{\dim(\BC_\kappa})}  \mathcal{P}_\kappa(O)\,.
\end{split}    
\end{equation}
This proves the first equality. 
In order to prove the second let us notice first that $Q^{1}_{\kappa}$ can be related to $Q^{0}_{\kappa}$ as follows
\begin{equation}
Q^{1}_{\kappa}=\mathcal{N}_{\kappa}i^{\kappa \,{\rm mod\, 2}}\sum_{\vec{s}\in\binom{[2n]}{\kappa}}  c^{\vec{s}} \otimes P c^{\vec{s}}=i^{\kappa \,{\rm mod\, 2}}(\id \otimes P)Q^{0}_{\kappa}\,.
\end{equation}
We can use this and the previous result to write
\begin{equation} \begin{split}
    \Tr[O^{\otimes 2} Q^1_{\kappa}]&=i^{\kappa \,{\rm mod\, 2}}\Tr[(O\otimes O)(\id \otimes P) Q^0_{\kappa}]=i^{\kappa \,{\rm mod\, 2}}\Tr[(O\otimes OP) Q^0_{\kappa}]=\frac{(-1)^{\lfloor \frac{\kappa}{2}\rfloor}i^{\kappa \,{\rm mod\, 2}}}{\sqrt{\dim(\BC_\kappa})} \Tr[PO^2]\\
    &=\frac{(-1)^{\lfloor \frac{\kappa}{2}\rfloor}}{\sqrt{\dim(\BC_\kappa})} \mathcal{C}_\kappa(O)\,,
\end{split}    
\end{equation}
where we used that $\Tr[O\otimes O' Q^0_{\kappa}]=\frac{(-1)^{\lfloor \frac{\kappa}{2}\rfloor}}{\sqrt{\dim(\BC_\kappa})} \Tr[OO']$ which can be proven in complete analogy with the previous. The last equality follows from the definition of $\mathcal{C}_\kappa(O)$.
\end{proof}

We are now in a position to readily prove the second part of Theorem~\ref{sup-th:theorem1}.
By using the orthonormal basis $Q_\kappa^\lambda$ of ${\rm comm}(G^{\otimes 2})$ from Eq.~\eqref{sup-eq:comm-basis}, we know that the variance takes the simple form
\begin{equation}
      \Var_{\thv}[C(\thv)]=\sum_{\kappa=0}^{2n} \left(\Tr[O^{\otimes 2}Q_\kappa^0]\Tr[\rho^{\otimes 2}Q_\kappa^0]+\Tr[O^{\otimes 2}Q_\kappa^1]\Tr[\rho^{\otimes 2}Q_\kappa^1]\right)-\mathbb{E}_{\thv}[\ell_{\thv}(\rho,O)]^2\,.
\end{equation}
All of these terms have been related to generalized purities and/or coherences by the previous Supplemental Lemmas. Notice that the terms with $\kappa=0,2n$ are precisely canceled by the mean value squared. 
The final expression of Theorem~\ref{sup-th:theorem1} is then automatically recovered, with the first term in the sum giving rise to purities while the second is responsible for the coherences.

\subsection{Alternative variance expression in the parity-preserving basis}

Let us re-derive the formula for the loss variance using the operators in~\eqref{sup-eq:Q-parity-basis} (with even $\kappa$) as quadratic symmetries. As explained above, this basis is directly related to the parity-preserving subspaces in operator space. That is, from Supplemental Lemma~\ref{lem:sep-modules} we know that the modules with $\kappa$ even can be decomposed as a direct sum, as $\BC_\kappa \oplus \BC_{2n-\kappa}=\BC_{\kappa}^{e}\oplus \BC_{\kappa}^{o}$. This implies that the operator space $\BC$ can be broken down in the following way,
\begin{equation} \label{sup-eq:parity-modules}
    \BC =\bigoplus_{q=0}^{\lfloor \frac{n}{2}\rfloor}\left( \BC_{2q}^{e}\oplus \BC_{2q}^{o}\right) \bigoplus_{l=0}^{n-1} \BC_{2l+1}\,.
\end{equation}
The  basis of the quadratic symmetries associated with this decomposition can be re-written for convenience as
\begin{equation} \label{sup-eq:comm-parity-basis}
    Q_\kappa^{\lambda \gamma}=\frac{(-1)^{\lfloor \frac{\kappa}{2}\rfloor} }{2^{\delta_{\kappa n}/2}\sqrt{\dim(\BC_\kappa)}}\sum_{\vec{s}\in\binom{[2n]}{\kappa}}  B_{\kappa,\vec{s}}^\lambda\otimes B_{\kappa,\vec{s}}^\gamma\,,
\end{equation}
where $k\in \{0,1,\dots, n-1\}$ (the case $\kappa=n$ requires a special treatment, as shown below).
We recall that the $B_{\kappa,\vec{s}}^\lambda$ are the elements of the basis of $\BC_{\kappa}^p$ (see~\eqref{eq:operators-definite}), and that we are only interested in computing the contributions coming from even values of $\kappa$ (i.e. $\kappa \equiv 2q$).
This follows from the fact that since we know that covariance terms can only appear between modules $\kappa$ and $\kappa'=2n-\kappa$, we can use~\eqref{sup-eq:theo-1} to express the contributions corresponding to the modules with odd $\kappa$. Then, we will employ~\eqref{sup-eq:comm-parity-basis} to compute the contributions from the modules of even $\kappa$. 
Notably,  we find  that the expression for the variance has the same structure as before. However,
one motivation for using this alternative basis is that fermionic states and operators are such that they have definite parity, i.e. they only have support on either $\{\BC_{2q}^e\}$ or $\{\BC_{2q}^o\}$, which implies vanishing coherences, as we will show below.

We will now prove the following Supplemental Theorem.

\begin{suptheorem} \label{sup-th:parity-basis}
    Under the module separation~\eqref{sup-eq:parity-modules}, the mean value of the loss function in Eq.~\eqref{eq:loss-SI} takes the form
\begin{align} \label{sup-eq:parity-basis-mean}
    \mathbb{E}_{\thv}[\ell_{\thv}(\rho,O)]&\!=\sum_{p=e,o}\!\langle\rho_0^p,O_0^p\rangle_{\id}\,.
\end{align}
Moreover, if either the initial state $\rho$ or the measurement $O$ are fermionic (i.e. either $[\rho,P]=0$ or $[O,P]=0$),
the variance is given by
\begin{equation}
    \Var_{\thv}[\ell_{\thv}(\rho,O)]=\sum_{q=1}^{\lfloor \frac{n}{2} \rfloor} \,\sum_{p=e,o} \frac{\PC_{2q,p}(\rho)\PC_{2q,p}(O)}{\dim(\BC_{2q}^p)}\,,
\end{equation}
where we defined the generalized $(\kappa,p)$-purity of and operator $M\in\BC$ as $\mathcal{P}_{\kappa,p}(M)=\langle M_{\kappa,p}, M_{\kappa,p}\rangle_{\id}$.
\end{suptheorem}

In order to prove Supplemental Theorem~\ref{sup-th:parity-basis}, we introduce several Supplemental Lemmas and partial results.
First, let us consider the expression for the mean of the loss function. 
  As before, we know that the mean value projects onto the commutant of the dynamical Lie group. An orthonormal basis of the commutant which preserves the module structure is provided by the  two operators $B_0^p=\frac{\id +(-1)^{\lambda_p} P}{\sqrt{2d}}$, with $\lambda_p=0,1$ for $p=e,o$ respectively (we omit the index $\vec{s}$ since $k=0$).   Then we have
\begin{equation}
     \mathbb{E}_{\thv}[\ell_{\thv}(\rho,O)]=\Tr[\rho B_0^e]\Tr[O B_0^e]+\Tr[\rho B_0^o]\Tr[O B_0^o]\,.
\end{equation}
On the other hand, since $\rho^p_0=\Tr[\rho B_0^p]B_0^p$ and $O^p_0=\Tr[O B_0^p]B_0^p$, we have $\langle \rho^p_0, O^p_0\rangle_{\id}=\Tr[\rho B_0^p]\Tr[O B_0^p]$, thus leading to the desired expression~\eqref{sup-eq:parity-basis-mean}.

We find it convenient to introduce two additional Supplemental Lemmas before proceeding to the evaluation of the variance.

\begin{suplemma}\label{suplemma:meannewbasis}
The mean value squared of the loss function takes the form
    \begin{equation}
    \mathbb{E}_{\thv}[\ell_{\thv}(\rho,O)]^2= \sum_{p=e,o} \mathcal{P}_{0,p}(\rho)\mathcal{P}_{0,p}(O)+\mathcal{C}_{0,p}(\rho)\mathcal{C}_{0,p}(O)\,,
\end{equation} 
where we defined the generalized coherences as 
\begin{equation}
    \mathcal{C}_{\kappa,p}(O)= \Tr[T_\kappa(O_{\kappa,p})O_{\kappa,\bar{p}}]\,,
\end{equation}
and we introduced the linear map $T_\kappa$ acting on a Hermitian operator $O$ as follows,
    \begin{equation}\label{sup-eq:T-map}
    T_\kappa(O):= (-1)^{\lfloor \frac{\kappa}{2}\rfloor}\sqrt{\dim(\BC_\kappa)}\Tr_1\left[(O\otimes \id) (Q^{+-}_\kappa + Q^{-+}_\kappa)\right]\,.
    \end{equation}
\end{suplemma}

\begin{proof}
    We first notice that the square of the mean is given by
    \begin{equation}
     \mathbb{E}_{\thv}[\ell_{\thv}(\rho,O)]^2=\Tr[\rho B_0^e]^2\Tr[O B_0^e]^2+\Tr[\rho B_0^o]^2\Tr[O B_0^o]^2+2\Tr[\rho B_0^e]\Tr[\rho B_0^o]\Tr[O B_0^e]\Tr[O B_0^o]\,.
\end{equation}
The squared terms have a very clear interpretation: we recall that $\mathcal{P}_{0,p}=\Tr\left[O B_0^p\right]^2$, i.e., all the squared terms are generalized purities.
The cross terms require a little more analysis. We start with
\begin{equation}
     \Tr_1\left[(O_\kappa\otimes \id) Q^0_\kappa\right]=\frac{ (-1)^{\lfloor \frac{\kappa}{2}\rfloor}}{\sqrt{\dim(\BC_\kappa)}} \, O_\kappa\,,
\end{equation}
which follows from a very similar calculation to that in Eq.~\eqref{sup-eq:tracec_O2_Q0}. Indeed, we have
\begin{equation} \begin{split}
    \Tr_1\left[(O_\kappa\otimes \id) Q^0_\kappa\right] &=i^{\lfloor \frac{\kappa}{2}\rfloor}\mathcal{N}_\kappa \sum_{\nu_1< \dots< \nu_\kappa} O_{\nu_1\dots \nu_\kappa}  \sum_{\vec{s}\in\binom{[2n]}{\kappa}} \Tr[\frac{c_{\nu_1}\cdots c_{\nu_\kappa} \,c^{\vec{s}}}{\sqrt{d}}]  c^{\vec{s}} \\ &=(-i)^{\lfloor \frac{\kappa}{2}\rfloor}\mathcal{N}_\kappa \sum_{\nu_1< \dots< \nu_\kappa} O_{\nu_1\dots \nu_\kappa}  \sum_{\vec{s}\in\binom{[2n]}{\kappa}} \sqrt{d} \,\delta_{\nu_1, \dots, \nu_\kappa; \vec{s}} \;c^{\vec{s}} \\ &= \frac{(-i)^{\lfloor \frac{\kappa}{2}\rfloor}}{\sqrt{\dim(\BC_\kappa})}  \sum_{\nu_1< \dots< \nu_\kappa} O_{\nu_1\dots \nu_\kappa} \frac{c_{\nu_1\dots\nu_\kappa}}{\sqrt{d}} = \frac{(-1)^{\lfloor \frac{\kappa}{2}\rfloor}}{\sqrt{\dim(\BC_\kappa})}\, O_\kappa \,.
\end{split}    
\end{equation}
We then recall~\eqref{sup-eq:change-of-basis}, which tells us that $Q^{+-}_\kappa+Q^{-+}=Q^0_\kappa-Q^0_{2n-\kappa}$ (when $\kappa\neq n$). Hence, it holds that
\begin{equation}
    (-1)^{\lfloor \frac{\kappa}{2}\rfloor}\sqrt{\dim(\BC_\kappa)}\,\Tr_1\left[(O\otimes \id) (Q^{+-}_\kappa + Q^{-+}_\kappa)\right]=O_\kappa-O_{2n-\kappa}\,.
\end{equation}
Using the definition~\eqref{sup-eq:T-map},
we therefore arrive at
\begin{equation}
    T_\kappa(O)=O_\kappa-O_{2n-\kappa}\,.
\end{equation}
We notice that in this intermediate step we are using the previous module decomposition, $\BC=\bigoplus_\kappa \BC_\kappa$.
We conclude that the action of the linear map $T_\kappa$ on the basis elements $B_{\kappa,\vec{s}}^p$~\eqref{eq:operators-definite} of the parity-preserving modules is
\begin{equation}\label{sup-eq:T-action}
    T_\kappa(B_{\kappa,\vec{s}}^p)=\frac{i^{\lfloor \frac{\kappa}{2}\rfloor}}{\sqrt{2d}}\big(T_\kappa(c^{\vec{s}}) + (-1)^{\lambda_p} T_\kappa(P c^{\vec{s}})\big)=\frac{i^{\lfloor \frac{\kappa}{2}\rfloor}}{\sqrt{2d}}\left(c^{\vec{s}} - (-1)^{\lambda_p} P c^{\vec{s}}\right)=B_{\kappa,\vec{s}}^{\bar{p}}\,,
\end{equation}
where we used that $P c^{\vec{s}}\in \BC_{2n-\kappa}$. Now, we can use that, by definition, $\langle T_0(O^e_0),O^o_0\rangle_{\id}=\mathcal{C}_{0,e}$ to show that $\Tr[O B_0^e] \Tr[O B_0^o]=\mathcal{C}_{0,e}(O)=\mathcal{C}_{0,o}(O)$.

\end{proof}
Let us add a brief comment about the map $T_\kappa$. In the previous subsection, we defined the generalized coherence associated with the module separation $\BC=\bigoplus_\kappa \BC_\kappa$ as $\mathcal{C}_\kappa(O)=\Tr[PO_{2n-\kappa}O_\kappa]$. This is so because in this special case we can just write $P_\kappa(O)=PO$, which follows from defining $ P_\kappa(O)\propto\Tr_1[(O\otimes \id) Q^1_\kappa ]$.

We present another useful Supplemental Lemma.

\begin{suplemma}\label{suplemma:newbasisoverlap}
 The Hilbert-Schmidt inner product between $O^{\otimes 2}$ and the quadratic symmetry $Q^{\lambda \gamma}_\kappa$ in Eq.~\eqref{sup-eq:Q-parity-basis}, for $\kappa$ even and $\kappa \in \{2,4,\dots, K \}$, where $K=n-1$ or $K=n-2$ depending on whether $n$ is odd or even respectively, can be rewritten as 
    \begin{equation}
        \Tr[O^{\otimes 2}Q_\kappa^{pp}]= \frac{(-1)^{\lfloor \frac{\kappa}{2}\rfloor}}{\sqrt{\dim(\BC_\kappa})}  \mathcal{P}_{\kappa,p}(O)\,,\;\;\;\;\;\Tr[O^{\otimes 2}Q_\kappa^{p\bar{p}}]= \frac{(-1)^{\lfloor \frac{\kappa}{2}\rfloor}}{\sqrt{\dim(\BC_\kappa})}  \mathcal{C}_{\kappa,p}(O)\,.
    \end{equation}
\end{suplemma}

\begin{proof}
Consider the expansion of a general operator in the even $\kappa$ sectors: $O=\sum_{\kappa,p}\sum_{\vec{s}\in\binom{[2n]}{\kappa}} O^p_{\vec{s}}B^p_{\kappa,\vec{s}}+\text{odd contributions}$, where $O^p_{\vec{s}}$ is a real anti-symmetric tensor. We can ignore the 
odd contributions, whose overlap with the operators $Q_\kappa^{\lambda\gamma}$ is zero for $\kappa$ even. We then have 
\begin{equation}
    \begin{split}
        \Tr[O^{\otimes 2}Q^{\lambda \gamma}_{\kappa''}]&=\frac{(-1)^{\lfloor \frac{\kappa''}{2}\rfloor}}{\sqrt{\dim(\BC_\kappa})}\sum_{\kappa,\kappa',p,p'}\sum_{\vec{s},\vec{s}',\vec{s}''}O^p_{\vec{s}}O^{p'}_{\vec{s}'}\Tr[B_{k,\vec{s}}^p B_{k'',\vec{s}''}^\lambda]\Tr[B_{k',\vec{s}'}^{p'} B_{k'',\vec{s}''}^\gamma]\\&=\frac{(-1)^{\lfloor \frac{\kappa}{2}\rfloor}}{\sqrt{\dim(\BC_\kappa})}\sum_{\vec{s}\in\binom{[2n]}{\kappa_0}} O^\lambda_{\vec{s}}O^\gamma_{\vec{s}}\,,
    \end{split}
\end{equation}
where we used Eqs.~\eqref{eq:Borth} and~\eqref{sup-eq:comm-parity-basis}.
If we set $\lambda=\gamma\equiv p$ we thus obtain 
\begin{equation}
    \Tr[O^{\otimes 2}Q^{pp}_{\kappa}]=\frac{(-1)^{\lfloor \frac{\kappa}{2}\rfloor}}{\sqrt{\dim(\BC_\kappa})}\sum_{\vec{s}\in\binom{[2n]}{\kappa}} O^p_{\vec{s}}O^p_{\vec{s}}=\frac{(-1)^{\lfloor \frac{\kappa}{2}\rfloor}}{\sqrt{\dim(\BC_\kappa})}\mathcal{P}_{\kappa,p}(O)\,.
\end{equation}
The last equality follows immediately by noting that $\mathcal{P}_{\kappa,p}(O)=\Tr[O_{\kappa,p}^2]=\sum_{\vec{s}\in\binom{[2n]}{\kappa}} O^p_{\vec{s}}O^p_{\vec{s}}$, with $O_{\kappa,p}=\sum_{\vec{s}\in\binom{[2n]}{\kappa}} O^p_{\vec{s}}B_{\kappa,\vec{s}}^p$ the projection of $O$ onto the subspace $\BC_{\kappa}^p$. This proves the first part of the lemma.

The second part follows by setting $\lambda=p$ and $\gamma=\bar{p}$, from which we have that $\Tr[O^{\otimes 2}Q^{p\bar{p}}_{\kappa}]=\frac{(-1)^{\lfloor \frac{\kappa}{2}\rfloor}}{\sqrt{\dim(\BC_\kappa})}\sum_{\vec{s}\in\binom{[2n]}{\kappa}} O^p_{\vec{s}} O^{\bar{p}}_{\vec{s}}$. In order to show that this can be rewritten as a generalized coherence, let us recall that  $T_{\kappa}(B_{\kappa,\vec{s}}^{p})=B_{\kappa,\vec{s}}^{\bar{p}}$ (see Eq.~\eqref{sup-eq:T-action}), implying that
\begin{equation}
   \mathcal{C}_{\kappa,p}(O)=\Tr[T_\kappa(O_{\kappa,p})O_{\kappa,\bar{p}}]=\sum_{\vec{s},\vec{s}' \in\binom{[2n]}{\kappa}} O^{p}_{\vec{s}}O^{\bar{p}}_{\vec{s}'}\Tr[T_\kappa(B_{\kappa,\vec{s}}^{p})B_{\kappa,\vec{s}'}^{\bar{p}}]=\sum_{\vec{s},\vec{s}' \in\binom{[2n]}{\kappa}} O^{p}_{\vec{s}}O^{\bar{p}}_{\vec{s}'}\Tr[B_{\kappa,\vec{s}}^{\bar{p}}B_{\kappa,\vec{s}'}^{\bar{p}}]=\sum_{\vec{s}\in\binom{[2n]}{\kappa}} O^{\bar{p}}_{\vec{s}}O^p_{\vec{s}}\,. 
\end{equation}
Hence the result is proven. Notice also that $C_{\kappa,p}=C_{\kappa,\bar{p}}$, while $C_{n,p}=0$ since $T_n(O)=0$ for any operator.

\end{proof}

Consider now the module $\BC_n$, with even $n$. We can use the basis $\{Q^0_n, Q^1_n\}$ instead of~\eqref{sup-eq:Q-parity-basis}, since ${\rm span}\{Q^0_n, Q^1_n\}={\rm span}\{Q^{++}_n, Q^{--}_n\}$, and $Q^{+-}_n=Q^{-+}_n=0$.  Then, using the results of the previous subsection (see Eq.~\eqref{sup-eq:theo-1}) we get a contribution to the variance from
\begin{equation}
    \begin{split}
        \Tr[O^{\otimes 2} Q_n^0] \Tr[\rho^{\otimes 2} Q_n^0]+ \Tr[O^{\otimes 2} Q_n^1] \Tr[\rho^{\otimes 2} Q_n^1]=\frac{\mathcal{P}_n(O)\mathcal{P}_n(\rho)+\mathcal{C}_n(O)\mathcal{C}_n(\rho)}{\dim(\BC_n)}\,.
    \end{split}
\end{equation}
On the other hand, let us write
$O_n=O_n^{e}+O_n^{o}$ with $\Tr[O_n^{e}O_n^{o}]=0$.
It follows that $\mathcal{P}_n(O)=\mathcal{P}_e(O)+\mathcal{P}_o(O)$ and $\mathcal{C}_n(O)=\mathcal{P}_e(O)-\mathcal{P}_o(O)$ from the orthogonality between parity sectors and the fact that $PO^{p}_n=(-1)^{\lambda_p}O^{p}_n$ (we recall that $\mathcal{C}_n(O)=\langle PO_n,O_n\rangle$). Then,\begin{equation} \label{sup-eq:Bn-contribution}
     \Tr[O^{\otimes 2} Q_n^0] \Tr[\rho^{\otimes 2} Q_n^0]+ \Tr[O^{\otimes 2} Q_n^1] \Tr[\rho^{\otimes 2} Q_n^1]=\frac{\mathcal{P}_e(O)\mathcal{P}_e(\rho)+\mathcal{P}_o(O)\mathcal{P}_o(\rho)}{(\dim(\BC_n)/2)}=\frac{\mathcal{P}_e(O)\mathcal{P}_e(\rho)}{\dim(\BC^e_n)}+\frac{\mathcal{P}_o(O)\mathcal{P}_o(\rho)}{\dim(\BC^o_n)}\,.
\end{equation}
We can finally expand the variance as
\begin{equation}
\begin{split}
        \Var_{\thv}[\ell_{\thv}(\rho,O)]&=\sum_{q}\left(\Tr[O^{\otimes 2}Q_{2q}^{++}]\Tr[\rho^{\otimes 2}Q_{2q}^{++}]+\Tr[O^{\otimes 2}Q_{2q}^{--}]\Tr[\rho^{\otimes 2}Q_{2q}^{--}]\right)\\
        &+\sum_{q} \left(\Tr[O^{\otimes 2}Q_{2q}^{+-}]\Tr[\rho^{\otimes 2}Q_{2q}^{+-}]+\Tr[O^{\otimes 2}Q_{2q}^{-+}]\Tr[\rho^{\otimes 2}Q_{2q}^{-+}]\right)
        \\ &+
        \sum_{l} \left(\Tr[O^{\otimes 2}Q_{2l+1}^0]\Tr[\rho^{\otimes 2}Q_{2l+1}^0]+\Tr[O^{\otimes 2}Q_{2l+1}^1]\Tr[\rho^{\otimes 2}Q_{2l+1}^1]\right)-\mathbb{E}_{\thv}[\ell_{\thv}(\rho,O)]^2
        \,,
\end{split}
\end{equation}
where the indices $q$ and $l$ run as in the module decomposition~\eqref{sup-eq:parity-modules}.
The odd modules give the same contributions as in Theorem~\ref{sup-th:theorem1}. Furthermore, Supplemental Lemma \ref{suplemma:newbasisoverlap} relates the two terms in the first sum with generalized purities in the sectors of even and odd  parity respectively (except for the case in which $q=\lfloor n/2 \rfloor$ and $n$ is even, for which we use~\eqref{sup-eq:Bn-contribution}). Similarly, the remaining contributions yield the generalized coherences. Finally, the mean value squared cancels out the contribution coming from $q=0$, according to Supplemental Lemma \ref{suplemma:meannewbasis}. This results in the closed-form expression
\begin{equation}
          \Var_{\thv}[\ell_{\thv}(\rho,O)]=\sum_{q=1}^{\lfloor \frac{n}{2}\rfloor} \sum_{p=e,o}\frac{\PC_{2q,p}(\rho)\PC_{2q,p}(O)+(1-\delta_{2q,n})\CC_{2q,p}(\rho)\CC_{2q,p}(O)}{\dim(\BC_{2q}^p)}
          +\sum_{l=0}^{n-1}\frac{\PC_{2l+1}(\rho)\PC_{2l+1}(O)+\CC_{2l+1}(\rho)\CC_{2l+1}(O)}{\dim(\BC_{2l+1})}\,.
\end{equation}

The last step of the proof is to consider fermionic states or measurement operators. These have definite parity, which means that they only have support on either $\{\BC_{2q}^e\}_q$ or $\{\BC_{2q}^o\}_q$. It is thus clear that the contributions from the odd modules vanish. Moreover the generalized coherences also vanish by definition. Altogether, this completes the proof of Supplemental Theorem~\ref{sup-th:parity-basis}.

\section{Generalized entanglement and Fermionic Entanglement Entropies}\label{sec:ferm-ent}

In this section we provide operational meaning to the generalized entanglement arising from the Lie algebra of parametrized matchgate circuits in terms of known fermionic measures. This result is in itself remarkable, as we recall that  we are working in a framework where the underlying physical system is composed of (distinguishable) qubits on a quantum computer, and not of actual fermions. However, the fact that the Hilbert space of $n$-qubits $\HC=(\mathbb{C}^2)^{\otimes n}$ is isomorphic to the Fock space $\FC_{-}(\mathbb{C}^n)$ of local spinless fermions on $n$ sites~\cite{bratteli2012operator,de2009x}, along with the fact that the circuit's Lie algebra is isomorphic to the  free fermion algebra, leads to this notable connection between generalized entanglement and the fermionic entanglement measures introduced in~\cite{gigena2015entanglement, gigena2020one, gigena2021many}.

In what follows we will assume that $\rho$ is a fermionic state. Let us notice first that we can write the $\kappa$-purities as 
\begin{equation}\label{eq:kpuritytensor}
\mathcal{P}_\kappa(\rho)=\frac{1}{d}\sum_{\nu_1<\nu_2<\dots < \nu_\kappa} \Tr[\rho i^{\lfloor\frac{ \kappa }{2}\rfloor}c_{\nu_1}c_{\nu_2}\dots c_{\nu_\kappa}]^2=\frac{1}{d}\sum_{\nu_1<\nu_2<\dots < \nu_\kappa} C_{\nu_1\nu_2\dots \nu_\kappa}^2\,,
\end{equation}
where we recall that $d=2^n$, and where have defined the tensor of contractions
\begin{equation}
 C_{\nu_1\nu_2\dots \nu_\kappa}:=   \Tr[\rho i^{\lfloor\frac{ \kappa }{2}\rfloor}c_{\nu_1}c_{\nu_2}\dots c_{\nu_\kappa}]\,,\;\;\;\; \nu_1<\nu_2<\dots < \nu_\kappa\,,
\end{equation}
which is antisymmetric in all its indices. 
Let us focus now in the case $\kappa=2$, i.e., on the so-called $\g$-purity~\cite{ragone2023unified}. We can express $C$ as
\begin{equation}
   C=i (\langle \bm{c}\bm{c}^t \rangle-\id)\,,
\end{equation}
for $\bm{c}=(c_1,c_2,c_3,\dots)^t$. Then, we can write 
\begin{equation}
    \mathcal{P}_\mf{g}(\rho)=\frac{1}{2d}\sum_{\nu_1,\nu_2}C_{\nu_1\nu_2}C_{\nu_1\nu_2}=-\frac{1}{2d}\sum_{\nu_1,\nu_2}C_{\nu_1\nu_2}C_{\nu_2\nu_1}=-\frac{1}{2d}\Tr[C^2]\,.
\end{equation}

In order to relate the $\g$-purity with  fermionic entanglement measures, we
introduce Dirac fermionic operators
from Majoranas as usual:
\begin{equation}\label{dicar-ops}
    d_i:=\frac{c_{2i-1}+i c_{2i}}{2}\,.
\end{equation}
One can readily verify that the following anti-commutation relations hold
\begin{align}\label{eq:anticomm-JW}
    &\{d_j,d_{j'}\}=\vec{0}\,,\quad \{d_j\ad,d_{j'}\ad\}=\vec{0}\,,\quad
    \{d_j,d_{j'}\ad\}=\delta_{jj'}\id\,,
\end{align}
for all $1\leq j,j'\leq n$. Here, $\vec{0}$ denotes the $2^n\times 2^n$ all zero matrix.
We can now
 define the Dirac contraction matrix (as opposed to Majorana) as
\begin{equation}\label{eq:diraconctr}
    D=\begin{pmatrix}
        \langle \bm{d}^\dag \bm{d}\rangle & \langle \bm{d} \bm{d}\rangle\\
         \langle \bm{d}^\dag \bm{d}^\dag\rangle &  \langle \bm{d} \bm{d}^\dag\rangle
    \end{pmatrix}=\id-\langle \bm{\gamma}\bm{\gamma}^\dag \rangle\,,
\end{equation}
with $\bm{\gamma}:=(\bm{d}\,,\bm{d}^\dag)^t$ and $\bm{d}=(d_1,d_2,d_3,\dots)^t$. 

In Ref.~\cite{gigena2015entanglement} the authors introduce the entropies of $D$ (therein $D$ is denoted as ``extended'' one-body reduced density matrix) as a measure of fermionic entanglement. For clarity, let us recall that how this notion of entanglement is  defined. First, an entanglement entropy between a single mode and its orthogonal complement is introduced. Then the sum over all modes of the previous entropy is defined. Finally, a minimization of the ensuing quantity, over all free fermion (Bogoliubov) transformations is performed. It was then proven in Ref.~\cite{gigena2015entanglement} that the resulting entanglement entropy is equal to the corresponding entropy of the matrix $D$. Hence, by construction these entropies are invariant under free fermion transformations. We also mention that this measure of entanglement unifies previous approaches based on either mode entanglement or particle entanglement.

Notably, we can relate the $\mf{g}$-purity to the linear entropy of $D$. This can be proven by relating
the matrices $C$ and $D$. We begin by defining   
$\tilde{\bm{\gamma}}=(d_1,d_1^\dag, d_2,d_2^\dag,\dots)^t$ which can be obtained  from $\vec{d}$ as $\tilde{\bm{\gamma}}=P\bm{\gamma}$ with  $P$ being  a permutation matrix. Then, we can easily see from Eq.~\eqref{dicar-ops} that
\begin{equation}
    \tilde{\bm{\gamma}}=\frac{1}{\sqrt{2}}\Omega \bm{c}, \quad \text{where} \quad \Omega=\bigoplus_{i=1}^n \frac{1}{\sqrt{2}}\begin{pmatrix}
        1& i \\
        1& -i
    \end{pmatrix}\,,
\end{equation}
with $\Omega^\dag \Omega=\id$, 
implying
$
    \bm{\gamma}=\frac{1}{\sqrt{2}}P\Omega \bm{c}\,.
$
From this we obtain the relation
\begin{equation}\label{eq:CDrelation}
  C=i\Omega^\dag P(1-2D) P \Omega\,, 
\end{equation}
which leads directly to the following Lemma.

\begin{suplemma}\label{lem:fermion-ent}
    Consider the fermionic entanglement measure $S_2(\rho):=2 \Tr[D(1-D)]$, i.e,  the linear entropy of $D$. 
    The following relation holds:
    \begin{equation}
     \mathcal{P}_\mf{g}(\rho)=\frac{2 \Tr[D^2]-n}{d}=  \frac{n-S_2(\rho)}{d}\,. 
    \end{equation}
   Moreover, if $O\in i\mf{g}$ and $\Tr[O^2]=2^n$ (i.e., $O$ is a Pauli operator), we have
    \begin{equation}
        \Var_{\thv}[\ell_{\thv}(\rho,O)]= \frac{n-S_2(\rho)}{\dim(\mf{g})}\,.
    \end{equation}
\end{suplemma}
\begin{proof}
    The proof follows by noting that $ \mathcal{P}_\mf{g}(\rho)=-\frac{1}{2d}\Tr[C^2]=\frac{1}{2d}\Tr[(\id-D)(\id-D)]=\frac{2 \Tr[D^2]-n}{d}$, where we used $\Tr[\id]=2n$, $\Tr[D]=n$ and the relation \eqref{eq:CDrelation}.
\end{proof}

This lemma has several important implications. First, it establishes a rigorous connection between a generalized entanglement measure, and 
$S_2(\rho)$, a measure of fermionic entanglement. Second, it allows us to directly express the variance as a function of $S_2(\rho)$. This shows that the maximum variance is achieved for (fermionic) unentangled states, i.e., Slater determinants. Conversely, states which are maximally entangled have have no component in the algebra (by studying the eigenvalues of $D$ one finds that the maximum value of $S_2(\rho)$ is $n$~\cite{gigena2015entanglement}). Interestingly, Supplemental Lemma~\ref{lem:fermion-ent} shows that in order for the variance to be large, the quantum resources measured by fermionic entanglement \cite{gigena2020one} must be small. Crucially, this realization begs the question: \textit{If the amount of resource is small (and the model is trainable), is it also classical simulable?} In Section~\ref{sec:analitics} we explore this question as well as the subtle connection between entanglement, trainability, and simulability.

While we have successfully related variances to an operationally meaningful notion of fermionic entanglement in many important scenarios (those for $O\in \mf{g}$), the question of whether 
all variances are related to some notion of fermionic entanglement naturally arises. In other words, we can ask if for $\kappa\neq 2$
the $\kappa$-purities are also related to some notion of fermionic entanglement. The first thing to notice is that according to Eq.~\eqref{eq:kpuritytensor} contractions of higher order are needed. This means that $\mathcal{P}_\kappa(\rho)$ is necessarily a function of fermionic density matrices of bodyness up to $\kappa/2$ (assuming $\kappa$ even, since for fermionic states the odd contractions vanish). Thus it is clear that an interpretation of $\mathcal{P}_\kappa(\rho)$ in terms of fermionic entanglement requires the latter to be defined via general reduced density matrices. Interestingly, partial extensions of the previous fermionic entanglement ideas have been recently presented \cite{gigena2021many} by precisely using  higher order reduced density matrices. In this case the number of fermions nonetheless has to be fixed, something which is not required for defining 
the purities. In this sense, the notion of generalized entanglement we presented might shed additional light on the problem of defining entanglement in fermionic systems.

\section{Analytical evaluation of variances for particular families of states}\label{sec:analitics}

\subsection{Gaussian states}
A pure Gaussian state can be defined as any state of the form $|\psi\rangle=U|0\rangle$ with $U\in e^{\mf{g}}$ for the even parity sector and as $|\psi\rangle=Uc_1|0\rangle$  the odd parity sector. As we know, the $\kappa$-purities are invariant under the action of $U$, so we just need to compute the purities of either $|0\rangle$ or $c_1|0\rangle$ to characterize all these states.

Let us consider first the even sector. We can write for $|\psi\rangle=|0\rangle^{\otimes n}$, and thus
\begin{equation}
    \rho=\left(|0\rangle\langle 0|\right)^{\otimes n}=\left(\frac{\id+Z}{2}\right)^{\otimes n}=\frac{1}{d}\left(\id+\sum_i Z_i+\sum_{i<j}Z_i Z_j+\sum_{i<j<k}Z_i Z_jZ_k+\dots\right)\,.
\end{equation}
Notice that a term containing $\kappa/2$ $Z$ operators ($\kappa$ even) belongs to $\BC_\kappa$: $Z_i=-i c_{2i-1}c_{2i}$ implying that each of these terms is an ordered product of $\kappa$ Majoranas. Explicitly $\sum_{i_1<i_2<\dots i_{\kappa/2}}Z_{i_1}Z_{i_2}\dots Z_{i_{\kappa/2}}=\prod_{j=1}^\kappa \sum'_{i_j}Z_{i_j}=(-i)^{\kappa/2}\prod_{j=1}^\kappa \sum'_{i_j}c_{2i_j-1}c_{2i_j}$ with the prime in the sum indicating the ordering that $i_1<i_2<\dots i_{\kappa/2}$.
Hence, with this expansion, it is easy to compute the purities:
\begin{equation}
    \mathcal{P}_\kappa(\rho)= \sum_{l=1}^{\dim(\BC_\kappa)}\Tr[\frac{(-i)^{\frac{\kappa}{2}}}{d}\prod_{j=1}^\kappa \sum_{i_j}c_{2i_j-1}c_{2i_j} \frac{(i)^{ \frac{\kappa}{2}}c^{\vec{s}^l}}{\sqrt{d}}]^2=\frac{1}{d}\sum'_l=\frac{1}{d}\binom{n}{\kappa/2}\,,
\end{equation}
where we have denoted as $\sum'_l$ the sum over those $l$ such that $\vec{s}_\kappa^l$ matches one of the $\prod_j c_{2i_j-1}c_{2i_j}$ terms. This is essentially just counting how many of the these terms there are, which can be easily found to be $\binom{n}{\kappa/2}$, as it corresponds to choosing $\kappa/2$ positions for the $n$ $Z$-operators.

Consider now the odd sector. Each term $\prod_{j=1}^\kappa \sum_{i_j}c_{2i_j-1}c_{2i_j}$ is now replaced by $\prod_{j=1}^\kappa \sum_{i_j}c_1c_{2i_j-1}c_{2i_j}c_1=\pm \prod_{j=1}^\kappa \sum_{i_j}c_{2i_j-1}c_{2i_j}$ where we used the fact that both these terms and the operator $c_1$ are Paulis (so they either commute or anti-commute), and that $c_1^2=\id$. These signs do not contribute to the value of $\mathcal{P}_\kappa(\rho)$, which is then the same as before. 

On the other hand, it is clear that $P|\psi\rangle= \pm |\psi\rangle$. This means that $\CC_\kappa(\rho)=\mathcal{P}_\kappa(\rho)$.
The final form of the variance for $\rho$ a pure Gaussian state is then
\begin{equation}
    \Var_{\thv}[\ell_{\thv}(\rho,O)]=  \sum_{\kappa'=1}^{n-1} \frac{1}{d} \binom{n}{\kappa'}\binom{2n}{2\kappa'}^{-1} \big(\mathcal{P}_{\kappa'}(O)+\CC_{\kappa'}(O)\big)\,.
\end{equation}

We now focus on the case $O\in \BC_\kappa$, which corresponds to one of the main body examples. Since the Majorana operators are Pauli strings, their product is as well. This means that if we want to consider normalized operators (having eigenvalues $\pm 1$) they necessarily have the form $ O=i^{\lfloor \frac{\kappa}{2}\rfloor}U^\dag c_{\nu_1}\cdots c_{\nu_\kappa} U$. As we said before, we can omit $U$ when computing a purity. Thus we have
\begin{equation}
\mathcal{P}_\kappa(O)=\sum_{l=1}^{\dim(\BC_\kappa)}\Tr[i^{\lfloor \frac{\kappa}{2}\rfloor}c_{\nu_1}\cdots c_{\nu_\kappa} \frac{(i)^{\lfloor \frac{ \kappa}{2}\rfloor}c^{\vec{s}^l}}{\sqrt{d}}]^2=d\,.
\end{equation}
Clearly $\CC_\kappa(O)=0$ which leads us to 
\begin{equation}
    \Var_{\thv}[\ell_{\thv}(\rho,O)]= \binom{n}{\kappa/2}\binom{2n}{\kappa}^{-1} \,.
\end{equation}

\subsection{Non-fermionic states and operators}
Consider the non-fermionic state $|\psi\rangle=\alpha |0\rangle+\beta c_1|0\rangle$ which is a superposition of states with different parities, and observables of the form $O=X_j$. Let us also assume $\alpha,\beta\in \mathbb{R}$. Notice first that $X_j=\prod_{i=1}^{2j-1}c_{i}\in \BC_{2j-1}$ so that $\mathcal{P}_\kappa(O)=\delta_{\kappa,2j-1}d$. On the other hand, in other to compute the purities of the state let us notice that for $k$ odd
\begin{equation}
\mathcal{P}_\kappa(|\psi\rangle\langle\psi|)=\mathcal{P}_\kappa(\alpha^2|0\rangle\langle 0|+\beta c_1 |0\rangle\langle 0|c_1+\alpha\beta(c_1|0\rangle\langle 0|+|0\rangle\langle 0|c_1))=\alpha^2\beta^2\mathcal{P}_\kappa(\{c_1,|0\rangle\langle 0|\})\,,
\end{equation}
where we used that the first two terms only include an even number of Majoranas, and where $\{\cdot,\cdot\}$ denotes the anti-commutator. Since $c_1=X_1$ and Pauli operators either commute or anti-commute one finds
\begin{equation}
    \{c_1,|0\rangle\langle 0|\}=\frac{2}{d}X_1\left(\id+\sum_{i>1} Z_i+\sum_{i>1,i<j}Z_i Z_j+\sum_{i>1,i<j<k}Z_i Z_jZ_k+\dots\right)\,.
\end{equation}
We can now compute the parity as in the previous Gaussian case taking into account that we can choose from $n-1$ $Z$ operators and that we need $j-1$ of them for the $2j-1$-purity. Thus we get (the covariance of the observables is null):
\begin{equation}
     \Var_{\thv}[\ell_{\thv}(\rho,O)]=  4\alpha^2\beta^2 \binom{n-1}{j-1}\binom{2n}{2j-1}^{-1}\,.
\end{equation}

\subsection{Magic states}\label{sec:magicst}

We are going to study the case of tensor products of $4$-qubit magic states of the form
\begin{equation}
|\psi(\tau)\rangle=\frac{|0000\rangle+|0011\rangle+|1100\rangle+e^{i\tau}|1111\rangle}{2}\,.
\end{equation}
Namely, the complete state has the form $|\Psi(\tau)\rangle=|\psi(\tau)\rangle^{\otimes \frac{n}{4}}$. These states have been employed recently to classically simulate quantum algorithms via free fermion techniques plus magic \cite{reardon2023improved}. It was therein conjecture, and later proven in\cite{dias2023classical,cudby2023gaussian}, that the complexity grows exponentially with the magic. It can be easily seen that for $\tau=0$ the state is Gaussian, while for $\tau=\pi$ we achieve maximum magic.

Consider first a single copy $|\psi(\tau)\rangle$. We are going to characterize the purities via the fermionic entanglement techniques introduced in \cite{gigena2015entanglement} and used in  the previous section. Therein it was shown that for $n=4$ fermions the matrix $D$ (defined in \eqref{eq:diraconctr}) can be diagonalized analytically leading to a normal form
\begin{equation}\label{eq:normalform}
|\psi(\tau)\rangle=\alpha(\tau)|0000\rangle'+\beta(\tau)|1111\rangle'\,,
\end{equation}
via a Bogoliubov transformation. Here $|0000\rangle'$ is the vacuum of new fermions operators related to the previous ones linearly. It also holds that $|\alpha(\tau)|=\sqrt{f_+}$, $|\beta(\tau)|=\sqrt{f_{-}}$ for $f_{\pm}$ the $4$-degenerate eigenvalues of $D$, which can be written as $f_{\pm}=\frac{1\pm \sqrt{1-\frac{S_2(D)}{4}}}{2}$. For these states one can easily find $f_{\pm}=\frac{1\pm \cos(\tau/2)}{2}$ leading to $S(|\psi(\tau)\rangle \langle \psi(\tau)|)=\sin^2(\tau/2)$, or equivalently 
\begin{equation}
    \mathcal{P}_2(|\psi(\tau)\rangle \langle \psi(\tau)|)=\frac{1}{4}\cos^2(\tau/2)\,,
\end{equation}
and $|\alpha(\tau)|=\cos(\tau/4)$, $|\beta(\tau)|=\sin(\tau/4)$. This confirms that the Gaussian states $\tau=0$ have no fermionic entanglement, while $\tau=\pi$ is the maximum magical and entangled case.

Let us now consider two copies of the previous state, and compute the purity in the basis of \eqref{eq:normalform}. We have three kinds of contribution to the purity, terms involving operators acting on the first copy, terms 
involving operators acting on the second copy, and crossed terms arising from traces of operators in the algebra which act on both. The first two contributions yield twice $\frac{2^4}{2^8}\mathcal{P}_2(|\psi(\tau)\rangle\langle \psi(\tau)|)$ (the pre-factor must be included because the Hilbert space dimension which normalizes the Majorana products is now $d=2^8$). One can easily show that the crossed terms vanish: the algebra elements acting on both copies are of the form $\widehat{S_iS_j}$ with $i=1,2,3,4$, $j=5,6,7,8$, $S=X,Y$. Both $X$ and $Y$ induce a bit flip in the computational basis, meaning that $\langle\widehat{S_iS_j}\rangle=0$, as can be seen explicitly by using \eqref{eq:normalform}: $\langle aaaabbbb|\widehat{S_iS_j}|ccccdddd\rangle=\delta_{ac}\delta_{bd}\langle aaaabbbb|\widehat{S_iS_j}|aaaabbbb\rangle=0$ (here $a,b,c,d=0,1$ representing the possible contributions involved in $\langle\widehat{S_iS_j}\rangle$). Thus we obtain $\mathcal{P}_2(|\Psi(\tau)\rangle\langle \Psi(\tau)|)=\frac{8}{d}\cos^2(\tau/2)$ for $n/4=2$ copies. By repeating the previous argument one easily obtains
\begin{equation}
    \mathcal{P}_2(|\Psi(\tau)\rangle\langle \Psi(\tau)|)=\frac{n}{d}\cos^2(\tau/2)\,,
\end{equation}
implying
\begin{equation}\label{eq:o-in-bk-lol}
   \Var_{\thv}[\ell_{\thv}(\rho,O)]=  
\frac{n\cos^2(\tau/2)}{\dim(\mf{g})}=  
\frac{\cos^2(\tau/2)}{2n-1} \,.
\end{equation}

Let us also discuss the extent associated with these states: we recall that the fermionic extent is defined as the minimum $\xi=(\sum_i |a_i|)^2$ over all expansions of a given state $|\psi\rangle$ in terms of Gaussian states, namely $|\psi\rangle=\sum_i a_i |\varphi_i\rangle$ with $|\varphi_i\rangle$ Gaussian. In other words, the fermionic extent quantifies how many terms one needs to expand an arbitrary state as a linear combination of Gaussian states. While mean values over Gaussian states are efficiently computable via Pffafians, a mean value over a state with exponential extent in the number of qubits cannot be efficiently computed with this method. 
It is interesting to notice that for a single copy the previous decomposition gives an upper bound to the extent $\xi=(|\cos(\tau/4)|+|\sin(\tau/2)|)^2=1+\sin(\tau/2)$ which coincides with the actual value found in \cite{reardon2023improved}. On the other hand, it has been recently proven that for $n=4$ qubits this quantity is multiplicative \cite{dias2023classical,cudby2023gaussian} thus yielding
\begin{equation}
    \xi=[1+\sin(\tau/2)]^{n/4}\,,
\end{equation}
for the state $|\Psi(\tau)\rangle$. Remarkably, this model is trainable even for states with an exponential extent. This does not render the model non-classically-simulable as we explain in the next section. It is also interesting to compare $\xi$ to the entanglement entropy $S_2(\rho)=n \sin^2(\tau/2)$, which scales linearly with the number of copies for fixed $\tau$.

\section{Remarks on classical simulability of trainable models based on matchgate circuits}

In the main body we raised the following question: \textit{Does the inherent structure that precludes the presence of BPs in  a matchgate-based model (a requisite for trainability) simultaneously render it classically simulable?} 

To analyze the previous question, we need to break the problem apart into pieces. First, let us discuss classical simulation methods for matchgate circuits. The most common simulation technique for fermionic systems is based on Wick's theorem~\cite{bravyi2004lagrangian} which allows us to express the expectation value of a general observable over some Gaussian state in terms of the expectation value of the same state but only over elements in the algebra. Hence, this method is extremely well suited for cases when $\rho$ is a Gaussian state and $O$ is arbitrary. For non-Gaussian (entangled) states Wick-based methods start to become more expensive. In fact, recent advancements ~\cite{dias2023classical,cudby2023gaussian} have shown that the cost of simulation grows with the amount of ``extent'' (or magic) in the state. This is due to the fact that the extent is related to a decomposition of the initial state into a sum of Gaussian states. A second important simulation technique is called $\mf{g}$-sim~\cite{somma2005quantum,goh2023lie}, which is extremely well suited for measurement operators in $\BC_\kappa$ with $\kappa$ not scaling with $n$. The key idea here is to work in the Heisenberg picture and backwards propagate the observable through the circuit, leveraging knowledge of the  Lie algebra's structure constants. Importantly, we recall that  for $\mf{g}$-sim to work, one needs to be able to estimate the expectation value of the initial state over the Lie algebra basis elements. 

As we will discuss now, the cases where the loss function does not exhibit a BP, appear to be precisely the situations where the loss function can be efficiently classically estimated. First, let us recall from the main text that a necessary condition for trainability is to have small generalized globality (as otherwise the term $\dim(\BC_\kappa)$ in the denominator of the variance will lead to an exponential concentration). This already restricts the settings where BPs can be avoided. Regarding the initial state, we know that if it Gaussian, then Wick-based methods become sufficient. This means that some form of fermionic entanglement \cite{gigena2015entanglement} is necessary and that one should instead use $\g$-sim (which, as previously mentioned, is meant to be used with measurement of small generalized globality). To guarantee that Wick-based simulations are not possible, let us focus on the state $\ket{\Psi(\tau)}$ presented in the main text. We will also assume for simplicity $O\in i\g$. We know that for $\tau\in(0,\pi]$ and $\tau-\pi\in \Omega(1/\poly(n))$, then the state  has an  exponentially large amount of extent, but the variance is not exponentially vanishing since this is also a low entanglement state (notice that this exponential separation, proven in Section~\ref{sec:magicst}, is a remarkable result which clearly shows that these two quantities measure different types of resources). It should be clear that in this case, one can also trivially compute the expectation value of all the basis elements of the Lie algebra, thus enabling the use of $\g$-sim.

While our previous discussions are by no means complete or thorough, they point to the fact that the absence of BP requires both low generalized globality and entanglement, the first condition enabling the use of the  $\g$-sim method. If having low fermionic entanglement always allows for an efficient classical method that computes the expectation values of the Lie algebra basis elements, there seems to be no room for quantum advantage. 
One could then argue that the same structure which makes a model ``trainable'' renders it simulable. While we pose this as an open research direction, it is worth highlighting that more fundamental work will be required. The theory of entanglement in fermionic systems is not developed far enough to fully characterize variances beyond the DLA case. Moreover, the previous argument also ignores the presence of coherence effects between modules, which we have seen can increase the variances. 

While our work provides a fundamental step in the rigorous search for quantum advantage, highlighting the physical elements that enable quantum trainability in parameterized matchgate circuits,  it is not clear at the moment how these considerations apply to other circuits. One could then  conjecture that in general trainability implies some form of simulability (perhaps even for shallow circuits). Here we can even envision that the first step is to show that those few models for which trainability guarantees have been proven are within the grasp of classical simulations. Conversely, rigorously proving classical hardness would discard the conjecture, thus opening the path for guarantees of variational quantum advantage.

\end{document}